\newif\ifextended\extendedtrue
\newcommand{\extended}[2]{\ifextended#1\else#2\fi}

\PassOptionsToPackage{final}{microtype} 
\ifextended%
  \documentclass[acmsmall,screen,authorversion]{acmart}
  \authorsaddresses{}
\else%
  \documentclass[acmsmall,screen]{acmart}
\fi
\usepackage{pervasives,dextersMacros}
\allowdisplaybreaks[3]

\setcopyright{acmlicensed}
\acmPrice{}
\acmDOI{10.1145/3371129}
\acmYear{2020}
\copyrightyear{2020}
\acmJournal{PACMPL}
\acmVolume{4}
\acmNumber{POPL}
\extended{}{\acmArticle{61}}
\acmMonth{1}



\begin{document}

\title{Guarded Kleene Algebra with Tests}
\subtitle{Verification of Uninterpreted Programs in Nearly Linear Time}
\ifextended%
  \titlenote{Extended version with appendix.}
  \def\authornotenum{2}
\else%
  \def\authornotenum{1}
\fi

\author{Steffen Smolka}
\affiliation{%
  \institution{Cornell University}%
  \city{Ithaca, NY}
  \country{USA}%
}
\email{smolka@cs.cornell.edu}

\author{Nate Foster}
\affiliation{%
  \institution{Cornell University}%
  \city{Ithaca, NY}
  \country{USA}%
}
\email{jnfoster@cs.cornell.edu}

\author{Justin Hsu}
\affiliation{%
  \institution{University of Wisconsin--Madison}
  \city{Madison, WI}
  \country{USA}%
}
\email{email@justinh.su}

\author{Tobias Kapp\'{e}}
\affiliation{%
  \institution{University College London}%
  \city{London}
  \country{UK}%
}
\email{tkappe@cs.ucl.ac.uk}

\author{Dexter Kozen}
\affiliation{%
  \institution{Cornell University}%
  \city{Ithaca, NY}
  \country{USA}%
}
\email{kozen@cs.cornell.edu}

\author{Alexandra Silva}
\affiliation{%
  \institution{University College London}%
  \city{London}
  \country{UK}%
}
\email{alexandra.silva@ucl.ac.uk}

\extended{\authorsaddresses{}}{}

\begin{abstract}
Guarded Kleene Algebra with Tests (GKAT) is
a variation on Kleene Algebra with Tests (KAT) that arises by restricting
the union ($+$) and iteration ($*$) operations from KAT to predicate-guarded
versions. We develop the (co)algebraic theory of GKAT and show how it can be 
efficiently used to reason about imperative programs. In contrast to KAT,
whose equational theory is PSPACE-complete, we show that the equational theory
of GKAT is (almost) linear time. We also provide a full Kleene theorem and
prove completeness for an analogue of Salomaa's axiomatization of Kleene Algebra.
\end{abstract}

\maketitle

\begin{quote}
This paper is dedicated to Laurie J. Hendren (1958--2019), whose passion for teaching and research inspired us and many others. Laurie's early work on McCAT~\cite{ErosaH94} helped us understand the limits of Guarded Kleene Algebra with Tests and devise a suitable definition of \emph{well-nestedness} that underpins our main results.
\end{quote}

\section{Introduction}%
\label{sec:intro}

Computer scientists have long explored the connections between
families of programming languages and abstract machines.
This dual perspective has furnished deep
theoretical insights as well as practical tools. As an example,
Kleene's classic result establishing the equivalence of regular
expressions and finite automata~\cite{kleene-1956} inspired decades of work across a
variety of areas including programming language design, mathematical
semantics, and formal verification.

Kleene Algebra with Tests (KAT)~\citep{K96b}, which combines Kleene Algebra (KA)
with Boolean Algebra (BA), is a modern example of this approach.
Viewed from the program-centric perspective, a KAT
models the fundamental constructs that arise in programs:
sequencing, branching, iteration, etc. The equational
theory of KAT enables algebraic reasoning and can be finitely
axiomatized~\citep{KS96a}. Viewed from the machine-centric perspective, a KAT
describes a kind of automaton that
generates a regular language of traces. This shift in perspective
admits techniques from coalgebra for reasoning about
program behavior. In particular, there are efficient algorithms for
checking bisimulation, which can be optimized using properties of
bisimulations~\citep{HopcroftKarp71,bonchi-pous-2013} or symbolic automata
representations~\citep{pous-2014}.

KAT has been used to model computation across a wide variety of areas
including program transformations~\citep{AK01a,K97c}, concurrency
control~\citep{Co94d}, compiler optimizations~\citep{KP00}, cache
control~\citep{BK02a,Co94c}, and more~\citep{Co94c}.  A prominent recent
application is NetKAT~\citep{AFGJKSW13a}, a language for reasoning about the
packet-forwarding behavior of software-defined
networks. NetKAT has a sound and complete
equational theory, and a coalgebraic decision procedure that can be
used to automatically verify many important networking properties
including reachability, loop-freedom, and isolation~\citep{FKMST15a}. However, while
NetKAT's implementation scales well in practice,
deciding equivalence for NetKAT is PSPACE-complete in the worst case~\citep{AFGJKSW13a}.

A natural question to ask is whether there is an efficient fragment of
KAT that is reasonably expressive, while retaining a solid foundation. We
answer this question positively with a comprehensive study of Guarded
Kleene Algebra with Tests (GKAT), the guarded fragment of KAT\@. GKAT is
a propositional abstraction of imperative while programs.
We establish the fundamental properties of GKAT and develop its algebraic and coalgebraic theory.
GKAT replaces the union $(e + f)$ and iteration $(e\star)$ constructs
in KAT with guarded versions: conditionals $(e +_b f)$ and loops
$(e^{(b)})$ guarded by Boolean predicates $b$. The resulting language is
a restriction of full KAT, but sufficiently expressive to
model typical, imperative programs---\eg, essentially all NetKAT
programs needed to solve practical verification problems can be
expressed as guarded programs.

In exchange for a modest sacrifice in expressiveness, GKAT offers two
significant advantages. First, program equivalence (for a fixed Boolean algebra)
is decidable in \emph{nearly linear time}---a substantial improvement over the
PSPACE complexity for KAT~\cite{CKS96a}. Specifically, any GKAT
expression $e$ 
can be represented as a deterministic automaton of size
$\OO(|e|)$, while KAT expressions can require as many as $\OO(2^{|e|})$
states. As a consequence, any language property that is efficiently decidable
for deterministic automata is also efficiently decidable for GKAT\@.
Second, we
believe that GKAT is a better foundation for probabilistic languages
due to well-known issues that arise when combining
non-determinism---which is native to KAT---with probabilistic
choice~\cite{VaraccaWinskel06,MISLOVE2006261}. For example,
ProbNetKAT~\cite{probnetkat-cantor}, a 
probabilistic extension of NetKAT, does not satisfy the KAT axioms,
but its guarded restriction forms a proper GKAT\@.

Although GKAT is a simple restriction of KAT at the syntactic level,
its semantics is surprisingly subtle. In particular,
the ``obvious'' notion of GKAT automata can encode behaviors
that would require non-local control-flow operators
(e.g, $\kw{goto}$ or multi-level $\kw{break}$ statements)~\cite{KT08a}. In contrast, GKAT models programs whose control-flow
always follows a lexical, nested structure. To overcome this
discrepancy, we identify restrictions on
automata to enable an analogue of Kleene's theorem---every
GKAT automaton satisfying our restrictions can be converted to a program, and vice versa. Besides
the theoretical interest in this result, we believe it may
also have practical applications, such as reasoning about
optimizations in a compiler~\cite{mccat}. We also develop an
equational axiomatization for GKAT and prove that it is sound and
complete over a coequationally-defined language model.
The main challenge is that without $+$,
the natural order on KAT programs can no longer be used to
axiomatize a least fixpoint. We instead axiomatize a unique
fixed point, in the style of Salomaa's work on Kleene Algebra~\citep{Salomaa66}.

\subsubsection*{Outline.}
We make the following contributions in this paper.
\begin{itemize}
  \item We initiate a comprehensive study of GKAT, a guarded version of KAT, and show how GKAT models
relational and probabilistic programming languages (\cref{sec:gkat}).
  \item We give a new construction of linear-size automata from GKAT programs (\cref{sec:kleene}).
  As a consequence, the equational theory of GKAT (over a fixed Boolean algebra)
   is decidable in nearly linear time (\cref{sec:decision}).

  \item We identify a class of automata representable as GKAT expressions (\cref{sec:kleene}) that contains all automata produced by the previous construction, yielding a Kleene theorem.

  \item We present axioms for GKAT (\cref{sec:axioms}) and prove that our axiomatization is complete for equivalence with respect to a coequationally-defined language model (\cref{sec:complete}).
\end{itemize}
Omitted proofs appear in the
\extended%
  {appendix.}
  {appendix of the extended version of this paper~\cite{gkat-full}.}

\section{Overview: An Abstract Programming Language}%
\label{sec:gkat}

This section introduces the syntax and semantics of GKAT, an abstract
programming language with uninterpreted actions. Using examples, we
show how GKAT can model relational and probabilistic programming
languages---\ie, by giving actions a concrete
interpretation. An equivalence between abstract GKAT programs thus
implies a corresponding equivalence between concrete programs.




\subsection{Syntax}

The syntax of GKAT is parameterized by abstract sets of \emph{actions}
$\Sigma$ and \emph{primitive tests} $T$, where $\Sigma$ and $T$ are
assumed to be disjoint and nonempty, and $T$ is assumed to be finite.
We reserve $p$ and $q$ to range over actions, and $t$ to range over primitive tests. The language consists of
Boolean expressions, $\Bexp$, and GKAT expressions, $\Exp$, as defined
by the following grammar:\[
  \begin{array}[t]{r@{\;\;}l@{\qquad}l}
  \multicolumn{3}{l}{b,c,d \in \Bexp::=}\\[0.2em]
  \quad \mid  &0              &\kw{false}\\[-0.3mm]
  \mid        &1              &\kw{true}\\[-0.3mm]
  \mid        &t \in T        &t\\[-0.3mm]
  \mid        &b \cdot c      &b~\kw{and}~c\\[-0.3mm]
  \mid        &b + c          &b~\kw{or}~c\\[-0.3mm]
  \mid        &\bneg{b}       &\kw{not}~b
  \end{array}
  \qquad\qquad\qquad
  \begin{array}[t]{r@{\;\,}l@{\qquad}l}
  \multicolumn{3}{l}{e,f,g \in \Exp ::=}\\[0.2em]
  \quad \mid  &p\in \Sigma      &\kw{do}~p\\[-0.3mm]
  \mid        &b \in \Bexp      &\kw{assert}~b\\[-0.3mm]
  \mid        &e \cdot f        &e\kw{;}\,f\\[-0.3mm]
  \mid        &e+_b f           &\kw{if}~b~\kw{then}~e~\kw{else}~f\\[0.3mm]
  \mid        &e^{(b)}          &\kw{while}~b~\kw{do}~e
  \end{array}
\]
The algebraic notation on
the left is more convenient when manipulating terms, while the
notation on the right may be more intuitive when writing programs. We
often abbreviate $e \cdot f$ by $ef$, and omit parentheses following
standard conventions, \eg, writing $bc + d$ instead of $(bc) + d$ and
$ef^{(b)}$ instead of $e(f^{(b)})$.

\subsection{Semantics: Language Model}

Intuitively, we interpret a GKAT expression as the set of ``legal'' execution traces it
induces, where a trace is legal if no assertion fails. To make this formal, let $b
\equiv_\text{BA} c$ denote Boolean equivalence. Entailment is a preorder on the set
of Boolean expressions, $\Bexp$, and can be characterized in terms of equivalence
as follows: $b \leq c \iff b + c \equiv_\text{BA} c$. In the quotient set
$\Bexp/\equiv_\text{BA}$ (the \emph{free Boolean algebra} on generators $T =
\sset{t_1, \dots, t_n}$), entailment is a partial order ${[b]}_{\equiv_\text{BA}} \leq
{[c]}_{\equiv_\text{BA}} \defiff b + c \equiv_\text{BA} c$, with minimum and maximum elements
given by the equivalence classes of $0$ and $1$, respectively.  The minimal
nonzero elements of this order are called \emph{atoms}. We let $\At$ denote the
set of atoms and use lowercase Greek letters $\alpha,\beta, \dots$ to denote
individual atoms.  Each atom is the equivalence class of an expression of the
form $c_1 \cdot c_2 \cdots c_n \in \Bexp$ with $c_i \in \sset{t_i,
\bneg{t_i}}$. Thus we can think of each atom as representing a truth assignment on
$T$, \eg, if $c_i = t_i$ then $t_i$ is set to true, otherwise if $c_i =
\bneg{t_i}$ then $t_i$ is set to false. Likewise, the set $\set{\alpha \in
\At}{\alpha \leq b}$ can be thought of as the set of truth assignments where $b$
evaluates to true; $\equiv_\text{BA}$ is \emph{complete} with respect to this interpretation in that two Boolean expressions are related by $\equiv_\text{BA}$ if and only if their atoms coincide~\cite{birkhoff-bartee-1970}.

A \emph{guarded string} is an element of the regular set
$\Gs \defeq \At \cdot {(\Sigma \cdot \At)}^*$. Intuitively, a non-empty string
$\alpha_0p_1\alpha_1\cdots p_n\alpha_n \in \Gs$ describes a trace of an abstract
program: the atoms $\alpha_i$ describe the state of the system at
various points in time, starting from an initial state $\alpha_0$ and
ending in a final state $\alpha_n$, while the actions $p_i \in \Sigma$ are the
transitions triggered between the various states.
Given two traces, we can combine them sequentially by running one after the other.
Formally, guarded strings compose via a partial \emph{fusion product}
$\diamond \colon \Gs \times \Gs \rightharpoonup \Gs$,
defined for $x,y \in {(\At \cup \Sigma)}^*$ as \[
  x\alpha \diamond \beta y \defeq \begin{cases}
    x \alpha y       &\text{if } \alpha=\beta\\
    \text{undefined} &\text{otherwise.}
  \end{cases}
\]
This product lifts to a total function on languages $L,K \subseteq \Gs$ of
guarded strings, given by
\[
    L \diamond K \defeq \set{x \diamond y}{x \in L, y \in K}.
\]
We need a few more constructions before we can
interpret GKAT expressions as languages representing their possible
traces. First, $2^\Gs$ with the fusion product forms a monoid with identity
$\At$ and so we can define the $n$-th power $L^n$ of a language $L$ inductively
in the usual way:
\begin{mathpar}
L^0 \defeq \At
\and
L^{n+1} \defeq L^n \diamond L
\end{mathpar}
Second, in the special case where $B \subseteq \At$, we write $\bneg{B}$ for $\At - B$
and define:
\begin{mathpar}
L +_B K \defeq (B \diamond L) \cup (\bneg{B} \diamond K)
\and
L^{(B)} \defeq \bigcup_{n \geq 0} {(B \diamond L)}^n \diamond \bneg{B}
\end{mathpar}

We are now ready to interpret GKAT expressions as languages of guarded strings
via the semantic map
$\sem{-} ~\colon~ \Exp \to 2^\Gs$ as follows:
\begin{gather*}
\begin{aligned}[t]
\sem{p} &\defeq \set{\alpha p \beta}{\alpha,\beta \in \At}\\
\sem{b} &\defeq \set{\alpha  \in \At}{\alpha\leq b}
\end{aligned}
\qquad\qquad
\begin{aligned}[t]
\sem{e \cdot f} &\defeq \sem{e} \diamond \sem{f}\\
\sem{e+_b f} &\defeq \sem{e} +_{\sem{b}} \sem{f}
\end{aligned}
\qquad\qquad
\begin{aligned}[t]
\sem{e^{(b)}} &\defeq \sem{e}^{(\sem{b})}
\end{aligned}
\end{gather*}
We call this the \emph{language model} of GKAT\@. Since we
make no assumptions about the semantics of actions, we interpret them
as sets of traces beginning and ending in arbitrary states; this
soundly overapproximates the behavior of any instantiation. A test is
interpreted as the set of states satisfying the test. The traces of
$e \cdot f$ are obtained by composing traces from $e$ with traces from
$f$ in all possible ways that make the final state of an $e$-trace
match up with the initial state of an $f$-trace. The traces of $e +_b
f$ collect traces of $e$ and $f$, restricting to $e$-traces
whose initial state satisfies $b$ and $f$-traces whose initial state
satisfies $\bneg{b}$. The traces of $e^{(b)}$ are obtained by
sequentially composing zero or more $be$-traces and selecting
traces ending in a state satisfying $\bneg b$.

\begin{remark}[Connection to KAT]%
\label{rem:kat-connection}
The expressions for KAT, denoted $\Kexp$, are generated by the same
grammar as for GKAT, except that KAT's union ($+$) replaces
GKAT's guarded union ($+_b$) and KAT's iteration ($e^*$) replaces
GKAT's guarded iteration ($e^{(b)}$).
GKAT's guarded operators can be encoded in KAT;\@ this encoding, which goes back
to early work on Propositional Dynamic Logic~\cite{FL79}, is the standard method
to model conditionals and while loops:
\begin{mathpar}
  e +_b f \mapsto be + \bneg{b}f
  \and
  e^{(b)} \mapsto {(be)}^* \bneg{b}
\end{mathpar}
Thus, there is a homomorphism $\phi\colon \Exp \to \Kexp$ from GKAT
to KAT expressions.
We inherit KAT's language model~\cite{KS96a},
$\KK\den{-}\colon \Kexp \to 2^\Gs$, in the sense that
$\den{-} = \KK\den{-} \circ \phi$.
\end{remark}

\noindent
The languages denoted by GKAT programs satisfy an important property:
\begin{definition}[Determinacy property]%
\label{def:action-det}
A language of guarded strings $L \subseteq \Gs$
satisfies the \emph{determinacy property} if,
whenever string $x,y \in L$ 
agree on their first $n$ atoms, then they agree
on their first $n$ actions (or lack thereof). For example,
$\{\alpha p\gamma, \alpha p \delta, \beta q\delta\}$ and $\{\alpha p\gamma,\beta\}$
for $\alpha\ne\beta$ satisfy the determinacy property, while
$\{\alpha p\beta,\alpha\}$ and $\{\alpha p\beta,\alpha q\delta\}$
for $p\ne q$ do not.
\end{definition}

\noindent
We say that two expressions $e$ and $f$ are \emph{equivalent}
if they have the same semantics---\ie, if $\den{e} = \den{f}$.
In the following sections, we show that this notion of equivalence
\begin{itemize}
\item is sound and complete for relational and  probabilistic interpretations
  (\cref{sec:rel-model,sec:prob-model}),
\item can be finitely and equationally axiomatized in a
  sound (\cref{sec:axioms}) and complete (\cref{sec:complete}) way, and
\item is efficiently decidable in time nearly linear in the sizes of the
  expressions (\cref{sec:kleene,sec:decision}).
\end{itemize}

\subsection{Relational Model}%
\label{sec:rel-model}
This subsection gives an interpretation of GKAT expressions as binary
relations, a common model of input-output behavior for many programming languages.
We show that the language model is sound and complete for this interpretation.
Thus GKAT equivalence implies program equivalence for any programming language
with a suitable relational semantics.

\begin{definition}[Relational Interpretation]
Let $i = (\State, \eval, \sat)$ be a triple consisting of
\begin{itemize}
  \item a set of \emph{states} $\State$,
  \item for each action $p \in \Sigma$,
    a binary relation $\eval(p) \subseteq \State \times \State$, and
  \item for each primitive test $t \in T$,
    a set of states $\sat(t) \subseteq \State$.
\end{itemize}
Then the \emph{relational interpretation} of an expression $e$ with respect
to $i$ is the smallest binary relation
$\rden{i}{e} \subseteq \State \times \State$
satisfying the following rules,
\begin{mathpar}
\inferrule{(\sigma, \sigma') \in \eval(p)}{(\sigma, \sigma') \in \rden{i}{p}}
\and
\inferrule{ \sigma \in \sat^\dagger(b) }{ (\sigma, \sigma) \in \rden{i}{b} }
\and
\inferrule{
 (\sigma, \sigma') \in \rden{i}{e}\\
 (\sigma', \sigma'') \in \rden{i}{f}
}{
 (\sigma, \sigma'') \in \rden{i}{e \cdot f}
}
\and
\inferrule{
 \sigma \in \sat^\dagger(b)\\
 (\sigma, \sigma') \in \rden{i}{e}
}{
 (\sigma, \sigma') \in \rden{i}{e +_b f}
}
\and
\inferrule{
 \sigma \in \sat^\dagger(\bneg{b})\\
 (\sigma, \sigma') \in \rden{i}{f}
}{
 (\sigma, \sigma') \in \rden{i}{e +_b f}
}
\and
\inferrule{
 \sigma \in \sat^\dagger(b)\\
 (\sigma, \sigma') \in \rden{i}{e}\\
 (\sigma', \sigma'') \in \rden{i}{e^{(b)}}\\
}{
 (\sigma, \sigma'') \in \rden{i}{e^{(b)}}
}
\and
\inferrule{
 \sigma \in \sat^\dagger(\bneg{b})
}{
 (\sigma, \sigma) \in \rden{i}{e^{(b)}}
}
\end{mathpar}
where $\sat^\dagger : \Bexp \to 2^\State$ is the usual lifting of $\sat : T \to 2^\State$ to Boolean expression over $T$.
\end{definition}

The rules defining $\rden{i}{e}$ are reminiscent of the big-step
semantics of imperative languages, which arise as instances of the
model for various choices of $i$. The following result says that the
language model from the previous section abstracts the various
relational interpretations in a sound and complete way. It was first
proved for KAT by~\citet{KS96a}.

\begin{restatable}{theorem}{soundcompleteforrel}%
\label{thm:sound-complete-for-rel}
The language model is sound and complete for the relational model: 
\[ \den{e} = \den{f} \quad \iff \quad\forall i.\, \rden{i}{e} = \rden{i}{f}\]
\end{restatable}
\noindent
It is worth noting that \Cref{thm:sound-complete-for-rel} also holds
for refinement (\ie, with $\subseteq$ instead of $=$).

\newcommand\Var{\mathsf{Var}}
\begin{example}[IMP]%
\label{ex:imp}
Consider a simple imperative programming language IMP with variable assignments
and arithmetic and boolean expressions:\[
\begin{array}{r@{\quad}r@{~~}c@{~~}l}
\textit{arithmetic expressions} & a \in \Aa & ::= &
  x \in \Var \mid n \in \ZZ \mid a_1 + a_2 \mid a_1 - a_2 \mid a_1 \times a_2 \\
\textit{boolean expressions} & b \in \BB & ::= &
  \kw{false} \mid \kw{true} \mid a_1 < a_2 \mid \kw{not}~b
  \mid b_1~\kw{and}~b_2 \mid b_1~\kw{or}~b_2  \\
\textit{commands} &c \in \CC & ::= &
  \kw{skip} \mid x \coloneqq a \mid c_1;c_2
    \mid \kw{if}~b~\kw{then}~c_1~\kw{else}~c_2
    \mid \kw{while}~b~\kw{do}~c
\end{array}
\]
IMP can be modeled in GKAT using actions for assignments
and primitive tests for comparisons,\footnote{%
Technically, we can only reserve a test for
a \emph{finite subset} of comparisons, as $T$ is
finite. However, for reasoning about pairwise equivalences of
programs, which only contain a finite number of comparisons, this
restriction is not essential.}
\begin{mathpar}
  \Sigma = \set{x\coloneq a}{x \in \Var, a \in \Aa}
  \and
  T = \set{a_1 < a_2}{a_1,a_2 \in \Aa}
\end{mathpar}
and interpreting GKAT expressions over the state space of variable
assignments $\State \coloneq \Var \to \ZZ$:
\begin{align*}
  \eval(x \coloneq a) &\defeq
    \set{(\sigma, \sigma[x:=n])}{\sigma \in \State, n = \Aa\den{a}(\sigma)}\\
  \sigma[x \defeq n] &\defeq \lambda y.\, \begin{cases}
      n         &\text{if } y=x\\
      \sigma(y) &\text{else}
    \end{cases}\\
  \sat(a_1 < a_2) &\defeq
    \set{\sigma \in \State}{\Aa\den{a_1}(\sigma) < \Aa\den{a_2}(\sigma)},
\end{align*}
where $\Aa\den{a} : \State \to \ZZ$ denotes arithmetic evaluation.
Sequential composition, conditionals, and while loops in IMP are modeled by
their GKAT counterparts; $\kw{skip}$ is modeled by $1$.
Thus, IMP equivalence refines GKAT equivalence
(\Cref{thm:sound-complete-for-rel}). For example, the program transformation
\begin{align*}
  &\kw{if}~x<0~\kw{then}~(x\defeq 0 - x; x \defeq 2 \times x)~\kw{else}~
    (x \defeq 2 \times x)\\
  \rightsquigarrow{}
   &(\kw{if}~x<0~\kw{then}~x\defeq 0 - x~\kw{else}~\kw{skip});
    x \defeq 2 \times x
\end{align*}
is sound by the equivalence $pq +_b q \equiv (p +_b 1)\cdot q$.
We study such equivalences further in \Cref{sec:axioms}.
\end{example}

\subsection{Probabilistic Model}%
\label{sec:prob-model}
In this subsection, we give a third interpretation of GKAT expressions in terms
of sub-Markov kernels, a common model for probabilistic programming languages
(PPLs). We show that the language model is sound and complete for this model as well.

We briefly review some basic primitives commonly used in the denotational
semantics of PPLs. For a countable set\footnote{%
We restrict to countable state spaces (i.e., discrete distributions)
for ease of presentation, but this assumption is not essential.
\extended%
  {\Cref{sec:prob-model-continuous}}
  {See the extend version~\cite{gkat-full}}
for a more general version using
measure theory and Lebesgue integration.}
$X$, we
let $\Dist(X)$ denote the set of subdistributions over $X$, \ie, the set of
probability assignments $f : X \to [0,1]$ summing up to at most $1$---\ie, \(
  \sum_{x \in X} f(x) \leq 1
\).
A common distribution is the \emph{Dirac distribution} or \emph{point mass} on
$x \in X$, denoted $\delta_x\in\Dist(X)$; it is the map $y\mapsto [y=x]$ assigning
probability $1$ to x, and probability $0$ to $y \neq x$.
(The \emph{Iverson bracket} $[\phi]$ is defined to be $1$ if the
statement $\phi$ is true, and $0$ otherwise.)
Denotational models of PPLs typically interpret programs as
\emph{Markov kernels}, maps of type $X \to \Dist(X)$.
Such kernels can be composed in sequence using Kleisli composition,
since $\Dist(-)$ is a monad~\cite{giry1982categorical}.

\begin{definition}[Probabilistic Interpretation]
Let $i = (\State, \eval, \sat)$ be a triple consisting of
\begin{itemize}
  \item a countable set of \emph{states} $\State$;
  \item for each action $p \in \Sigma$,
    a sub-Markov kernel $\eval(p) \colon \State \to \Dist(\State)$; and
  \item for each primitive test $t \in T$,
    a set of states $\sat(t) \subseteq \State$.
\end{itemize}
Then the \emph{probabilistic interpretation} of an expression $e$ with respect to $i$
is the sub-Markov kernel $\pden{i}{e} \colon \State \to \Dist(\State)$ defined
as follows:
\begin{align*}
  \pden{i}{p} &\defeq
    \eval(p)\qquad
  \pden{i}{b}(\sigma) \defeq
    [\sigma \in \sat^\dagger(b)] \cdot \delta_\sigma\\
  \pden{i}{e\cdot f}(\sigma)(\sigma') &\defeq
    \sum_{\sigma''} \pden{i}{e}(\sigma)(\sigma'') \cdot
    \pden{i}{f}(\sigma'')(\sigma')\\
  \pden{i}{e +_b f}(\sigma) &\defeq
    [\sigma \in \sat^\dagger(b)] \cdot \pden{i}{e}(\sigma)
    + [\sigma \in \sat^\dagger(\bneg{b})] \cdot \pden{i}{f}(\sigma)\\
  \pden{i}{e^{(b)}}(\sigma)(\sigma') &\defeq
    \lim_{n\to\infty} \pden{i}{{(e +_b 1)}^n \cdot \bneg{b}}(\sigma)(\sigma')
&&\qedhere
\end{align*}
The proofs that the limit exists and that $\pden{i}{e}$ is sub-Markov
for all $e$ can be found
in \extended{\Cref{lem:pmodel-well-defined}}{the extended
version of the paper~\cite{gkat-full}}.

\end{definition}

\begin{restatable}{theorem}{soundcompleteforprob}%
\label{thm:sound-complete-for-prob}
The language model is sound and complete for the probabilistic model: 
\[ \den{e} = \den{f} \quad \iff \quad\forall i.\, \pden{i}{e} = \pden{i}{f}\]
\end{restatable}
\begin{proof}[Proof Sketch]
By mutual implication.
\begin{itemize}
\item[$\Rightarrow$:]
  For soundness, we define a map
  $\kappa_i \colon \Gs \to \State \to \Dist(\State)$
  from guarded strings to sub-Markov kernels:
  \begin{align*}
    \kappa_i(\alpha)(\sigma) &\defeq
      [\sigma \in \sat^\dagger(\alpha)] \cdot \delta_\sigma\\
    \kappa_i(\alpha p w)(\sigma)(\sigma') &\defeq
      [\sigma \in \sat^\dagger(\alpha)] \cdot
      \sum_{\sigma''} \eval(p)(\sigma)(\sigma'') \cdot
      \kappa_i(w)(\sigma'')(\sigma)
  \end{align*}
  We then lift $\kappa_i$ to languages via pointwise summation,  \(
    \kappa_i(L) \defeq \sum_{w \in L} \kappa_i(w)
  \),
  and establish that any probabilistic interpretation factors through
  the language model via $\kappa_i$:
$\pden{i}{-} = \kappa_i \circ \den{-}$.

\item[$\Leftarrow$:]
  For completeness, we construct an interpretation $i \defeq (\Gs, \eval, \sat)$
  over $\Gs$ as follows,
  \begin{mathpar}
    \eval(p)(w) \defeq \Unif(\set{wp\alpha}{\alpha \in \At})
    \and
    \sat(t) \defeq \set{x\alpha \in \Gs}{\alpha \leq t}
  \end{mathpar}
  and show that $\den{e}$ is fully determined by $\pden{i}{e}$:
  \begin{equation*}
    \den{e} = \set{\alpha x \in \Gs}{\pden{i}{e}(\alpha)(\alpha x) \neq 0}.
  \qedhere
  \end{equation*}
\end{itemize}
\end{proof}
\noindent
As for \Cref{thm:sound-complete-for-rel},
\Cref{thm:sound-complete-for-prob} can also be shown for refinement
(\ie, with $\subseteq$ and $\leq$ instead of $=$).

\begin{example}[Probabilistic IMP]%
\label{ex:prob-imp}
We can extend IMP from \Cref{ex:imp} with a \emph{probabilistic assignment}
command $x \sim \mu$,
where $\mu$ ranges over sub-distributions on $\ZZ$, as follows:
\begin{mathpar}
  c ::= \ldots \mid x \sim \mu
  \and
  \Sigma = \ldots \cup \set{x \sim \mu}{x \in \Var, \mu \in \Dist(\ZZ)}
\end{mathpar}
The interpretation $i = (\Var \to \ZZ, \eval, \sat)$ is as before,
except we now restrict to a finite set of variables to guarantee that
the state space is countable, and interpret actions as sub-Markov
kernels:
\begin{mathpar}
  \eval(x \defeq n)(\sigma) \defeq
    \delta_{\sigma[x \defeq n]}
  \and
  \eval(x \sim \mu)(\sigma) \defeq
    \sum_{n \in \ZZ} \mu(n) \cdot \delta_{\sigma[x \defeq n]}
\end{mathpar}
\end{example}

A concrete example of a PPL based on GKAT is McNetKAT~\cite{mcnetkat},
a recent language and verification tool for reasoning about the
packet-forwarding behavior in networks.



\section{Axiomatization}%
\label{sec:axioms}

In most programming languages, the same behavior can be realized using
different programs. For example, we expect the programs
$\kw{if}~b~\kw{then}~e~\kw{else}~f$ and
$\kw{if}~(\kw{not}~b)~\kw{then}~f~\kw{else}~e$ to encode the same
behavior. Likewise, different expressions in GKAT can denote the same
language of guarded strings. For instance, the previous example is
reflected in GKAT by the fact that the language semantics of $e +_b f$
and $f +_{\bneg{b}} e$ coincide. This raises the questions: what other
equivalences hold between GKAT expressions? And, can all equivalences
be captured by a finite number of equations? In this section, we give
some initial answers to these questions, by proposing a set of axioms
for GKAT and showing that they can be used to prove a large class of
equivalences.

\subsection{Some Simple Axioms}

\begin{figure}
\small
\def\arraystretch{1.2}
\begin{tabular}{l@{~} >{$}r<{$}@{~} >{$}c<{$}@{~} >{$}l<{$} l l@{~} >{$}r<{$}@{~} >{$}c<{$}@{~} >{$}l<{$} l}
\multicolumn{5}{l}{\hspace*{-1ex}\textbf{Guarded Union Axioms}}
    & \multicolumn{5}{l}{\hspace*{-1ex}\textbf{Sequence Axioms} (inherited from KA)} \\
\customlabel{ax:idemp}{U1}. &e+_b e            &\equiv &e                     & (idempotence) & 
   \customlabel{ax:seqassoc}{S1}. & (e \cdot f) \cdot g &\equiv  &e \cdot (f \cdot g)  & (associativity)\\ 
\customlabel{ax:skewcomm}{U2}. &e +_b f           &\equiv &f +_{\bneg{b}} e      & (skew commut.) & 
   \customlabel{ax:absleft}{S2}. &0 \cdot e           &\equiv  &0                    & (absorbing left)\\ 
\customlabel{ax:skewassoc}{U3}. & (e +_b f) +_c g   &\equiv &e +_{bc} (f +_c g)    & (skew assoc.) & 
   \customlabel{ax:absright}{S3}. &e \cdot 0           &\equiv  &0                    & (absorbing right)\\ 
\customlabel{ax:guard-if}{U4}. &e +_b f           &\equiv &be +_b f              & (guardedness) & 
   \customlabel{ax:neutrleft}{S4}. &1 \cdot e           &\equiv  &e                    & (neutral left)\\ 
\customlabel{ax:rightdistr}{U5}. &eg +_b fg &\equiv & (e +_b f) \cdot g             & (right distrib.) & 
   \customlabel{ax:neutrright}{S5}. &e \cdot 1           &\equiv  &e                    & (neutral right) \\[0.5em] 
\multicolumn{5}{l}{\hspace*{-1ex}\textbf{Guarded Loop Axioms}} \\
    \customlabel{ax:unroll}{W1}. &e^{(b)} &\equiv  & ee^{(b)} +_{b} 1  & (unrolling) & 
\multirow{2}{*}{\customlabel{ax:fixpoint}{W3}.} & \multicolumn{3}{c}{\multirow{2}{*}{$\inferrule{g \equiv eg +_b f}{g \equiv e^{(b)} f}\ \text{if}\ E(e) \equiv 0$}} & \multirow{2}{*}{(fixpoint)} \\
    \customlabel{ax:tighten}{W2}. & {(e +_c 1)}^{(b)} &\equiv& {(ce)}^{(b)} & (tightening)
\end{tabular}%
\caption{Axioms for GKAT-expressions.}\label{fig:axioms}
\end{figure}

As an initial answer to the first question, we propose the following.

\begin{definition}
We define $\equiv$ as the smallest congruence (with respect to all operators) on
$\Exp$ that satisfies the axioms given in Figure~\ref{fig:axioms} (for all
$e, f, g \in \Exp$ and $b, c, d \in \Bexp$) and subsumes Boolean equivalence
in the sense that $b \equiv_\text{BA} c$ implies  $b \equiv c$.
\end{definition}

The guarded union axioms~(\nameref{ax:idemp}-\nameref{ax:rightdistr})
can be understood intuitively in terms of conditionals. For instance,
we have the law $e +_b f \equiv f +_{\bneg{b}} e$ discussed before,
but also $eg +_b fg \equiv (e +_b f) \cdot g$, which says that $g$ can
be ``factored out'' of branches of a guarded union. Equivalences for
sequential composition are also intuitive. For instance, $0 \cdot
e \equiv 0$ encodes that any instruction after failure is irrelevant,
because the program has failed. The axioms for
loops~(\nameref{ax:unroll}--\nameref{ax:fixpoint}) are more subtle.
The axiom $e^{(b)} \equiv ee^{(b)} +_b 1$~(\nameref{ax:unroll}) says
that we can think of a guarded loop as equivalent to its
unrolling---\ie, the program $\kw{while}\ b\
\kw{do}\ e$ has the same behavior as the program $\kw{if}\ b\ \kw{then}\ (e;\
\kw{while}\ b\ \kw{do}\ e)\ \kw{else}\ \kw{skip}$. The axiom ${(e +_c 1)}^{(b)} \equiv {(ce)}^{(b)}$~(\nameref{ax:tighten}) states that
if part of a loop body does not have an effect (\ie, is equivalent to $\kw{skip}$), it can be omitted;
we refer to this transformation as \emph{loop tightening}.

To explain the fixpoint axiom~(\nameref{ax:fixpoint}), disregard the side-condition for a moment. In a
sense, this rule states that if $g$ tests (using $b$) whether to
execute $e$ and loop again or execute $f$ (\ie, if $g \equiv eg
+_b f$) then $g$ is a $b$-guarded loop followed by $f$ (\ie,
$g \equiv e^{(b)}f$). However, such a rule is not sound in general.
For instance, suppose $e, f, g, b = 1$; in that case, $1 \equiv
1 \cdot 1 +_1 1$ can be proved using the other axioms, but applying
the rule would allow us to conclude that $1 \equiv 1^{(1)} \cdot 1$,
even though $\sem{1} = \At$ and $\sem{1^{(1)} \cdot 1} = \emptyset$!
The problem here is that, while $g$ is tail-recursive as required by
the premise, this self-similarity is trivial because $e$ does not
represent a productive program. We thus need to restrict the
application of the inference rule to cases where the loop body
is \emph{strictly productive}---\ie, where $e$ is guaranteed to
execute \emph{some} action. To this end, we define the function $E$ as
follows.

\begin{definition}
The function $E: \Exp\to \Bexp$ is defined inductively as follows:
\begin{mathpar}
E(b) \defeq b
\and
E(p) \defeq 0
\and
E(e +_b f) \defeq b \cdot E(e) + \bneg{b} \cdot E(f)
\and
E(e \cdot f) \defeq E(e) \cdot E(f)
\and
E(e^{(b)}) \defeq \bneg{b}
\end{mathpar}
\qedhere
\end{definition}

Intuitively, $E(e)$ is the weakest test that guarantees that $e$ terminates successfully, but does not perform any action.
For instance, $E(p)$ is $0$---the program $p$ is guaranteed to perform the action $p$.
Using $E$, we can now restrict the application of the fixpoint rule to the cases where $E(e) \equiv 0$, \ie, where $e$ performs an action under any circumstance.

\begin{restatable}[Soundness]{theorem}{soundness}
The GKAT axioms are sound for the language model: for all $e, f \in \Exp$,
\[
  e \equiv f \quad \implies \quad \den{e} = \den{f}.
\]
\end{restatable}
\begin{proof}[Proof Sketch]
By induction on the length of derivation of the congruence $\equiv$.
We provide the full proof in the \extended{appendix}{extended version~\cite{gkat-full}} and show just the proof for the fixpoint rule.
Here, we should argue that if $E(e) \equiv 0$ and $\sem{g} = \sem{eg +_b f}$, then also $\sem{g} = \sem{e^{(b)}f}$.
We note that, using soundness of~(\nameref{ax:unroll}) and~(\nameref{ax:rightdistr}),
we can derive that $\sem{ e^{(b)} f} = \sem{ (ee^{(b)}+_b 1) f} = \sem{ee^{(b)} f +_b f}$.

We reason by induction on the length of guarded strings.
In the base case, we know that $\alpha \in \sem{g}$ if and only if $\alpha \in \sem{eg +_b f}$; since $E(e) \equiv 0$, the latter holds precisely when $\alpha \in \sem{f}$ and $\alpha \leq \bneg{b}$, which is equivalent to $\alpha \in \sem{e^{(b)}f}$.
For the inductive step, suppose the claim holds for $y$; then
\begin{align*}
&\phantom{{}\iff{}} \alpha p y \in  \sem{g} \\
&\iff  \alpha p y\in  \sem{eg +_b f} \\
&\iff \alpha p y \in \sem{eg} \land \alpha \leq  b \;\;\textbf{or}\;\; \alpha p y \in \sem{f} \land \alpha \leq \bneg b\\
&\iff \exists y, \beta.\ y = y_1\beta{}y_2 \wedge \alpha p y_1\beta \in \sem{e} \land \beta{}y_2 \in \sem{g} \land \alpha \leq  b \;\;\textbf{or}\;\; \alpha p y \in \sem{f} \land \alpha \leq \bneg b \tag{$E(e) = 0$} \\ 
&\iff \exists y, \beta.\ y = y_1\beta{}y_2 \wedge \alpha p y_1\beta \in \sem{e} \land \beta{}y_2 \in \sem{e^{(b)}f} \land \alpha \leq  b \;\;\textbf{or}\;\; \alpha p y \in \sem{f} \land \alpha \leq \bneg b \tag{IH} \\
&\iff \alpha p y \in \sem{ee^{(b)}f} \land \alpha \leq  b \;\;\textbf{or}\;\; \alpha p y \in \sem{f} \land \alpha \leq \bneg b &  \tag{$E(e) = 0$} \\ 
&\iff \alpha p y \in \sem{ee^{(b)}f +_b f} = \sem{e^{(b)}f}
\tag*{\qedhere}
\end{align*}
\end{proof}

\subsection{A Fundamental Theorem}\label{sec:ft}

\begin{figure}
\small
\def\arraystretch{1.2}
\begin{tabular}{l@{~} >{$}r<{$}@{~} >{$}c<{$}@{~} >{$}l<{$} l@{\qquad\qquad} l@{~} >{$}r<{$}@{~} >{$}c<{$}@{~} >{$}l<{$} l}
\multicolumn{5}{l}{\hspace*{-1ex}\textbf{Guarded Union Facts}} &
    \multicolumn{5}{l}{\hspace*{-2ex}\textbf{Guarded Iteration Facts}}\\[.1ex]
\customlabel{fact:skewassoc-dual}{U3'}. &e +_b (f +_c g)   &\equiv & (e +_b f) +_{b+c} g & (skew assoc.) & 
    \customlabel{fact:guard-loop-out}{W4}.  &e^{(b)} &\equiv  &e^{(b)} \cdot \bneg{b}                    & (guardedness)\\ 
\customlabel{fact:guard-if-dual}{U4'}. &e +_b f           &\equiv &e +_b \bneg{b}f     & (guardedness) & 
   \customlabel{fact:guard-loop-in}{W4'}. &e^{(b)} &\equiv  & {(be)}^{(b)}                                & (guardedness)\\ 
\customlabel{fact:leftdistr}{U5'}. &b \cdot (e +_c f) &\equiv &be +_c bf           & (left distrib.) & 
   \customlabel{fact:neutral-loop}{W5}.  &e^{(0)} &\equiv  &1                                         & (neutrality)\\ 
\customlabel{fact:neutrright2}{U6}. &e +_b 0            &\equiv &be                  & (neutral right) & 
   \customlabel{fact:absorb-loop}{W6}.  &e^{(1)} &\equiv  &0                                         & (absorption)\\ 
\customlabel{fact:trivright}{U7}. &e +_0 f            &\equiv &f                   & (trivial right) & 
   \customlabel{fact:absorb-loop-bool}{W6'}. &b^{(c)} &\equiv &\bneg{c}                                   & (absorption)\\ 
\customlabel{fact:branchsel}{U8}. & b \cdot (e +_b f) &\equiv &be            & (branch selection) &
    \customlabel{fact:fusion}{W7}.  &e^{(c)} &\equiv &e^{(bc)} \cdot e^{(c)}                     & (fusion)\\ 
\end{tabular}
\caption{Derivable GKAT facts}\label{fig:axioms-derivable}
\end{figure}

The side condition on~(\nameref{ax:fixpoint}) is inconvenient when proving facts about loops.
However, it turns out that we can transform any loop into an equivalent, \emph{productive} loop---\ie, one with a loop body $e$ such that $E(e) \equiv 0$.
To this end, we need a way of decomposing a GKAT expression into a guarded sum of an expression that describes termination, and another (strictly productive) expression that describes the next steps that the program may undertake.
As a matter of fact, we already have a handle on the former term: $E(e)$ is a Boolean term that captures the atoms for which $e$ may halt immediately.
It therefore remains to describe the next steps of a program.

\begin{definition}[Derivatives]
For $\alpha \in \At$
we define $D_\alpha\colon \Exp \to 2+\Sigma\times \Exp$ inductively as follows,
where $2=\sset{0,1}$ is the two-element set:
\begin{mathpar}
D_\alpha(b) = \begin{cases}
1 & \alpha \leq b \\
0 & \alpha \not\leq b
\end{cases}
\and
D_\alpha(p) = (p,1)
\and
D_\alpha(e +_b f) = \begin{cases}
D_\alpha(e) & \alpha \leq b \\
D_\alpha(f) & \alpha \leq \bneg{b} \\
\end{cases}
\and
D_\alpha(e \cdot f) = \begin{cases}
(p, e'\cdot f) & D_\alpha(e) = (p,e')\\
0 & D_\alpha(e) = 0\\
D_\alpha(f) & D_\alpha(e) = 1\\
\end{cases}
\and
D_\alpha(e^{(b)})  = \begin{cases}
(p, e'\cdot e^{(b)}) & \alpha \leq  b \land D_\alpha(e) = (p,e')\\
0 & \alpha \leq  b \land D_\alpha(e) \in 2\\
1 & \alpha \leq \bneg b \\
\end{cases}\qedhere
\end{mathpar}
\end{definition}

\noindent
We will use a general type of guarded union to sum over an atom-indexed set of expressions.

\begin{restatable}{definition}{generalsum}
Let $\Phi \subseteq \At$, and let ${\{e_\alpha\}}_{\alpha \in \Phi}$ be a set of expressions indexed by $\Phi$. We write
\begin{mathpar}
\dsum_{\alpha \in \Phi} e_\alpha =
\begin{cases}
e_{\beta} +_\beta \Bigl(\; \dsum\limits_{\alpha \in \Phi\setminus\{\beta\}} e_\alpha \Bigr) & \beta \in \Phi \\
0 & \Phi = \emptyset
\end{cases}
\end{mathpar}
Like other operators on indexed sets, we may abuse notation and replace $\Phi$ by a predicate over some atom $\alpha$, with $e_\alpha$ a function of $\alpha$; for instance, we could write $\dsum_{\alpha \leq 1} \alpha \equiv 1$.
\end{restatable}

\begin{remark}
The definition above is ambiguous in the choice of $\beta$.
However, because of skew-commutativity~(\nameref{ax:skewcomm})
and skew-associativity~(\nameref{ax:skewassoc}),
that does not change the meaning of the expression
as far as $\equiv$ is concerned.
For the details, see \extended{\Cref{sec:gen_choice}}
  {the extended version~\cite{gkat-full}}.
\end{remark}

We are now ready to state the desired decomposition of terms. Following~\cite{rutten-2000,silva-thesis}, we call this the \emph{fundamental theorem} of GKAT, in reference to the strong analogy with the fundamental theorem of calculus, as explained in~\cite{rutten-2000,silva-thesis}.
The proof is included in
\extended{\Cref{apx:omitted}}{the extended version~\cite{gkat-full}}.

\begin{restatable}[Fundamental Theorem]{theorem}{ft}\label{thm:ft}
For all GKAT programs $e$, the following equality holds:
\begin{equation}\label{eq:ft}
e \equiv 1 +_{E(e)} D(e),\qquad\text{ where }~
  D(e) \defeq \dsum_{\alpha\colon D_\alpha(e)=(p_\alpha ,e_\alpha)}  p_\alpha \cdot e_\alpha.
\end{equation}
\end{restatable}

\noindent
The following observations about $D$ and $E$ are also useful.
\begin{restatable}{lemma}{lemmaed}\label{lem:d-vs-e}
Let $e \in \Exp$; its components $E(e)$ and $D(e)$ satisfy
the following identities:
\begin{mathpar}
E(D(e)) \equiv 0
\and
\bneg{E(e)} \cdot D(e) \equiv D(e)
\and
\bneg{E(e)} \cdot e \equiv D(e)
\end{mathpar}
\end{restatable}

Using the fundamental theorem and the above, we can now show how to syntactically transform any loop into an equivalent loop whose body $e$ is strictly productive.

\begin{lemma}[Productive Loop]\label{lem:prod}
Let $e \in \Exp$ and $b \in \Bexp$. We have $e^{(b)} \equiv {D(e)}^{(b)}$.
\end{lemma}
\begin{proof}
Using \Cref{lem:d-vs-e}, we derive as follows:
\[
    e^{(b)}
        \stackrel{\text{\hyperlink{thm:ft}{FT}}}\equiv {(1 +_{E(e)} D(e))}^{(b)}
        \stackrel{\text{\nameref{ax:skewcomm}}}\equiv {(D(e) +_{\bneg{E(e)}} 1)}^{(b)}
        \stackrel{\text{\nameref{ax:tighten}}}\equiv {(\bneg{E(e)} D(e))}^{(b)}
        \equiv {D(e)}^{(b)}
    \qedhere
\]
\end{proof}

\subsection{Derivable Facts}

The GKAT axioms can be used to derive other natural equivalences of programs, such as the ones in \Cref{fig:axioms-derivable}.
For instance, $e^{(b)} \equiv e^{(b)}\bneg{b}$, labelled~(\nameref{fact:guard-loop-out}), says that $b$ must be false when $e^{(b)}$ ends.

\begin{restatable}{lemma}{derivable}%
\label{lem:derivable}
The facts in \Cref{fig:axioms-derivable} are derivable from the axioms.
\end{restatable}
\begin{proof}[Proof Sketch]
Let us start by showing~(\nameref{fact:neutrright2}).
\begin{align*}
e +_b 0
    &\equiv be +_b 0
        \tag{\nameref{ax:guard-if}. $e +_b f \equiv be +_b f$} \\ 
    &\equiv 0 +_{\bneg{b}} be
        \tag{\nameref{ax:skewcomm}. $e +_b f \equiv f +_{\bneg{b}} e$} \\ 
    &\equiv \bneg{b} be +_{\bneg{b}} be
        \tag{Boolean algebra and \nameref{ax:absleft}. $0 \equiv 0e$} \\ 
    &\equiv be +_{\bneg b} be
        \tag{\nameref{ax:guard-if}. $e +_b f\equiv be +_b f$} \\ 
    &\equiv be
        \tag{\nameref{ax:idemp}. $e +_b e \equiv e$} 
\intertext{%
    To prove~(\nameref{fact:fusion}), we use the productive loop lemma and the fixpoint axiom~(\nameref{ax:fixpoint}).
}
e^{(c)}
    &\equiv e^{(c)} +_{bc}  e^{(c)}
        \tag{\nameref{ax:idemp}. $e +_b e \equiv e$} \\ 
    &\equiv {(D(e))}^{(c)} +_{bc}  e^{(c)}
        \tag{Productive loop lemma}\\
    &\equiv (D(e){D(e)}^{(c)} +_c 1) +_{bc}  e^{(c)}
        \tag{\nameref{ax:unroll}. $e^{(b)} \equiv ee^{(b)} +_b 1$} \\ 
    &\equiv c \cdot (D(e){D(e)}^{(c)} +_c 1) +_{bc} e^{(c)}
        \tag{\nameref{ax:guard-if} and Boolean algebra} \\
    &\equiv c \cdot D(e){D(e)}^{(c)} +_{bc}  e^{(c)}
        \tag{\nameref{fact:branchsel}. $b \cdot (e +_b f) \equiv be$} \\ 
    &\equiv D(e){D(e)}^{(c)} +_{bc}  e^{(c)}
        \tag{\nameref{ax:guard-if} and Boolean algebra} \\
    &\equiv D(e)e^{(c)}  +_{bc}  e^{(c)}
        \tag{Productive loop lemma}\\
    &\equiv {D(e)}^{(bc)} e^{(c)}
        \tag{\nameref{ax:fixpoint}}\\
    &\equiv e^{(bc)} e^{(c)}
        \tag{Productive loop lemma}
\end{align*}
The remaining proofs appear in the \extended{appendix}{extended version~\cite{gkat-full}}.
\end{proof}

We conclude our presentation of derivable facts by showing one more interesting
fact. Unlike the derived facts above, this one is an implication: if the test
$c$ is invariant for the program $e$ given that a test $b$ succeeds, then $c$ is
preserved by a $b$-loop on $e$.

\begin{lemma}[Invariance]\label{lem:invariance}
Let $e \in \Exp$ and $b,c \in \Bexp$.
If $cbe \equiv cbec$, then $ce^{(b)} \equiv {(ce)}^{(b)}c$.
\end{lemma}
\begin{proof}
We first derive a useful equivalence, as follows:
\begin{align*}
cb \cdot D(e)
    &\equiv cb \cdot \bneg{E(e)} \cdot e
        \tag{\Cref{lem:d-vs-e}} \\
    &\equiv \bneg{E(e)} \cdot cbe
        \tag{Boolean algebra} \\
    &\equiv \bneg{E(e)} \cdot cbec
        \tag{premise} \\
    &\equiv cb \cdot \bneg{E(e)} \cdot ec
        \tag{Boolean algebra} \\
    &\equiv cb \cdot D(e) \cdot c
        \tag{\Cref{lem:d-vs-e}}
\intertext{%
    Next, we show the main claim by deriving
}
ce^{(b)}
    &\equiv c \cdot {D(e)}^{(b)}
        \tag{Productive loop lemma} \\
    &\equiv c \cdot (D(e) \cdot {D(e)}^{(b)} +_b 1)
        \tag{\nameref{ax:unroll}} \\
    &\equiv c \cdot (D(e) \cdot e^{(b)} +_b 1)
        \tag{Productive loop lemma} \\
    &\equiv c \cdot (b \cdot D(e) \cdot e^{(b)} +_b 1)
        \tag{\nameref{ax:skewcomm}} \\
    &\equiv cb \cdot D(e) \cdot e^{(b)} +_b c
        \tag{\nameref{fact:leftdistr}} \\
    &\equiv cb \cdot D(e) \cdot ce^{(b)} +_b c
        \tag{above derivation} \\
    &\equiv c \cdot D(e) \cdot ce^{(b)} +_b c
        \tag{\nameref{ax:skewcomm}} \\
    &\equiv {(c \cdot D(e))}^{(b)}c
        \tag{\nameref{ax:fixpoint}} \\
    &\equiv {D(ce)}^{(b)}c
        \tag{Def. $D$, Boolean algebra} \\
    &\equiv {(ce)}^{(b)}c
        \tag{Productive loop lemma}
\end{align*}
This completes the proof.
\end{proof}

\subsection{A Limited Form of Completeness}%
\label{sec:partial-completness}

Above, we considered a number of axioms that were proven sound with respect to the language model. Ultimately, we would like to show the converse, \ie, that these axioms are sufficient to prove all equivalences between programs, meaning that whenever $\sem{e} = \sem{f}$, it also holds that $e \equiv f$.

We return to this general form of completeness in \Cref{sec:complete}, when we can rely on the coalgebraic theory of GKAT developed in \Cref{sec:kleene,sec:decision}.
At this point, however, we can already prove a special case of completeness related to
Hoare triples.
Suppose $e$ is a GKAT program, and $b$ and $c$ are Boolean expressions
encoding pre- and postconditions.
The equation $\sem{be} = \sem{bec}$ states that every finite, terminating
run of $e$ starting from a state satisfying $b$ concludes in
a state satisfying $c$.
The following states that all valid Hoare triples of this kind can
be established axiomatically:

\begin{restatable}[Hoare completeness]{theorem}{hoarecompleteness}%
\label{thm:hoare-completeness}
Let $e \in \Exp$, $b, c \in \Bexp$.
If $\sem{bec} = \sem{be}$, then $bec \equiv be$.
\end{restatable}
\begin{proof}[Proof Sketch]
  By induction on $e$. We show only the case for while loops and defer the
  full proof to \extended{\Cref{apx:omitted}}{the extended version~\cite{gkat-full}}.

  If $e = e_0^{(d)}$, first note that if $b \equiv 0$, then the claim follows trivially.
  For $b \not\equiv 0$, let
  \[
      h = \sum \{ \alpha \in \At : \exists n. \sem{b} \diamond \sem{de_0}^n \diamond \sem{\alpha} \neq \emptyset \}.
  \]
  We make the following observations.
  \begin{enumerate}[(i), leftmargin=1cm]
      \item\label{invariant:start}
      Since $b \not\equiv 0$, we have that $\sem{b} \diamond \sem{{de_0}}^0 \diamond \sem{b} = \sem{b} \neq \emptyset$, and thus $b \leq h$.

      \item\label{invariant:end}
      If $\alpha \leq h\bneg{d}$, then in particular $\gamma{}w\alpha \in \sem{b} \diamond \sem{de_0}^n \diamond \sem{\alpha}$ for some $n$ and $\gamma{}w$.
      Since $\alpha \leq \bneg{d}$, it follows that $\gamma{}w\alpha \in \sem{be_0^{(d)}} = \sem{be_0^{(d)}c}$, and thus $\alpha \leq c$.
      Consequently, $h\bneg{d} \leq c$.

      \item\label{invariant:hold}
      If $\alpha{}w\beta \in \sem{dhe_0}$, then $\alpha \leq h$ and hence there exists an $n$ such that $\gamma{}x\alpha \in \sem{b} \diamond \sem{de_0}^n \diamond \sem{\beta}$.
      But then $\gamma{}x\alpha{}w\beta \in \sem{b}\ \diamond \sem{de_0}^{n+1} \diamond \sem{\beta}$, and therefore $\beta \leq h$.
      We can conclude that $\sem{dhe_0} = \sem{dhe_{0}h}$; by induction, it follows that $dhe_{0}h \equiv dhe_0$.
  \end{enumerate}

  \noindent
  Using these observations and the invariance lemma (\Cref{lem:invariance}), we derive
  \begin{align*}
  be_{0}^{(d)}c
      &\equiv bhe_0^{(d)}c
          \tag{By~\ref{invariant:start}} \\
      &\equiv b \cdot {(he_0)}^{(d)}hc
          \tag{Invariance and~\ref{invariant:hold}} \\
      &\equiv b \cdot {(he_0)}^{(d)}\bneg{d}hc
          \tag{\nameref{fact:guard-loop-out}} \\
      &\equiv b \cdot {(he_0)}^{(d)}\bneg{d}h
          \tag{By~\ref{invariant:end}} \\
      &\equiv b \cdot {(he_0)}^{(d)}h
          \tag{\nameref{fact:guard-loop-out}} \\
      &\equiv bh{e_0}^{(d)}
          \tag{Invariance and~\ref{invariant:hold}} \\
      &\equiv b {e_0}^{(d)}
          \tag{By~\ref{invariant:start}}
  \end{align*}
This completes the proof.
\end{proof}

\noindent
As a special case, the fact that a program has no traces at all can be shown
axiomatically.

\begin{corollary}[Partial Completeness]%
\label{cor:partial-completeness}
If $\sem{e} = \emptyset$, then $e \equiv 0$.
\end{corollary}
\begin{proof}
We have $\sem{1 \cdot e} = \sem{e} = \emptyset = \sem{1 \cdot e \cdot 0}$,
and thus $e \equiv 1 \cdot e \equiv 1 \cdot e \cdot 0 \equiv 0$ by
\Cref{thm:hoare-completeness}.
\end{proof}

We will return to deriving a general completeness result in \Cref{sec:complete}.
This will rely on the coalgebraic theory of GKAT, which we develop
next (\Cref{sec:kleene,sec:decision}).

\section{Automaton Model and Kleene Theorem}%
\label{sec:kleene}

In this section, we present an automaton model that accepts traces
(\ie, guarded strings) of GKAT programs. We then present language-preserving
constructions from GKAT expressions to automata, and conversely, from automata
to expressions. Our automaton model is rich enough to express programs that go
beyond GKAT;\@ in particular, it can encode traces of programs with \textbf{goto}
statements that have no equivalent GKAT program~\cite{KT08a}. In order to
obtain a Kleene Theorem for GKAT, that is, a correspondence between automata
and GKAT programs, we identify conditions ensuring that the
language accepted by an automaton corresponds to a valid GKAT program.

\subsection{Automata and Languages}
Let $G$ be the functor $GX = {(2 + \Sigma\times X)}^\At$,
where $2=\sset{0,1}$ is the two-element set.
A \emph{$G$-coalgebra} is a pair $\X = \angl{X,\delta^\X}$ with
\emph{state space} $X$ and \emph{transition map} $\delta^\X\colon X\to G X$.
The outcomes $1$ and $0$ model immediate acceptance and rejection, respectively.
From each state $s \in X$, given an input $\alpha\in\At$, the coalgebra performs exactly one of three possible actions: it either produces an output $p\in\Sigma$ and moves to a new state $t$, halts and accepts, or halts and rejects;
that is, either $\delta^\X(s)(\alpha) = (p,t)$, or $\delta^\X(s)(\alpha) = 1$,
or $\delta^\X(s)(\alpha) = 0$.

A \emph{$G$-automaton} is a $G$-coalgebra with a designated start state
$\iota$, commonly denoted as a triple $\X = \angl{X, \delta^\X, \iota}$.
We can represent $G$-coalgebras graphically as in \Cref{fig:coalgebra-example}.
\begin{figure}
    \begin{tikzpicture}
    \node[smallstate,label=above:{$s_1$}] (x) {};
    \node[smallstate,label=above:{$s_2$},right=2cm of x] (x') {};
    \node[smallstate,label=above:{$s_3$},right=2cm of x'] (x'') {};

    \draw[->] (x) edge[output-edge] node[below=1mm] {$\alpha$} ($(x) + (0, -0.40)$);
    \draw[->] (x) edge[transition-edge] node[above] {$\beta/p$} (x');
    \draw[->] (x') edge[transition-edge] node[above] {$c/q$} (x'');
    \end{tikzpicture}
    \caption{%
        Graphical depiction of a $G$-coalgebra $\angl{X, \delta^\X}$.
        States are represented by dots, labeled with the name of that state whenever relevant.
        In this example, $\delta^\X(s_1)(\alpha) = 1$, and $\delta^\X(s_1)(\beta) = (p, s_2)$.
        When $\gamma \in \At$ such that $\delta^\X(s)(\gamma) = 0$, we draw no edge at all.
        We may abbreviate drawings by combining transitions with the same target into a Boolean expression; for instance, when $c = \alpha + \beta$,
        we have $\delta^\X(s_2)(\alpha) = \delta^\X(s_2)(\beta) = (q, s_3)$.
    }\label{fig:coalgebra-example}
\end{figure}
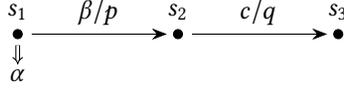

A $G$-coalgebra $\X=\angl{X,\delta^\X}$ can be viewed both as an acceptor
of finite guarded strings $\GS = \At \cdot {(\Sigma \cdot \At)}^*$, or as
an acceptor of finite \emph{and} infinite guarded strings
$\GS \cup \wGs$, where $\wGs \defeq {(\At
\cdot \Sigma)}^\omega$. Acceptance for a state $s$ is captured by the
following equivalences:
\begin{equation}
\label{def:accept}
\begin{aligned}
  \accept(s,\alpha) &~\iff~ \delta^\X(s)(\alpha) = 1\\
  \accept(s,\alpha p x) &~\iff~
    \exists t.\; \delta^\X(s)(\alpha) = (p,t) \land \accept(t,x)
\end{aligned}
\end{equation}
The language of finite guarded strings $\lang\X(s) \subseteq GS$ accepted
from state $s \in X$ is the \emph{least fixpoint} solution of the above system;
in other words, we interpret~\eqref{def:accept}  inductively.
The language of finite and infinite guarded strings
$\Lang\X(s) \subseteq \Gs \cup \wGs$
accepted from state $s$ is the \emph{greatest fixpoint} solution
of the above system;
in other words, we interpret~\eqref{def:accept} coinductively.%
\footnote{%
  The set $\FF$ of maps $F \colon X \to \powerset{\Gs \cup \wGs}$ ordered pointwise by
  subset inclusion forms a complete lattice.
  The monotone map
  \begin{align*}
    \tau\colon \FF \to \FF,\,
    \tau(F) = \lam{s \in X}{
      \set{\alpha\in \At}{\delta^\X(s)(\alpha) = 1}
      \cup \set{apx}{\exists t.\; \delta^\X(s)(\alpha) = (p,t) \land x \in F(t)}
    }
  \end{align*}
  arising from~\eqref{def:accept}
  has least and greatest fixpoints, $\lang\X$ and $\Lang\X$, by the
  Knaster-Tarksi theorem.
}
The two languages are related by the equation $\lang\X(s) = \Lang\X(s) \cap \Gs$.
Our focus will mostly be on the finite-string semantics,
$\lang\X(-)\colon X \to 2^\Gs$, since GKAT expressions denote finite-string
languages, $\den{-}\colon \Exp \to 2^\Gs$.


The language accepted by a $G$-automaton $\X = \angl{X, \delta^X, \iota}$ is the
language accepted by its initial state~$\iota$.  Just like the language model
for GKAT programs, the language semantics of a $G$-automaton satisfies the
determinacy property (see \Cref{def:action-det}). In fact, every language
that satisfies the determinacy property can be recognized by a $G$-automaton,
possibly with infinitely many states.
(We will prove this formally in \Cref{thm:final-for-normal}.)

\begin{figure}[t]
\def\arraystretch{1.2}
\centerline{\(
\begin{array}{c||c|c|c} 
  e & X_e & \delta_e \in X_e \to G X_e
          & \iota_e(\alpha) \in 2 + \Sigma \times X_e
          \\\hline 
  b & \emptyset & \emptyset     & [\alpha\leq b]\\
  p & \sset{*}  & * \mapsto \1  & (p, *)\\[0.5ex]
  f +_b g
    & X_f + X_g
    & \delta_f + \delta_g
    & \begin{cases}
        \iota_f(\alpha) &\alpha \leq b\\
        \iota_g(\alpha) &\alpha \leq \bneg{b}\\
      \end{cases}
    \\[1.2em]
  f\cdot g
    & X_f + X_g
    & (\delta_f + \delta_g)\umod{X_f,\iota_g}
    & \begin{cases}
        \iota_f(\alpha) &\iota_f(\alpha) \neq 1\\
        \iota_g(\alpha) &\iota_f(\alpha) = 1
      \end{cases}
    \\[1.2em]
  f^{(b)}
    & X_f
    & \delta_f\umod{X_f,\iota_e}
    & \begin{cases}
        1               &\alpha \leq \bneg{b}\\
        0               &\alpha\leq b, \iota_f(\alpha)=1\\
        \iota_f(\alpha) &\alpha\leq b, \iota_f(\alpha) \neq 1
      \end{cases}
\end{array}\)}
\caption{Construction of the Thompson coalgebra $\X_e = \angl{X_e, \delta_e}$ with initial
  pseudostate $\iota_e$.}%
\label{fig:thompson}
\end{figure}

\subsection{Expressions to Automata: a Thompson Construction}%
\label{sec:Thompson}

We translate expressions to $G$-coalgebras using a construction reminiscent of
Thompson's construction for regular expressions~\cite{thompson-1968}, where
automata are formed by induction on the structure of the expressions and
combined to reflect the various GKAT operations.

We first set some notation. A \emph{pseudostate} is an element $h \in G X$. We let $\1 \in G X$ denote the pseudostate $\1(\alpha) = 1$, \ie, the constant function returning $1$.
Let $\X=\angl{X, \delta}$ be a $G$-coalgebra. The \emph{uniform continuation} of
$Y \subseteq X$ by $h \in GX$ (in $\X$) is the coalgebra
 $\X\umod{Y,h} \defeq \angl{X, \delta\umod{Y,h}}$, where
\[
  \delta\umod{Y,h}(x)(\alpha) \defeq \begin{cases}
    h(\alpha)           &\text{if } x \in Y, \delta(x)(\alpha) = 1\\
    \delta(x)(\alpha)   &\text{otherwise.}
  \end{cases}
\]
Intuitively, uniform continuation replaces termination of states in a region $Y$
of $\X$ by a transition described by $h\in GX$; this construction will be useful
for modeling operations that perform some kind of sequencing.
\Cref{fig:continuation} schematically describes the uniform continuation
operation, illustrating different changes to the automaton that can occur as a
result; observe that since $h$ may have transitions into $Y$, uniform
continuation can introduce loops.

\begin{figure}
\begin{tikzpicture}
\draw[dotted] (0, -0.2) rectangle (4, 1.8);
\node at (0.3,1.5) {$Y$};

\node[smallstate] (y0) at (0.7,0.7) {};
\node[smallstate] (y1) at (1.9,0.7) {};
\node[smallstate] (y2) at (3.1,0.7) {};

\node[smallstate] (x0) at (0.7,2.3) {};
\node[smallstate] (x1) at (1.9,2.3) {};

\draw[->] (y0) edge[output-edge] node[below=1mm] {\small$\alpha$} ($(y0) + (0,-0.40)$);
\draw[->,gray!80] (y1) edge[output-edge] node[below=1mm] {\small$\beta$} ($(y1) + (0,-0.40)$);
\draw[->,gray!80] (y2) edge[output-edge] node[below=1mm] {\small$\gamma$} ($(y2) + (0,-0.40)$);
\draw[->] (x0) edge[output-edge] node[above=1mm] {\small$\beta$} ($(x0) + (0,0.40)$);

\draw[->,red] (y1) edge[transition-edge] node[right] {\small$\beta/q$} (x1);
\draw[->,red] (y2) edge[transition-edge,loop] node[above] {\small$\gamma/p$} (y2);
\draw[->] (y0) edge[transition-edge] node[above] {\small$\beta/p$} (y1);

\node[smallstate,blue,label={\color{blue}$h$}] (h) at (4.5,2) {};
\draw[->,blue] (h) edge[transition-edge,bend left,out=45,in=125] node[below right] {\small$\gamma/p$} (y2);
\draw[->,blue] (h) edge[transition-edge,bend right] node[above right] {\small$\beta/q$} (x1);
\draw[->,blue] (h) edge[output-edge] node[right=1mm] {\small$\alpha$} ($(h) + (0.40,0)$);
\end{tikzpicture}
\caption{Schematic explanation of the \emph{uniform continuation} $\X\umod{Y,h}$ of $\X$, where
$Y \subseteq X$ and $h \in GX$. The pseudostate $h$ and its transitions are drawn in blue. Transitions present in $\X$ unchanged by the extension are drawn in black; grayed out transitions are replaced by transitions drawn in red as a result of the extension.}\label{fig:continuation}
\end{figure}
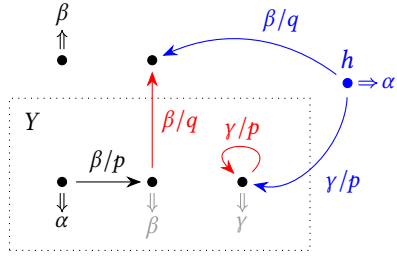
We will also need coproducts to combine coalgebras. Intuitively, the coproduct of two coalgebras is just the juxtaposition of both coalgebras. Formally, for
$\X=\angl{X, \delta_1}$ and $\Y=\angl{Y, \delta_2}$,
we write the coproduct as $\X+\Y=\angl{X+Y, \delta_1+\delta_2}$,
where $X + Y$ is the disjoint union of $X$ and $Y$,
and $\delta_1 + \delta_2\colon X+Y \to G(X+Y)$ is the map that applies
$\delta_1$ to states in $X$ and $\delta_2$ to states in $Y$.

\Cref{fig:thompson} presents our translation from expressions $e$ to coalgebras
$\X_e$ using coproducts and uniform continuations, and
\Cref{fig:thompson-schematic} sketches the transformations used to construct the
automaton of a term from its subterms.  We model initial states as
pseudostates, rather than proper states. This trick avoids the
$\epsilon$-transitions that appear in the classical Thompson construction and
yields compact, linear-size automata.
\Cref{fig:thompson-examples} depicts some examples of our construction.

To turn the resulting coalgebra into an automaton, we simply convert the initial
pseudostate into a proper state. Formally, when $\X_e = \angl{X_e, \delta_e}$, we
write $\X_e^\iota$ for the $G$-automaton $\angl{\{\iota\} + X_e,
\delta^\iota_e, \iota}$, where for $x \in X_e$, we set $\delta^\iota_e(x) =
\delta_e(x)$ as well as $\delta_e^\iota(\iota) = \iota_e$.
We call $\X_e$ and $\X_e^\iota$ the \emph{Thompson coalgebra} and
\emph{Thompson automaton} for $e$, respectively.

\noindent
The construction translates expressions to equivalent automata in
the following sense:
\begin{theorem}[Correctness I]%
\label{thm:thompson-correct}
The Thompson automaton for $e$ recognizes $\den{e}$, that is
$\lang{\X_e^\iota}(\iota) = \den{e}$.
\end{theorem}
\begin{proof}[Proof Sketch]
This is a direct corollary of
\Cref{prop:sol-preserves-lang,thm:thompson-roundtrip}, to follow.
\end{proof}
\noindent
Moreover, the construction is efficiently implementable and yields
small automata:
\begin{proposition}%
\label{prop:thompson-efficient}
The Thompson automaton for $e$ is effectively constructible in time
$\OO(|e|)$ and has $\#_\Sigma(e) + 1$ (thus, $\OO(|e|)$) states,
where $|\At|$ is considered a constant for the time complexity claim,
$|e|$ denotes the size of the expression,
and $\#_\Sigma(e)$ denotes the number of occurrences of actions in $e$.
\end{proposition}

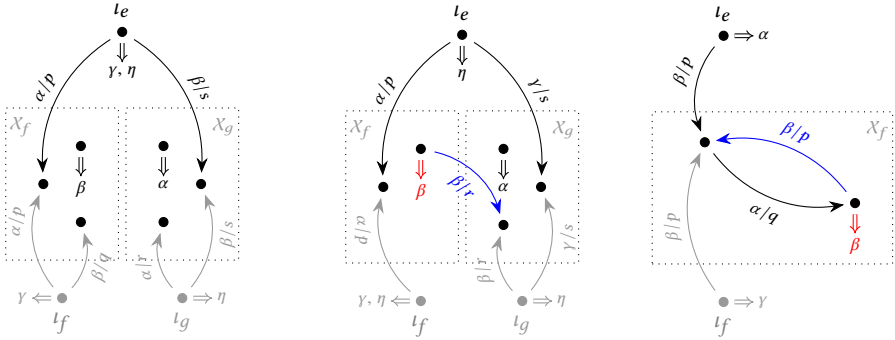
\begin{figure}[t!]
    \begin{subfigure}[b]{0.38\textwidth}
        \centering
        \begin{tikzpicture}
            \draw[dotted] (0,0) rectangle (1.5,2);
            \draw[dotted] (1.6,0) rectangle (3.1,2);
            \node[gray!80] at (0.2, 1.75) (Xf) {\tiny$\X_f$};
            \node[gray!80] at (2.9, 1.75) (Xg) {\tiny$\X_g$};

            \node[smallstate] at (0.5,1.0) (sf1) {};
            \node[smallstate] at (1.0,0.5) (sf2) {};
            \node[smallstate] at (1.0,1.5) (sf3) {};
            \node[smallstate,gray!80,label=below:{\small\color{gray}$\iota_f$}] at (0.75, -0.5) (if) {};

            \draw[->] (if) edge[transition-edge, bend left,gray!80,pos=0.7] node[above,sloped] {\tiny$\alpha/p$} (sf1);
            \draw[->] (if) edge[transition-edge, bend right,gray!80] node[below,sloped] {\tiny$\beta/q$} (sf2);
            \draw[->] (if) edge[output-edge, gray!80] node[left=1mm] {\tiny$\gamma$} ($(if) - (0.40, 0)$);
            \draw[->] (sf3) edge[output-edge] node[below=1mm] {\tiny$\beta$} ($(sf3) + (0, -0.40)$);

            \node[smallstate] at (2.6,1.0) (sg1) {};
            \node[smallstate] at (2.1,0.5) (sg2) {};
            \node[smallstate] at (2.1,1.5) (sg3) {};
            \node[smallstate,gray!80,label=below:{\small\color{gray!80}$\iota_g$}] at (2.35, -0.5) (ig) {};

            \draw[->] (ig) edge[transition-edge, bend left,gray!80] node[above,sloped,rotate=180] {\tiny$\alpha/r$} (sg2);
            \draw[->] (ig) edge[transition-edge, bend right,gray!80,pos=0.6] node[below,sloped] {\tiny$\beta/s$} (sg1);
            \draw[->] (ig) edge[output-edge, gray!80] node[right=1mm] {\tiny$\eta$} ($(ig) + (0.40, 0)$);
            \draw[->] (sg3) edge[output-edge] node[below=1mm] {\tiny$\alpha$} ($(sg3) + (0, -0.40)$);

            \node[smallstate,label=above:{\small$\iota_e$}] at (1.55, 3) (i) {};
            \draw[->] (i) edge[transition-edge, bend right] node[sloped,above] {\tiny$\alpha/p$} (sf1);
            \draw[->] (i) edge[transition-edge, bend left] node[sloped,above] {\tiny$\beta/s$} (sg1);
            \draw[->] (i) edge[output-edge] node[below=1mm] {\tiny$\gamma,\eta$} ($(i) + (0, -0.40)$);
        \end{tikzpicture}
        \caption{$e = f +_b g$, with $\alpha,\gamma \leq b$ and $\beta,\eta \leq \bneg{b}$.}
    \end{subfigure}
    \begin{subfigure}[b]{0.26\textwidth}
        \centering
        \begin{tikzpicture}
            \draw[dotted] (0,0) rectangle (1.5,2);
            \draw[dotted] (1.6,0) rectangle (3.1,2);
            \node[gray!80] at (0.2, 1.75) (Xf) {\tiny$\X_f$};
            \node[gray!80] at (2.9, 1.75) (Xg) {\tiny$\X_g$};

            \node[smallstate] at (0.5,1.0) (sf1) {};
            \node[smallstate] at (1.0,1.5) (sf3) {};
            \node[smallstate,gray!80,label=below:{\small\color{gray}$\iota_f$}] at (1, -0.5) (if) {};

            \draw[->] (if) edge[transition-edge, bend left,gray!80,pos=0.7] node[below,sloped] {\tiny$\alpha/p$} (sf1);
            \draw[->] (if) edge[output-edge, gray!80] node[left=1mm] {\tiny$\gamma,\eta$} ($(if) - (0.40, 0)$);
            \draw[->] (sf3) edge[output-edge,red] node[below=1mm] {\color{red}\tiny$\beta$} ($(sf3) + (0, -0.40)$);

            \node[smallstate] at (2.6,1.0) (sg1) {};
            \node[smallstate] at (2.1,0.5) (sg2) {};
            \node[smallstate] at (2.1,1.5) (sg3) {};
            \node[smallstate,gray!80,label=below:{\small\color{gray!80}$\iota_g$}] at (2.35, -0.5) (ig) {};

            \draw[->] (ig) edge[transition-edge, bend left,gray!80] node[above,sloped,rotate=180] {\tiny$\beta/r$} (sg2);
            \draw[->] (ig) edge[transition-edge, bend right,gray!80,pos=0.6] node[below,sloped] {\tiny$\gamma/s$} (sg1);
            \draw[->] (ig) edge[output-edge, gray!80] node[right=1mm] {\tiny$\eta$} ($(ig) + (0.40, 0)$);
            \draw[->] (sg3) edge[output-edge] node[below=1mm] {\tiny$\alpha$} ($(sg3) + (0, -0.40)$);

            \node[smallstate,label=above:{\small$\iota_e$}] at (1.55, 3) (i) {};
            \draw[->] (i) edge[transition-edge, bend right] node[sloped,above] {\tiny$\alpha/p$} (sf1);
            \draw[->] (i) edge[transition-edge, bend left] node[sloped,above] {\tiny$\gamma/s$} (sg1);
            \draw[->] (i) edge[output-edge] node[below=1mm] {\tiny$\eta$} ($(i) + (0, -0.40)$);

            \draw[->] (sf3) edge[transition-edge,blue,bend left] node[sloped,below] {\color{blue}\tiny$\beta/r$} (sg2);
        \end{tikzpicture}
        \caption{$e = f \cdot g$}
    \end{subfigure}
    \begin{subfigure}[b]{0.33\textwidth}
        \centering
        \begin{tikzpicture}
            \draw[dotted] (-0.2,0) rectangle (3.0,2);
            \node[gray!80] at (2.8, 1.75) (Xf) {\tiny$\X_f$};

            \node[smallstate] at (0.5,1.6) (sf1) {};
            \node[smallstate] at (2.5,0.8) (sf3) {};
            \node[smallstate,gray!80,label=below:{\small\color{gray}$\iota_f$}] at (0.75, -0.5) (if) {};

            \draw[->] (if) edge[transition-edge, bend left,gray!80,pos=0.5] node[above,sloped,rotate=180] {\tiny$\beta/p$} (sf1);
            \draw[->] (if) edge[output-edge, gray!80] node[right=1mm] {\tiny$\gamma$} ($(if) + (0.40, 0)$);
            \draw[->] (sf3) edge[output-edge,red] node[below=1mm] {\color{red}\tiny$\beta$} ($(sf3) + (0, -0.40)$);

            \node[smallstate,label=above:{\small$\iota_e$}] at (0.75, 3) (i) {};
            \draw[->] (i) edge[transition-edge, bend right,pos=0.4] node[sloped,above] {\tiny$\beta/p$} (sf1);
            \draw[->] (i) edge[output-edge] node[right=1mm] {\tiny$\alpha$} ($(i) + (0.40, 0)$);
            \draw[->] (sf3) edge[transition-edge,blue,bend right,pos=0.5] node[sloped,above] {\tiny$\beta/p$} (sf1);
            \draw[->] (sf1) edge[transition-edge,bend right,pos=0.5] node[sloped,below] {\tiny$\alpha/q$} (sf3);
        \end{tikzpicture}
        \vspace{-0.5mm} 
        \caption{$e = f^{(b)}$, with $\beta,\gamma \leq b$ and $\alpha \leq \bneg{b}$}
    \end{subfigure}
    \caption{Schematic depiction of the Thompson construction for guarded union, sequencing and guarded loop operators. The initial psuedostates of the automata for $f$ and $g$ are depicted in gray. Transitions in red are present in the automata for $f$ and $g$, but overridden by a uniform extension with the transitions in blue.}\label{fig:thompson-schematic}
\end{figure}

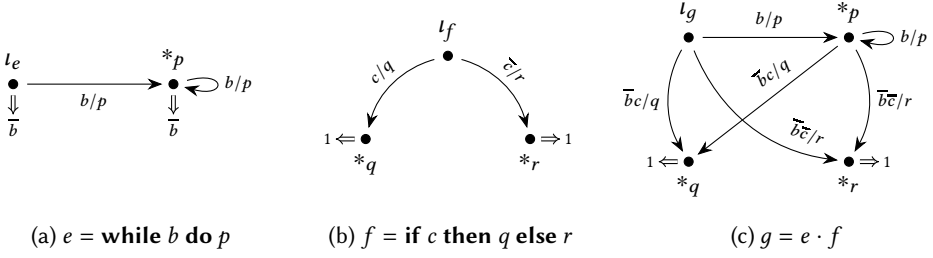
\begin{figure}
    \begin{subfigure}[b]{0.30\textwidth}
        \centering
        \begin{tikzpicture}
        \node[smallstate,label={$\iota_e$}] (ie1) {};
        \node[smallstate,label=above:{$*_p$},right=2cm of ie1] (se1) {};

        \draw[->] (ie1) edge[transition-edge] node[below] {\tiny$b/p$} (se1);
        \draw[->] (ie1) edge[output-edge] node[below=1mm] {\tiny$\bneg{b}$} ($(ie1) + (0, -0.40)$);
        \draw[->] (se1) edge[output-edge] node[below=1mm] {\tiny$\bneg{b}$} ($(se1) + (0, -0.40)$);
        \draw[->] (se1) edge[transition-edge,loop right,looseness=50] node[right] {\tiny$b/p$} (se1);
        \end{tikzpicture}
        \vspace{8mm}
        \caption{$e = \kw{while}\ b\ \kw{do}\ p$}
    \end{subfigure}
    \begin{subfigure}[b]{0.30\textwidth}
        \centering
        \begin{tikzpicture}
        \node[smallstate,label={$\iota_f$}] (if1) {};
        \node[smallstate,label=below:{$*_q$},below left=of if1] (sq1) {};
        \node[smallstate,label=below:{$*_r$},below right=of if1] (sr1) {};

        \draw[->] (if1) edge[transition-edge,bend right] node[above,sloped] {\tiny$c/q$} (sq1);
        \draw[->] (if1) edge[transition-edge,bend left] node[above,sloped] {\tiny$\bneg{c}/r$} (sr1);
        \draw[->] (sq1) edge[output-edge] node[left=1mm] {\tiny$1$} ($(sq1) + (-0.40, 0)$);
        \draw[->] (sr1) edge[output-edge] node[right=1mm] {\tiny$1$} ($(sr1) + (0.40, 0)$);
        \end{tikzpicture}
        \vspace{3mm}
        \caption{$f = \kw{if}\ c\ \kw{then}\ q\ \kw{else}\ r$}
    \end{subfigure}
    \begin{subfigure}[b]{0.33\textwidth}
        \begin{tikzpicture}
        \node[smallstate,label={$\iota_g$}] (ief2) {};
        \node[smallstate,label=above:{$*_p$},right=2cm of ief2] (sp2) {};

        \node[smallstate,label=below:{$*_q$},below=15mm of ief2] (sq2) {};
        \node[smallstate,label=below:{$*_r$},below=15mm of sp2] (sr2) {};

        \draw[->] (ief2) edge[transition-edge] node[above] {\tiny$b/p$} (sp2);
        \draw[->] (sp2) edge[transition-edge,loop right,looseness=50] node[right] {\tiny$b/p$} (sp2);
        \draw[->] (ief2) edge[transition-edge, bend right] node[left] {\tiny$\bneg{b}c/q$} (sq2);
        \draw[->] (sp2) edge[transition-edge, bend left] node[right] {\tiny$\bneg{b}\bneg{c}/r$} (sr2);
        \draw[->] (ief2) edge[transition-edge, bend right] node[above,sloped,pos=0.8] {\tiny$\bneg{b}\bneg{c}/r$} (sr2);
        \draw[->] (sp2) edge[transition-edge] node[above,sloped,pos=0.4] {\tiny$\bneg{b}c/q$} (sq2);

        \draw[->] (sq2) edge[output-edge] node[left=1mm] {\tiny$1$} ($(sq2) + (-0.40, 0)$);
        \draw[->] (sr2) edge[output-edge] node[right=1mm] {\tiny$1$} ($(sr2) + (0.40, 0)$);
        \end{tikzpicture}
        \caption{$g = e \cdot f$}
    \end{subfigure}
    \caption{Concrete construction of an automaton using the Thompson construction. First, we construct an automaton for $e$, then an automaton for $f$, and finally we combine these into an automaton for $g$. In these examples, $p,q,r$ are single action letters, not arbitrary expressions.}\label{fig:thompson-examples}
\end{figure}

\subsection{Automata to Expressions: Solving Linear Systems}
The previous construction shows that every GKAT expression can be translated to
an equivalent G-automaton. In this section we consider the reverse direction,
from G-automata to GKAT expressions. The main idea is to interpret the
coalgebra structure as a system of equations, with one
variable and equation per state,
and show that there are GKAT expressions solving the system,
modulo equivalence; this idea goes back to \citet{conway-1971} and
\citet{backhouse-1975}. Not all systems arising from G-coalgebras have
a solution, and so not all G-coalgebras can be captured by GKAT expressions.
However, we identify a subclass of G-coalgebras that can be represented as GKAT
terms.
%
By showing that this class contains the coalgebras produced by our
expressions-to-automata translation, we obtain an equivalence between GKAT
expressions and coalgebras in this class.

We start by defining when a map assigning expressions to coalgebra states
is a solution.

\begin{definition}[Solution]\label{def:solution}
Let $\X = \angl{X, \delta^\X}$ be a $G$-coalgebra.
  We say that $s\colon X \to \Exp$ is a \emph{solution} to $\X$ if for all $x \in X$ it holds that
\[
    s(x) \equiv \dsum_{\alpha \leq 1} \floor{\delta^\X(x)(\alpha)}_s
    \qquad
    \text{where}
    \quad
    \floor{0}_s \defeq 0
    \quad
    \floor{1}_s \defeq 1
    \quad
    \floor{\angl{p, x}}_s \defeq p \cdot s(x)
\]
\end{definition}
\begin{example}
Consider the Thompson automata in \Cref{fig:thompson-examples}.
\begin{enumerate}[label={(\alph*)}]
    \item
    Solving the first automaton requires, by \Cref{def:solution}, finding
    an expression $s_e(*_p)$ such that $s_e(*_p) \equiv p \cdot s_e(*_p) +_b 1$.
    By~(\nameref{ax:unroll}), we know that $s_e(*_p) = p^{(b)}$ is valid;
    in fact,~(\nameref{ax:fixpoint}) tells us that this choice of $x$ is the \emph{only} valid solution up to GKAT-equivalence.
    If we include $\iota_e$ as a state, we can choose $s_e(\iota_e) = p^{(b)}$ as well.

    \item
    The second automaton has an easy solution: both $*_q$ and $*_r$ are solved by setting $s_f(*_q) = s_f(*_r) = 1$.
    If we include $\iota_f$ as a state, we can choose $s_f(\iota_f) = q \cdot s_f(*_q) +_b r \cdot s_f(*_r) \equiv q +_b r$.

    \item
    The third automaton was constructed from the first two;
    similarly, we can construct its solution from the solutions to the first two.
    We set $s_g(*_p) = s_e(*_p) \cdot s_f(\iota_f)$, and $s_g(*_q) = s_f(*_q)$,
    and  $s_g(*_r) = s_f(*_r)$.
    If we include $\iota_g$ as a state, we can choose
    $s_g(\iota_g) = s_e(\iota_e) \cdot s_f(\iota_f)$.
\end{enumerate}
\end{example}

\noindent
Solutions are language-preserving maps from states to  expressions
in the following sense:

\begin{restatable}{proposition}{solpreserveslang}%
\label{prop:sol-preserves-lang}
If $s$ solves $\X$ and $x$ is a state, then $\den{s(x)} = \lang\X(x)$.
\end{restatable}
\begin{proof}[Proof Sketch]
Show that $w \in \den{s(x)} \Iff w \in \lang\X(x)$ by induction on the
length of $w \in \Gs$.
\end{proof}

We would like to build solutions for $G$-coalgebras, but \citet{KT08a} showed
that this is not possible in general: there is a 3-state $G$-coalgebra that does
not correspond to any $\kw{while}$~program, but instead can only be modeled by a
program with multi-level breaks. In order to obtain an exact correspondence to
GKAT programs, we first identify a sufficient condition for $G$-coalgebras to
permit solutions, and then show that the Thompson coalgebra defined previously
meets this condition.

\begin{definition}[Well-nested Coalgebra]%
\label{def:simple}
Let $\X = \angl{X, \delta^\X}$ and $\Y = \angl{Y, \delta^\Y}$ range
over  $G$-coalgebras.
The collection of \emph{well-nested} coalgebras is inductively defined as follows:
\begin{enumerate}[label={(\roman*)},ref={\roman*}]
    \item\label{rule:s1}
    If $\X$ has no transitions, \ie, if  $\delta^\X \in X \to 2^\At$, then $\X$ is well-nested.
    \item\label{rule:s2}
    If $\X$ and $\Y$ are well-nested and $h \in G(X + Y)$,
    then $(\X + \Y)\umod{X,h}$ is well-nested.
\qedhere
\end{enumerate}
\end{definition}

\noindent
We are now ready to construct solutions to well-nested coalgebras.

\begin{restatable}[Existence of Solutions]{theorem}{solutionsexist}%
\label{thm:simple-solvable}
Any well-nested coalgebra admits a solution.
\end{restatable}
\begin{proof}[Proof Sketch]
Assume $\X$ is well-nested.
We proceed by rule induction on the well-nestedness derivation.
\begin{enumerate}
  \item[(\ref{rule:s1})]
  Suppose $\delta^\X \colon X \to 2^\At$. Then
  \[
      s^{\X}(x) \defeq \sum \set{ \alpha \in \At }{ \delta^\X(x)(\alpha) = 1 }
  \]
  is a solution to $\X$.

  \item[(\ref{rule:s2})]
  Let $\Y = \angl{Y, \delta^\Y}$ and $\Z = \angl{Z, \delta^\Z}$ be well-nested $G$-coalgebras, and let $h \in G(Y + Z)$ be such that $\X = (\Y + \Z)\umod{Y,h}$.
  By induction, $\Y$ and $\Z$ admit solutions $s^\Y$ and $s^\Z$ respectively;
  we need to find a solution $s^\X$ to $X = Y + Z$.
  The idea is to retain the solution that we had for states in $\Z$---whose
  behavior has not changed under uniform continuation---while modifying the
  solution to states in $\Y$ in order to account for transitions from $h$.
  To this end, we choose the following expressions:
  \begin{mathpar}
    b \defeq \sum \set{\alpha \in \At}{h(\alpha) \in \Sigma \times X}
    \and
    \ell \defeq {\Bigl( \dsum_{\alpha \leq b} \floor{h(\alpha)}_{s^\Y} \Bigr)}^{(b)} \cdot \dsum_{\alpha \leq \bneg{b}} \floor{h(\alpha)}_{s^\Z}
  \end{mathpar}
  We can then define $s$ by setting $s(x) = s^\Y(x) \cdot \ell$ for $x \in Y$, and $s(x) = s^\Z(x)$ for $x \in Z$.
  A detailed argument showing that $s$ is a solution can be found in
  the \extended{appendix}{extended version of the paper~\cite{gkat-full}}.
  \qedhere
  \end{enumerate}
 \end{proof}

\noindent
As it turns out, we can do a round-trip, showing that the solution to the (initial state of the) Thompson automaton for an expression is equivalent to the original expression.
\begin{restatable}[Correctness II]{theorem}{thompsonroundtrip}%
\label{thm:thompson-roundtrip}
Let $e \in \Exp$. Then $\X_e^\iota$ admits a solution $s$ such that $e \equiv s(\iota)$.
\end{restatable}

\noindent
Finally, we show that the automata construction of the previous section gives well-nested automata.

\begin{theorem}[well-nestedness of Thompson construction]%
\label{thm:thompson-simple}
$\X_e$ and $\X_e^\iota$ are well-nested for all expressions $e$.
\end{theorem}
\begin{proof}
We proceed by induction on $e$.
In the base, let $\Z=\angl{\emptyset, \emptyset}$ and $\I=\angl{\sset{*}, * \mapsto \bf{1}}$ denote the coalgebras with no states and with a single all-accepting state, respectively. Note that $\Z$ and $\I$ are well-nested, and that for $b \in \Bexp$ and $p \in \Sigma$ we have $\X_b = \Z$ and $\X_p = \I$.

All of the operations used to build $\X_e$, as detailed in \Cref{fig:thompson},
can be phrased in terms of an appropriate uniform continuation of a coproduct;
for instance, when $e = f^{(b)}$ we have that $\X_e = (\X_f + \I)\umod{X_f,\iota_e}$.
Consequently, the Thompson automaton $\X_e$ is well-nested by construction.
Finally, observe that $\X_e^\iota = (\I + \X_e)\umod{\sset{*},\iota_e}$; hence, $\X_e^\iota$ is well-nested as well.
\end{proof}

\noindent
\Cref{thm:thompson-correct,thm:simple-solvable,thm:thompson-simple} now give us the desired Kleene theorem.

\begin{corollary}[Kleene Theorem]
Let $L\subseteq \GS$. The following are equivalent:
\begin{enumerate}
\item $L = \sem{e}$ for a GKAT expression $e$.
\item $L = \lang\X(\iota)$ for a well-nested, finite-state $G$-automaton $\X$ with initial state $\iota$.
\end{enumerate}
\end{corollary}

\section{Decision procedure}%
\label{sec:decision}

We saw in the last section that GKAT expressions can be efficiently
converted to equivalent automata with a linear number of states.
Equivalence of automata can be established algorithmically, supporting a
decision procedure for GKAT that is significantly more efficient than decision
procedures for KAT\@. In this section, we describe our algorithm.

First, we define bisimilarity of automata states in the usual way~\cite{KT08a}.
\begin{definition}[Bisimilarity]%
\label{def:bisimilarity}
Let $\X$ and $\Y$ be $G$-coalgebras. A \emph{bisimulation} between $\X$ and $\Y$
is a binary relation ${R}\subs X\times Y$ such that if
$x \mathrel{R} y$, then the following implications hold:
\begin{enumerate}[(i)]
\item if $\delta^\X(x)(\alpha)\in 2$, then $\delta^\Y(y)(\alpha) = \delta^\X(x)(\alpha)$; and
\item if $\delta^\X(x)(\alpha) = (p,x')$, then $\delta^\Y(y)(\alpha) = (p,y')$
  and $x' \mathrel{R} y'$ for some $y'$.
\end{enumerate}
States $x$ and $y$ are called \emph{bisimilar}, denoted
$x\sim y$, if there exists a bisimulation relating $x$ and $y$.
\end{definition}

\noindent
As usual, we would like to reduce automata equivalence to bisimilarity.
It is easy to see that bisimilar states recognize the same language.
\begin{lemma}\label{lem:bisim-equiv}
If $\X$ and $\Y$ are G-coalgebras with bisimilar states $x \sim y$, then
$\lang\X(x) = \lang\Y(y)$.
\end{lemma}
\begin{proof}
We verify that $w \in \lang\X(x) \Iff w \in \lang\Y(y)$ by
induction on the length of $w \in \Gs$:
\begin{itemize}[left=1ex..1.4\parindent]
  \item For $\alpha \in \Gs$, we have $
    \alpha \in \lang\X(x)
    \Iff \delta^\X(x)(\alpha) = 1
    \Iff \delta^\Y(y)(\alpha) = 1
    \Iff \alpha \in \lang\Y(y)$.
  \item For $\alpha p w \in \Gs$, we use bisimilarity and the
    induction hypothesis to derive
    \begin{align*}
      \alpha p w \in \lang\X(x)
      &\iff \exists x'.\, \delta^\X(x)(\alpha) = (p, x') \land w \in \lang\X(x')\\
      &\iff \exists y'.\, \delta^\Y(y)(\alpha) = (p, y') \land w \in \lang\Y(y')
      \iff \alpha p w \in \lang\Y(y).
    \qedhere
    \end{align*}
\end{itemize}
\end{proof}

\noindent
The converse direction, however, does not hold for $G$-coalgebras in general.
To see the problem, consider the following automaton, where $\alpha \in \At$
is an atom and $p \in \Sigma$ is an action:
\[
\begin{tikzpicture}[node distance=2cm, baseline=-1ex]
  \node (s1) [smallstate,label=above:{$s_1$}] {};
  \node (s2) [smallstate,label=above:{$s_2$},right of=s1] {};

  \draw (s1) edge[transition-edge] node[above] {\small$\alpha/p$} (s2);
\end{tikzpicture}
\]
Both states recognize the empty language, that is
\ie, $\lang{}(s_1) = \lang{}(s_2) = \emptyset$; but $s_2$ rejects
immediately, whereas $s_1$ may first take a transition.
As a result, $s_1$ and $s_2$ are not bisimilar.
Intuitively, the language accepted by a state does not distinguish
between early and late rejection, whereas bisimilarity does.
We solve this by disallowing late rejection, \ie, transitions that can
never lead to an accepting state; we call coalgebras that respect this
restriction \emph{normal}.

\subsection{Normal Coalgebras}%
\label{sec:normal}

We classify states and coalgebras as follows.
\begin{definition}[Live, Dead, Normal]
Let $\X = \angl{X,\delta^\X}$ denote a $G$-coalgebra.
A state $s \in X$ is \emph{accepting} if $\delta^\X(s)(\alpha)=1$
for some $\alpha \in \At$. A state is \emph{live} if it can transition
to  an accepting state one or more steps, or \emph{dead} otherwise.
A coalgebra is \emph{normal} if it has no transitions to dead states.
\end{definition}
\begin{remark}
Note that, equivalently, a state is live iff $\lang\X(s) \neq \emptyset$ and
dead iff $\lang\X(s) = \emptyset$.
Dead states can exist in a normal coalgebra, but they must immediately reject
all $\alpha\in\At$, since any successor of a dead state would also be dead.
\end{remark}

\begin{example}
  Consider the following automaton.


  
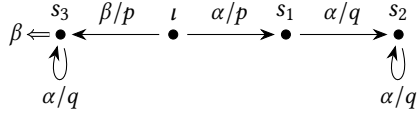
\[
    \begin{tikzpicture}[node distance=1.5cm]
      \node (i)   [smallstate,label=above:{$\iota$}] {};
      \node (s1)  [smallstate,right of=i,label=above:{$s_1$}] {};
      \node (s2)  [smallstate,right of=s1,label=above:{$s_2$}] {};
      \node (s3)  [smallstate,left of=i,label=above:{$s_3$}] {};

      \draw
        (i)  edge[transition-edge] node[above] {\small$\alpha/p$} (s1)
        (i)  edge[transition-edge,above] node[above] {\small$\beta/p$} (s3)
        (s1) edge[transition-edge] node[above] {\small$\alpha/q$} (s2)
        (s2) edge[transition-edge,loop below, looseness=50] node[below] {\small$\alpha/q$} (s2)
        (s3) edge[transition-edge,loop below, looseness=50] node[below] {\small$\alpha/q$} (s3)
        (s3) edge[output-edge] node[left=1.2mm] {\small$\beta$} ($(s3) + (-0.45, 0)$);
    \end{tikzpicture}
  \]
  The state $s_3$ is accepting.
  The states $\iota$ and $s_3$ are live, since they can reach an accepting state.
  The states $s_1$ and $s_2$ are dead, since they can only reach non-accepting
  states.  The automaton is not normal, since it contains
  the transitions $\iota \trans{\alpha/p} s_1$,
  $s_1\trans{\alpha/q} s_2$, and $s_2\trans{\alpha/q} s_2$ to
  dead states $s_1$ and $s_2$.
  We can \emph{normalize} the automaton by removing these transitions:


  
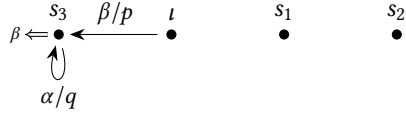
\[
    \begin{tikzpicture}[node distance=1.5cm]
      \node (i)   [smallstate,label=above:{$\iota$}] {};
      \node (s1)  [smallstate,right of=i,label=above:{$s_1$}] {};
      \node (s2)  [smallstate,right of=s1,label=above:{$s_2$}] {};
      \node (s3)  [smallstate,left of=i,label=above:{$s_3$}] {};

      \draw
        (i)  edge[transition-edge,above] node[above] {\small$\beta/p$} (s3)
        (s3) edge[transition-edge,loop below, looseness=50] node[below] {\small$\alpha/q$} (s3)
        (s3) edge[output-edge] node[left=1.2mm] {\tiny$\beta$} ($(s3) + (-0.45, 0)$);
    \end{tikzpicture}
  \]
The resulting automaton is normal: the dead states $s_1$ and $s_2$ reject
all $\alpha \in \At$ immediately.
\qed%
\end{example}

The example shows how $G$-coalgebra can be normalized.
Formally, let $\X = \angl{X, \delta}$ denote a coalgebra
with dead states $D \subseteq X$. We define the normalized
coalgebra $\hat{\X} \defeq \angl{X, \hat{\delta}}$ as follows: \[
  \hat{\delta}(s)(\alpha) \defeq \begin{cases}
    0                 &\text{if } \delta(s)(\alpha) \in \Sigma \times D\\
    \delta(s)(\alpha) &\text{otherwise.}
  \end{cases}
\]

\begin{lemma}[Correctness of normalization]%
\label{lem:normalization-correct}
Let $\X$ be a $G$-coalgebra. Then the following holds:
\begin{enumerate}[(i)]
\item $\X$ and $\hat{\X}$ have the same solutions: that is, $s: X \to \Exp$ solves $\X$ if and only if $s$ solves $\hat{\X}$; and
\item $\X$ and $\hat{\X}$ accept the same languages: that is, $\lang\X = \lang{\hat{\X}}$; and
\item $\hat{\X}$ is normal.
\end{enumerate}
\end{lemma}
\begin{proof}
For the first claim, suppose $s$ solves $\X$.
It suffices (\extended{by \Cref{lem:solution-alt}}
  {see the extended version~\cite{gkat-full}})
to show that for $x \in X$ and $\alpha \in \At$ we have $\alpha \cdot s(x) \equiv \alpha \cdot \floor{\delta^{\hat{\X}}(x)(\alpha)}_s$.
We have two cases.
\begin{itemize}[left=1ex..1.4\parindent]
    \item
    If $\delta^\X(x)(\alpha) = (p, x')$ with $x'$ dead, then by \Cref{prop:sol-preserves-lang} we know that $\sem{s(x')} = \ell^\X(x') = \emptyset$.
    By \Cref{cor:partial-completeness}, it follows that $s(x') \equiv 0$.
    Recalling that $\delta^{\hat{\X}}(x)(\alpha) = 0$ by construction,
    \[
        \alpha \cdot s(x)
            \equiv \alpha \cdot \floor{\delta^\X(x)(\alpha)}_s
            \equiv \alpha \cdot p \cdot s(x')
            \equiv \alpha \cdot 0
            \equiv \alpha \cdot \floor{\delta^{\hat{\X}}(x)(\alpha)}_s
    \]

    \item
    Otherwise, we know that $\delta^{\hat{\X}}(x)(\alpha) = \delta^\X(x)(\alpha)$, and thus
    \[
        \alpha \cdot s(x) \equiv \alpha \cdot \floor{\delta^\X(x)(\alpha)}_s \equiv \alpha \cdot \floor{\delta^{\hat{\X}}(x)(\alpha)}_s
    \]
\end{itemize}
The other direction of the first claim can be shown analogously.

For the second claim, we can establish $x \in \lang\X(s) \Iff x \in \lang{\hat{\X}}(s)$ for all states $s$ by a straightforward induction on the length of $x \in \Gs$, using that dead states accept the empty language.

For the third claim, we note that the dead states of $\X$ and $\hat{\X}$ coincide
by claim two; thus $\hat{\X}$ has no transition to dead states by construction.
\end{proof}

%

\subsection{Bisimilarity for Normal Coalgebras}%
\label{sec:bisimNormal}

We would like to show that, for normal coalgebras, states are bisimilar
if and only if they accept the same language. This will allow us to reduce
language-equivalence to bisimilarity, which is easy to establish
algorithmically. We need to take a slight detour.

Recall the determinacy property satisfied by GKAT languages
(\Cref{def:action-det}): a language $L \subseteq \Gs$ is deterministic
if, whenever strings $x,y \in L$ agree on the first $n$ atoms, they also agree on
the first $n$ actions (or absence thereof). Now, let $\L \subseteq 2^\Gs$
denote the set of deterministic languages. $\L$ carries a coalgebra
structure $\angl{\L, \delta^\L}$ whose transition map $\delta^\L$ is
the semantic Brzozowski derivative:
\begin{align*}
\delta^\LL(L)(\alpha) \defeq \begin{cases}
(p,\set{x \in \Gs}{\alpha p x \in L}) & \text{if $\set{x \in \Gs}{\alpha p x \in L}\ne\emptyset$}\\
1 & \text{if $\alpha\in L$}\\
0 & \text{otherwise.}
\end{cases}
\end{align*}
Note that the map is well-defined by determinacy: precisely one of the three cases holds.

Next, we define \emph{structure-preserving maps} between $G$-coalgebras
in the usual way:
\begin{definition}[Homomorphism]
A homomorphism between G-coalgebras $\X$ and $\Y$ is a map $h\colon X \to Y$
from states of $\X$ to states of $\Y$ that respects the transition
structures in the following sense: \[
  \delta^\Y(h(x)) = (G h)(\delta^\X(x)).
\]
More concretely, for all $\alpha \in \At$, $p \in \Sigma$, and $x,x' \in X$,
\begin{enumerate}[(i)]
\item if $\delta^\X(x)(\alpha)\in 2$, then $\delta^\Y(h(x))(\alpha) = \delta^X(x)(\alpha)$; and
\item if $\delta^\X(x)(\alpha) = (p,x')$, then $\delta^\Y(h(x))(\alpha) = (p,h(x'))$.
\qed%
\end{enumerate}
\end{definition}

\noindent
We can now show that the acceptance map $\lang\X \colon X \to 2^\Gs$ is
structure-preserving in the above sense. Moreover, it is the \emph{only}
structure-preserving map from states to deterministic languages:

\begin{restatable}{theorem}{finalfornormal}%
\label{thm:final-for-normal}%
If $\X$ is normal, then $\lang\X \colon X \to 2^{\Gs}$
is the unique homomorphism $\X \to \L$.
\end{restatable}
\noindent
Since the identity function is trivially a homomorphism,
\Cref{thm:final-for-normal} implies that $\lang\L$ is the identity.
That is, in the G-coalgebra $\L$, the state $L \in \L$ accepts the language $L$!
This proves that every deterministic language is recognized by a $G$-coalgebra,
possibly with an infinite number of states.

\Cref{thm:final-for-normal} says that $\L$ is \emph{final}
for normal $G$-coalgebras. The desired connection between bisimilarity
and language-equivalence is then a standard corollary~\cite{rutten-2000}:

\begin{corollary}%
\label{lem:equiv-bisim}
Let $\X$ and $\Y$ be normal with states $s$ and $t$.
Then $s \sim t$ if and only if $\lang{\X}(s) = \lang{\Y}(t)$.
\end{corollary}
\begin{proof}
The implication from left to right is \Cref{lem:bisim-equiv}.
For the other implication, we observe that the relation
$R \defeq \set{(s,t) \in X \times Y}{\lang\X(s) = \lang\Y(t)}$ is a
bisimulation, using that $\lang\X$ and $\lang\Y$ are homomorphisms by \Cref{thm:final-for-normal}:
\begin{itemize}[left=1ex..1.4\parindent]
\item Suppose $s \mathrel{R} t$ and $\delta^\X(s)(\alpha) \in 2$. Then $
    \delta^\X(s)(\alpha)
    = \delta^\L(\lang\X(s))(\alpha)
    = \delta^\L(\lang\Y(t))(\alpha)
    = \delta^\Y(t)(\alpha)$.

\item Suppose $s \mathrel{R} t$ and $\delta^\X(s)(\alpha) = (p,s')$. Then $
    \delta^\L(\lang\Y(t))(\alpha)
    = \delta^\L(\lang\X(s))(\alpha)
    = (p, \lang\X(s'))$.

  This implies that $\delta^\Y(t)(\alpha) = (p,t')$ for some $t'$,
  using that $\lang\Y$ is a homomorphism. Hence \[
    (p, \lang\Y(t'))
    = \delta^\L(\lang\Y(t))(\alpha)
    = (p, \lang\X(s'))
  \] by the above equation, which implies $s' \mathrel{R} t'$ as required.
\qedhere
\end{itemize}
\end{proof}
\noindent
We prove a stronger version of this result in
\extended{\Cref{lem:bisim}.}{the extended version of this paper~\cite{gkat-full}.}
\subsection{Deciding Equivalence}

We now have all the ingredients required for deciding efficiently
whether two expressions are equivalent. Given two expressions $e_1$ and $e_2$,
the algorithm proceeds as follows:
\begin{enumerate}
  \item Convert $e_1$ and $e_2$ to equivalent Thompson automata $\X_1$ and $\X_2$;
  \item Normalize the automata, obtaining $\hat{\X}_1$ and $\hat{\X}_2$;
  \item Check bisimilarity of the start states $\iota_1$ and $\iota_2$
     using Hopcroft-Karp (see \Cref{alg:hopcroft-karp});
  \item Return \textbf{true} if $\iota_1 \sim \iota_2$, otherwise return \textbf{false}.
\end{enumerate}

\begin{algorithm}[t]
\caption{Hopcroft and Karp's algorithm~\cite{HopcroftKarp71},
  adapted to $G$-automata.}%
\label{alg:hopcroft-karp}
\SetKwInput{Input}{Input}
\SetKwInput{Output}{Output}
\SetKw{True}{true}
\SetKw{False}{false}
\SetKw{KwNot}{not}
\SetKw{Continue}{continue}

\Input{$G$-automata $\X=\angl{X, \delta^\X, \iota^\X}$ and
  $\Y=\angl{Y, \delta^\Y, \iota^\Y}$, finite and normal; $X$, $Y$ disjoint.}
\Output{{\True} if $\X$ and $\Y$ are equivalent, {\False} otherwise.}
\BlankLine%

todo $\gets$ Queue.singleton($\iota^\X, \iota^\Y$)%
  \tcp*[r]{state pairs that need to be checked}
forest $\gets$ UnionFind.disjointForest($X \uplus Y$)\;

\While{\rm \KwNot todo.isEmpty()}{
  $x,y \gets$ todo.pop()\;
  $r_x,r_y \gets$ forest.find($x$), forest.find($y$)\;
  \lIf(\tcp*[f]{safe to assume bisimilar}){$r_x = r_y$}{\Continue}
  \For(\tcp*[f]{check \Cref{def:bisimilarity}}){$\alpha \in \At$}{
    \nllabel{loc:atom-loop}
    \Switch{$\delta^\X(x)(\alpha), \delta^\Y(y)(\alpha)$}{
      \uCase(\tcp*[f]{case (i) of \Cref{def:bisimilarity}})
        {\rm $b_1, b_2$ with $b_1 = b_2$}{\Continue}
      \uCase(\tcp*[f]{case (ii) of \Cref{def:bisimilarity}})
        {$(p,x'),\,(p,y')$}{todo.push($x',y'$)}
      \lOther(\tcp*[f]{not bisimilar})
        {\KwRet{\False}}
    }
  }

  forest.union($r_x,r_y$)%
    \tcp*[r]{mark as bisimilar}
}
\KwRet{\True}\;
\end{algorithm}

\begin{theorem}%
\label{thm:gkat-complexity-1}
  The above algorithm decides whether $\den{e_1} = \den{e_2}$ in time
  $\OO(n \cdot \alpha(n))$ for $|\At|$ constant, where $\alpha$
  denotes the inverse Ackermann function and $n = |e_1| + |e_2|$.
\end{theorem}
\begin{proof}
  The algorithm is correct by \Cref{thm:thompson-correct},
  \Cref{lem:normalization-correct}, and \Cref{lem:equiv-bisim}: \[
    \den{e_1} = \den{e_2}
    \ \iff\ %
    \lang{\X_1}(\iota_1) = \lang{\X_2}(\iota_2)
    \ \iff\ %
    \lang{\hat{\X}_1}(\iota_1) = \lang{\hat{\X}_2}(\iota_2)
    \ \iff\ %
    \iota_1 \sim \iota_2
  \]
  For the complexity claim, we analyze the running time of steps (1)--(3) 
  of the algorithm:
  \begin{enumerate}
    \item Recall by \Cref{prop:thompson-efficient} that the Thompson
    construction converts $e_i$ to  an automaton with $\OO(|e_i|)$
    states in time $\OO(|e_i|)$. Hence this step takes time $\OO(n)$.
    \item Normalizing $\X_i$ amounts to computing its dead states.
    This requires time $\OO(|e_i|)$
    using a breadth-first traversal as follows (since there are at most
    $|\At| \in \OO(1)$ transitions per state).
    We find all states that can reach an accepting state by first
    marking all accepting states, and then performing a
    reverse breadth-first search rooted at the accepting states.
    All marked states are then live; all unmarked states are dead.
    \item Since $\hat{\X}_i$ has $\OO(|e_i|)$ states and there are at most
    $|\At| \in \OO(1)$ transitions per state, Hopcroft-Karp requires time
    $\OO(n \cdot \alpha(n))$ by a classic result due to \citet{Tarjan75}.
    \qedhere
  \end{enumerate}
\end{proof}

\Cref{thm:gkat-complexity-1} establishes a stark complexity gap with KAT, where the same decision problem is
\textsc{PSPACE}-complete~\cite{CKS96a} even for a constant number of atoms.
Intuitively, the gap arises because GKAT expressions can be translated
to linear-size deterministic automata, whereas KAT expressions may require exponential-size deterministic automata.

A shortcoming of \Cref{alg:hopcroft-karp} is that it may scale poorly if the
number of atoms $|\At|$ is large. It is worth noting that there are
symbolic variants~\cite{pous-2014} of the algorithm that avoid enumerating
$\At$ explicitly (\cf{} \cref{loc:atom-loop} of \Cref{alg:hopcroft-karp}),
and can often scale to very large alphabets in practice.
As a concrete example, a version of GKAT specialized to
probabilistic network verification was recently shown~\cite{mcnetkat} to scale
to data-center networks with thousands of switches.
In the worst case, however, we have the following hardness result:
\begin{proposition}%
\label{prop:gkat-complexity-2}
If $|\At|$ is not a constant, GKAT equivalence is co-\textsc{NP}-hard,
but in \textsc{PSPACE}.
\end{proposition}
\begin{proof}
For the hardness result, we observe that Boolean \emph{un}satisfiability
reduces to GKAT equivalence:
$b \in \Bexp$ is unsatisfiable, interpreting the primitive tests
as variables, iff $\den{b} = \emptyset$.
The \textsc{PSPACE} upper bound is inherited from KAT by \Cref{rem:kat-connection}.
\end{proof}

\section{Completeness for the language model}%
\label{sec:complete}

In \Cref{sec:axioms} we presented an axiomatization that is sound with respect to the language model, and put forward the conjecture that it is also complete. We have already proven a partial completeness result (\Cref{cor:partial-completeness}). In this section, we return to this matter
and show we can prove completeness with a generalized version of~(\nameref{ax:fixpoint}).

\subsection{Systems of Left-affine Equations}

A \emph{system of left-affine equations} (or simply, a \emph{system}) is a finite set of equations of the form
\newcommand{\var}[1]{\boldsymbol{#1}}
\begin{gather}
\var{x_1} = e_{11}\var{x_1} +_{b_{11}} \cdots +_{b_{1(n-1)}} e_{1n}\var{x_n}
  +_{b_{1n}} d_1\nonumber\\[-2mm]
\vdots\label{eq:system}\\[-1mm]
\var{x_n} = e_{n1}\var{x_1} +_{b_{n1}} \cdots +_{b_{n(n-1)}} e_{nn}\var{x_n}
  +_{b_{nn}} d_n\nonumber
\end{gather}
where $+_b$ associates to the right,
the $\var{x_i}$ are variables, the $e_{ij}$ are GKAT expressions,
and the $b_{ij}$ and $d_i$ are Boolean guards satisfying the following
row-wise disjointness property for row $1 \leq i \leq n$:
\begin{itemize}
  \item for all $j \neq k$, the guards $b_{ij}$ and $b_{ik}$ are disjoint:
    $b_{ij}\cdot b_{ik} \equiv_\text{BA} 0$; and
  \item for all $1 \leq j \leq n$, the guards $b_{ij}$ and $d_i$ are disjoint:
    $b_{ij}\cdot d_i \equiv_\text{BA} 0$.
\end{itemize}
Note that by the disjointness property, the ordering of the summands is
irrelevant: the system is invariant (up to $\equiv$) under column permutations.
A \emph{solution} to such a system is a function
$s:\sset{\var{x_1}, \dots, \var{x_n}} \to \Exp$
assigning expressions to variables such that, for row $1 \leq i \leq n$:
\[
  s(\var{x_i}) \equiv
    e_{i1}\cdot s(\var{x_1}) +_{b_{i1}} \cdots
    +_{b_{i(n-1)}} e_{in} \cdot s(\var{x_n}) +_{b_{in}} d_i
\]

Note that any finite $G$-coalgebra gives rise to a system where each variable
represents a state, and the equations define what it means to be a solution to
the coalgebra (c.f. \Cref{def:solution}); indeed, a solution to a $G$-coalgebra
is precisely a solution to its corresponding system of equations, and vice versa.
In particular, for a coalgebra $\X$ with states $x_1$ to $x_n$, the parameters
for equation $i$ are:
\begin{mathpar}
b_{ij} = \sum \set{\alpha \in \At}{\delta^\X(x_i)(\alpha) \in \Sigma \times \sset{x_j}}
\\
d_i = \sum \set{\alpha \in \At}{\delta^\X(x_i)(\alpha) = 1}
\and
e_{ij} = \dsum_{\alpha\colon \delta^\X(x_i)(\alpha) = (p_\alpha,x_j)} p_\alpha
\end{mathpar}

Systems arising from $G$-coalgebras have a useful property:
for all $e_{ij}$, it holds that $E(e_{ij}) \equiv 0$.
This property is analogous to the \emph{empty word property} in
Salomaa's axiomatization of regular languages~\cite{Salomaa66}; we call such systems \emph{Salomaa}.

To obtain a general completeness result beyond \Cref{sec:partial-completness}, we assume an additional axiom:
\begin{quote}
\textit{Uniqueness axiom}.
Any system of left-affine equations that is Salomaa has \emph{at most} one solution, modulo $\equiv$.
More precisely, whenever $s$ and $s'$ are solutions to a Salomaa system, it holds that $s(x_i) \equiv s'(x_i)$ for all $1 \leq i \leq n$.
\end{quote}

\begin{remark}
We do not assume that a solution always exists, but only that if it does, then it is unique up to $\equiv$.
It would be unsound to assume that all such systems have solutions;
the following automaton and its system, due to~\cite{KT08a}, constitutes a counterexample:
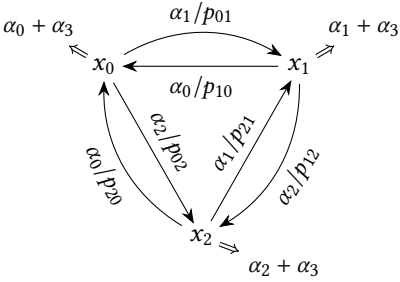
\begin{mathpar}%
\begin{tikzpicture}[baseline]
      \node (e)   {};
      \node[inner sep=0pt,minimum size=0pt] (1) at (30:1.5) {$x_1$};
      \node[inner sep=0pt,minimum size=0pt] (2) at (270:1.5) {$x_2$};
      \node[inner sep=0pt,minimum size=0pt] (0) at (150:1.5) {$x_0$};
      \node (0out) at (150:2.5) {\small$\alpha_0 + \alpha_3$};
      \node (1out) at (30:2.5) {\small$\alpha_1 + \alpha_3$};
      \node (2out) at (300:2.2) {\small$\alpha_2 + \alpha_3$};
      \draw
        (0)  edge[transition-edge,bend left] node[above,sloped] {\small$\alpha_1/p_{01}$} (1)
        (0)  edge[transition-edge] node[above,sloped] {\small$\alpha_2/p_{02}$} (2)
        (1)  edge[transition-edge] node[below,sloped] {\small$\alpha_0/p_{10}$} (0)
        (1)  edge[transition-edge,bend left] node[below,sloped] {\small$\alpha_2/p_{12}$} (2)
        (2)  edge[transition-edge] node[above,sloped] {\small$\alpha_1/p_{21}$} (1)
        (2)  edge[transition-edge,bend left] node[below,sloped] {\small$\alpha_0/p_{20}$} (0)
        (0) edge[output-edge] (0out)
             (1) edge[output-edge] (1out)
         (2) edge[output-edge] (2out);
\end{tikzpicture}
\and
\begin{array}{c@{\kern1ex}c@{\kern1ex}c}
\var{x_0} &\equiv&
  p_{01}\var{x_1} +_{\alpha_1} p_{02}\var{x_2} +_{\alpha_2} (\alpha_0 + \alpha_3) \\
\var{x_1} &\equiv&
  p_{10}\var{x_0}  +_{\alpha_0}  p_{12}\var{x_2} +_{\alpha_2} (\alpha_1 + \alpha_3) \\
\var{x_2} &\equiv&
  p_{20}\var{x_0} +_{\alpha_1} p_{21}\var{x_1} +_{\alpha_0} (\alpha_2 + \alpha_3)
\end{array}
\end{mathpar}
As shown in~\cite{KT08a}, no corresponding while program exists for this system.
\end{remark}

When $n=1$, a system is a single equation of the form $x = ex +_b d$.
In this case,~(\nameref{ax:unroll}) tells us that a solution does exist, namely $e^{(b)}d$, and~(\nameref{ax:fixpoint}) says that this solution is unique up to $\equiv$ under the proviso $E(e) \equiv 0$.
In this sense, we can regard the uniqueness axiom as a generalization of~(\nameref{ax:fixpoint}) to systems with multiple variables.

\begin{restatable}{theorem}{uniquesalomaa}
The uniqueness axiom is sound in the model of guarded strings: given a system of left-affine equations as in~\eqref{eq:system} that is Salomaa, there exists at most one $R:\{\seq x1n\} \to \pGS$ s.t.
\[
  R(x_i) = \left(
    \bigcup_{1 \leq j \leq n} \den{b_{ij}} \diamond \den{e_{ij}} \diamond R(x_j) \right)
    \cup \den{d_i}
\]
\end{restatable}
\begin{proof}[Proof Sketch]
We recast this system as a matrix-vector equation of the form $x = Mx + D$ in the KAT of $n$-by-$n$ matrices over $\pGS$; solutions to $x$ in this equation are in one-to-one correspondence with functions $R$ as above.
We then show that the map $\sigma(x) = Mx + D$ on the set of $n$-dimensional vectors over $\pGS$ is contractive in a certain metric, and therefore has a unique fixpoint by the Banach fixpoint theorem; hence, there can be at most one solution $x$.
\end{proof}

\subsection{General Completeness}

Using the generalized version of the fixpoint axiom, we can now establish completeness.

\begin{theorem}[Completeness]
The axioms are complete w.r.t. $\den{-}$: given $e_1, e_2 \in \Exp$,
\[
    \den{e_1} = \den{e_2} \implies e_1\equiv e_2.
\]
\end{theorem}
\begin{proof}
Let $\X_1$ and $X_2$ be the Thompson automata corresponding to $e_1$ and $e_2$, with initial states $\iota_1$ and $\iota_2$, respectively.
\Cref{thm:thompson-roundtrip} shows there are solutions $s_1$ and $s_2$,
with $s_1(\iota_1) \equiv e_1$ and $s_2(\iota_2) \equiv e_2$;
and we know from \Cref{lem:normalization-correct} that $s_1$ and $s_2$ solve
the normalized automata $\hat{\X_1}$ and $\hat{\X_2}$.
By \Cref{lem:normalization-correct}, \Cref{thm:thompson-correct},
and the premise,
we derive that $\hat{\X_1}$ and $\hat{\X_2}$ accept the same language:
\[
  \lang{\hat{\X_1}}(\iota_1) = \lang{\X_1}(\iota_1) = \den{e_1} = \den{e_2} = \lang{\X_2}(\iota_2) = \lang{\hat{\X_2}}(\iota_2).
\]

This implies, by \Cref{lem:equiv-bisim}, that there is a bisimulation $R$ between $\hat{\X_1}$ and $\hat{\X_2}$ relating $\iota_1$ and $\iota_2$.
This bisimulation can be given the following transition structure,
\[
    \delta^\Bis(x_1, x_2)(\alpha) \defeq
        \begin{cases}
        0 & \text{if } \delta^{\hat{\X_1}}(x_1)(\alpha) = 0 \text{ and } \delta^{\hat{\X_2}}(x_2)(\alpha) = 0 \\
        1 & \text{if } \delta^{\hat{\X_1}}(x_1)(\alpha) = 1 \text{ and } \delta^{\hat{\X_2}}(x_2)(\alpha) = 1 \\
        (p, (x_1', x_2')) & \text{if } \delta^{\hat{\X_1}}(x_1)(\alpha) = (p, x_1') \text{ and } \delta^{\hat{\X_2}}(x_2)(\alpha) = (p, x_2')
        \end{cases}
\]
turning $\Bis = \angl{R, \delta^\Bis}$ into a $G$-coalgebra; note that $\delta^\Bis$ is well-defined since $R$ is a bisimulation.

Now, define $s_1', s_2': R \to \Exp$ by $s_1'(x_1, x_2) = s_1(x_1)$ and $s_2'(x_1, x_2) = s_2(x_2)$.
We claim that $s_1'$ and $s_2'$ are both solutions to $\Bis$; to see this, note that for $\alpha \in \At$, $(x_1, x_2) \in R$, and $i \in \sset{1,2}$, we have that
\begin{align*}
\alpha \cdot s_i'(x_i, x_i)
    &\equiv \alpha \cdot s_i(x_i)
        \tag{Def. $s_i'$} \\
    &\equiv \alpha \cdot \floor{\delta^{\hat{\X_i}}(x_i)(\alpha)}_{s_i}
        \tag{$s_i$ solves $\hat{\X_i}$} \\
    &\equiv \alpha \cdot \floor{\delta^\Bis(x_1, x_2)(\alpha)}_{s_i'}
        \tag{Def. $s_i'$ and $\floor{-}$}
\end{align*}
Thus, $s_i'$ is a solution \extended{by \Cref{lem:solution-alt}}{(c.f.\ the extended proof of \Cref{lem:normalization-correct}~\cite{gkat-full})}.

Since the system of left-affine equations induced by $\Bis$ is Salomaa, the uniqueness axiom then tells us that $s_1(\iota_1) = s_1'(\iota_1, \iota_2) \equiv s_2'(\iota_1, \iota_2) = s_2(\iota_2)$; hence, we can conclude that $e_1 \equiv e_2$.
\end{proof}

\section{Coalgebraic Structure}%
\label{sec:discussion}

The coalgebraic theory of GKAT is quite different from that of KA and
KAT because the final $G$-coalgebra, without the normality assumption
from \cref{sec:normal}, is not characterized by sets of finite guarded
strings. Even including infinite accepted strings is not enough, as
this still cannot distinguish between early and late rejection. It is
therefore of interest to characterize the final $G$-coalgebra and
determine its precise relationship to the language model. We give a
brief overview of these results, which give insight into the nature of
halting versus looping and underscore the role of topology in
coequational specifications.

In \extended{\Cref{sec:coalg}}{the extended version of the
paper~\cite{gkat-full}} we give two characterizations of the final
$G$-coalgebra, one in terms of nonexpansive maps
$\At^\omega\to\Sigma^*\cup\Sigma^\omega$ with natural metrics defined
on both spaces\extended{~(\cref{sec:nonexpansive})}{} and one in terms
of labeled trees with nodes indexed by
$\At^+$\extended{~(\cref{sec:labeledtrees})}{}, and show their
equivalence.
\extended{In \cref{sec:bisim}, we state and prove \Cref{lem:bisim}}
  {We also state and prove}
 a stronger form of the bisimilarity lemma (\Cref{lem:equiv-bisim}).

We have discussed the importance of the determinacy property
(\Cref{def:action-det}). In \extended{\cref{sec:detclosure}}{the
extended version~\cite{gkat-full}} we identify another
important property satisfied by all languages $\Lang{\X}(s)$, a
certain closure property defined in terms of a natural topology on
$\Ato$.
\extended{In \cref{sec:languagemodel}, we}{We} define a language model $\LL'$, a $G$-coalgebra whose states are the subsets of $\GS\cup\wGs$ satisfying the determinacy and closure properties and whose transition structure is the semantic Brzozowski derivative:
\begin{align*}
\delta^{\LL'}(A)(\alpha) = \begin{cases}
(p,\set{x}{\alpha px\in A}) & \text{if $\set{x}{\alpha px\in A}\ne\emptyset$}\\
1 & \text{if $\alpha\in A$}\\
0 & \text{otherwise.}
\end{cases}
\end{align*}
Although this looks similar to the language model $\LL$ of \Cref{sec:bisimNormal}, they are not the same: states of $\LL$ contain finite strings only, and $\LL$ and $\LL'$ are not isomorphic.

We show that $\Lang{}$ is identity on $\LL'$ and that $\LL'$ is isomorphic to a subcoalgebra of the final $G$-coalgebra. It is not the final $G$-coalgebra, because early and late rejection are not distinguished: an automaton could transition before rejecting or reject immediately. Hence, $L:(X,\delta^X)\to\LL'$ is not a homomorphism in general. However, normality prevents this behavior, and $L$ is a homomorphism if $(X,\delta^X)$ is normal. Thus $\LL'$ contains the unique homomorphic image of all normal $G$-coalgebras.

Finally, \extended{in \cref{thm:normalfinal}}{} we identify a subcoalgebra $\LL''\subs\LL'$ that is normal and final in the category of normal $G$-coalgebras. The subcoalgebra $\LL''$ is defined topologically; roughly speaking, it consists of sets $A\subs\GS\cup\GSo$ such that $A$ is the topological closure of $A\cap\GS$.
Thus $\LL''$ is isomorphic to the language model $\LL$ of \Cref{sec:bisimNormal}: the states of $\LL$ are obtained from those of $\LL''$ by intersecting with $\GS$, and the states of $\LL''$ are obtained from those of $\LL$ by taking the topological closure. Thus $\LL$ is isomorphic to a coequationally-defined subcoalgebra of the final $G$-coalgebra.

We also remark that $\LL'$ itself is final in the category of $G$-coalgebras that satisfy a weaker property than normality, the so-called \emph{early failure property}, which can also be characterized topologically.

\section{Related Work}%
\label{sec:related-work}

\emph{Program schematology} is one of the oldest areas of study in the mathematics of computing. It is concerned with questions of translation and representability among and within classes of program schemes, such as flowcharts, while programs, recursion schemes, and schemes with various data structures such as counters, stacks, queues, and dictionaries~\cite{GarlandLuckham73,Ianov60,Rutledge64,PatersonHewitt70,ShepherdsonSturgis63,LuckhamParkPaterson70}.
A classical pursuit in this area was to find mechanisms to transform unstructured flowcharts to structured form~\cite{ashcroft.manna:translation,boehm.jacopini:flow,Kosaraju73,oulsnam:unraveling,Peterson73,ramshaw:eliminating,williams.ossher:conversion,Morris97}. A seminal result was the \emph{B\"ohm-Jacopini theorem}~\cite{boehm.jacopini:flow}, which established that all flowcharts can be converted to while programs provided auxiliary variables are introduced. B\"ohm and Jacopini conjectured that the use of auxiliary variables was necessary in general, and this conjecture was confirmed independently by \citet{ashcroft.manna:translation} and \citet{Kosaraju73}.

Early results in program schematology, including those of~\cite{boehm.jacopini:flow,ashcroft.manna:translation,Kosaraju73}, were typically formulated at the first-order uninterpreted (schematic) level. However, many restructuring operations can be accomplished without reference to first-order constructs. It was shown in~\cite{KT08a} that a purely propositional formulation of the B\"ohm-Jacopini theorem is false: there is a three-state deterministic propositional flowchart that is not equivalent to any propositional while program. As observed by a number of authors (\eg~\cite{Kosaraju73,Peterson73}), while loops with multi-level breaks are sufficient to represent all deterministic flowcharts without introducing auxiliary variables, and~\cite{Kosaraju73} established a strict hierarchy based on the allowed levels of the multi-level breaks. That result was reformulated and proved at the propositional level in~\cite{KT08a}.

The notions of functions on a domain, variables ranging over that domain, and variable assignment are inherent in first-order logic, but are absent at the propositional level. Moreover, many arguments rely on combinatorial graph restructuring operations, which are difficult to formalize. Thus the value of the propositional view is twofold: it operates at a higher level of abstraction and brings topological and coalgebraic concepts and techniques to bear.

Propositional while programs and their encoding in terms of the regular operators goes back to early work on Propositional Dynamic Logic~\cite{FL79}. GKAT as an independent system and its semantics were introduced in~\cite{KT08a,K08b} under the name \emph{propositional while programs}, although the succinct form of the program operators is new here. Also introduced in~\cite{KT08a,K08b} were the functor $G$ and automaton model (\Cref{sec:kleene}), the determinacy property (\Cref{def:action-det}) (called \emph{strict determinacy} there), and the concept of normality (\Cref{sec:normal}) (called \emph{liveness} there). The linear construction of an automaton from a while program was sketched in~\cite{KT08a,K08b}, based on earlier results for KAT automata~\cite{K03a}, but the complexity of deciding equivalence was not addressed. The more rigorous alternative construction given here (\Cref{sec:Thompson}) is needed to establish well-nestedness, thereby enabling our Kleene theorem. The existence of a complete axiomatization was not considered in~\cite{KT08a,K08b}.

Guarded strings, which form the basis of our language semantics, come from~\cite{Kaplan69}.

The axiomatization we propose for GKAT is closely related to Salomaa's axiomatization
of regular expressions based on unique fixed points~\cite{Salomaa66} and to Silva's coalgebraic generalization
of KA~\cite{silva-thesis}. The proof technique we used for completeness is inspired by~\cite{silva-thesis}.

The relational semantics is inspired by that for KAT~\cite{KS96a}, which goes back to work on Dynamic Logic~\cite{FL79}. Because the fixpoint axiom uses a non-algebraic side condition, extra care is needed to define the relational interpretation for GKAT\@.

\section{Conclusions and Future Directions}

We have presented a comprehensive algebraic and coalgebraic study of
GKAT, an abstract programming language with uninterpreted actions. Our
main contributions include: (i) a new automata construction yielding a
nearly linear time decidability result for program equivalence; (ii) a
Kleene theorem for GKAT providing an exact correspondence between
programs and a well-defined class of automata; and (iii) a set of
sound and complete axioms for program equivalence.

We hope this paper is only the beginning of a long and beautiful
journey into understanding the (co)algebraic properties of efficient 
fragments of imperative programming languages. We briefly discuss some
limitations of our current development and our vision for future work.

As in Salomaa's axiomatization of KA, our axiomatization is not fully
algebraic:
the side condition of (\nameref{ax:fixpoint}) is only sensible for the language model.
As a result, the current completeness proof does not generalize to
other natural models of interest---\eg, probabilistic or
relational. To overcome this limitation, we would like to adapt
Kozen's axiomatization of KA to GKAT by developing a natural order for
GKAT programs. In the case of KA we have $e\leq f \defiff e+f =f$, but
this natural order is no longer definable in the absence of $+$ and so we need to
axiomatize $e \leq f$ for GKAT programs directly. This
appears to be the main missing piece to obtain an algebraic
axiomatization.

On the coalgebraic side, we are interested in studying the different classes of
$G$-coalgebras from a coequational perspective. Normal coalgebras, for instance,
form a covariety, and hence are characterized by coequations. If well-nested $G$-coalgebras could be shown to form a covariety,
this would imply completeness of the axioms
without the extra uniqueness axiom from \Cref{sec:complete}.

Various extensions of KAT to reason about richer programs (KAT+B!,
NetKAT, ProbNetKAT) have been proposed, and it is natural to ask
whether extending GKAT in similar directions will yield interesting
algebraic theories and decision procedures for domain-specific
applications. For instance, recent work~\cite{mcnetkat} on a
probabilistic network verification tool suggests that
GKAT is better suited for
probabilistic models than KAT, as it avoids mixing non-determinism and
probabilities.
The complex semantics of  probabilistic programs would make a framework for
equational and automated reasoning especially valuable.

In a different direction, a language model containing infinite traces
could be interesting in many applications, as it could serve as a
model to reason about non-terminating programs---\eg, loops in NetKAT
in which packets may be forwarded forever. An interesting open
question is whether the infinite language model can be finitely
axiomatized.

Finally, another direction would be to extend the GKAT decision
procedure to handle extra equations. For instance, both KAT+B! and 
NetKAT have independently-developed decision procedures, that are
similar in flavor. This raises the question of whether the GKAT
decision procedure could be extended in a more generic way,
similar to the Nelson-Oppen approach~\cite{NelsonO79} for
combining decision procedures used in SMT solving.

\begin{acks}
We are grateful to the anonymous reviewers for their feedback and help in
improving this paper, and
thank Paul Brunet, Fredrik Dahlqvist, and Jonathan DiLorenzo for numerous
discussions and suggestions on GKAT\@.
Jonathan DiLorenzo and Simon Docherty provided helpful feedback on drafts of
this paper.
We thank the Bellairs Research Institute of McGill University for
providing a wonderful research environment.
This work was supported in part by the University of Wisconsin, a Facebook TAV
award,
ERC starting grant Profoundnet (679127),
a Royal Society Wolfson fellowship,
a Leverhulme Prize (PLP-2016-129), 
NSF grants AitF-1637532, CNS-1413978, and SaTC-1717581,
and gifts from Fujitsu and InfoSys.
\end{acks}

\bibliography{popl20}


\begin{thebibliography}{51}


\ifx \showCODEN    \undefined \def \showCODEN     #1{\unskip}     \fi
\ifx \showDOI      \undefined \def \showDOI       #1{#1}\fi
\ifx \showISBNx    \undefined \def \showISBNx     #1{\unskip}     \fi
\ifx \showISBNxiii \undefined \def \showISBNxiii  #1{\unskip}     \fi
\ifx \showISSN     \undefined \def \showISSN      #1{\unskip}     \fi
\ifx \showLCCN     \undefined \def \showLCCN      #1{\unskip}     \fi
\ifx \shownote     \undefined \def \shownote      #1{#1}          \fi
\ifx \showarticletitle \undefined \def \showarticletitle #1{#1}   \fi
\ifx \showURL      \undefined \def \showURL       {\relax}        \fi
\providecommand\bibfield[2]{#2}
\providecommand\bibinfo[2]{#2}
\providecommand\natexlab[1]{#1}
\providecommand\showeprint[2][]{arXiv:#2}

\bibitem[\protect\citeauthoryear{Anderson, Foster, Guha, Jeannin, Kozen,
  Schlesinger, and Walker}{Anderson et~al\mbox{.}}{2014}]%
        {AFGJKSW13a}
\bibfield{author}{\bibinfo{person}{Carolyn~Jane Anderson},
  \bibinfo{person}{Nate Foster}, \bibinfo{person}{Arjun Guha},
  \bibinfo{person}{Jean-Baptiste Jeannin}, \bibinfo{person}{Dexter Kozen},
  \bibinfo{person}{Cole Schlesinger}, {and} \bibinfo{person}{David Walker}.}
  \bibinfo{year}{2014}\natexlab{}.
\newblock \showarticletitle{{NetKAT}: Semantic Foundations for Networks}. In
  \bibinfo{booktitle}{\emph{Proc. Principles of Programming Languages
  ({POPL})}}. \bibinfo{publisher}{ACM}, \bibinfo{address}{New York, NY, USA},
  \bibinfo{pages}{113--126}.
\newblock
\urldef\tempurl%
\url{https://doi.org/10.1145/2535838.2535862}
\showDOI{\tempurl}


\bibitem[\protect\citeauthoryear{Angus and Kozen}{Angus and Kozen}{2001}]%
        {AK01a}
\bibfield{author}{\bibinfo{person}{Allegra Angus} {and} \bibinfo{person}{Dexter
  Kozen}.} \bibinfo{year}{2001}\natexlab{}.
\newblock \bibinfo{booktitle}{\emph{Kleene Algebra with Tests and Program
  Schematology}}.
\newblock \bibinfo{type}{{T}echnical {R}eport} TR2001-1844.
  \bibinfo{institution}{Computer Science Department, Cornell University}.
\newblock


\bibitem[\protect\citeauthoryear{Ashcroft and Manna}{Ashcroft and
  Manna}{1972}]%
        {ashcroft.manna:translation}
\bibfield{author}{\bibinfo{person}{Edward~A. Ashcroft} {and}
  \bibinfo{person}{Zohar Manna}.} \bibinfo{year}{1972}\natexlab{}.
\newblock \showarticletitle{The translation of GOTO programs into WHILE
  programs}. In \bibinfo{booktitle}{\emph{Proc. Information Processing
  ({IFIP})}}, Vol.~\bibinfo{volume}{1}. \bibinfo{publisher}{North-Holland},
  \bibinfo{address}{Amsterdam, The Netherlands}, \bibinfo{pages}{250--255}.
\newblock


\bibitem[\protect\citeauthoryear{Backhouse}{Backhouse}{1975}]%
        {backhouse-1975}
\bibfield{author}{\bibinfo{person}{Roland Backhouse}.}
  \bibinfo{year}{1975}\natexlab{}.
\newblock \emph{\bibinfo{title}{Closure algorithms and the star-height problem
  of regular languages}}.
\newblock \bibinfo{thesistype}{Ph.D. Dissertation}. \bibinfo{school}{University
  of London}.
\newblock
\urldef\tempurl%
\url{http://ethos.bl.uk/OrderDetails.do?uin=uk.bl.ethos.448525}
\showURL{%
\tempurl}


\bibitem[\protect\citeauthoryear{Barth and Kozen}{Barth and Kozen}{2002}]%
        {BK02a}
\bibfield{author}{\bibinfo{person}{Adam Barth} {and} \bibinfo{person}{Dexter
  Kozen}.} \bibinfo{year}{2002}\natexlab{}.
\newblock \bibinfo{booktitle}{\emph{Equational Verification of Cache Blocking
  in {LU} Decomposition using {K}leene Algebra with Tests}}.
\newblock \bibinfo{type}{{T}echnical {R}eport} TR2002-1865.
  \bibinfo{institution}{Computer Science Department, Cornell University}.
\newblock


\bibitem[\protect\citeauthoryear{Birkhoff and Bartee}{Birkhoff and
  Bartee}{1970}]%
        {birkhoff-bartee-1970}
\bibfield{author}{\bibinfo{person}{Garrett Birkhoff} {and}
  \bibinfo{person}{Thomas~C. Bartee}.} \bibinfo{year}{1970}\natexlab{}.
\newblock \bibinfo{booktitle}{\emph{Modern applied algebra}}.
\newblock \bibinfo{publisher}{McGraw-Hill}, \bibinfo{address}{New York, NY,
  USA}.
\newblock
\showISBNx{978-0070053816}


\bibitem[\protect\citeauthoryear{B{\"o}hm and Jacopini}{B{\"o}hm and
  Jacopini}{1966}]%
        {boehm.jacopini:flow}
\bibfield{author}{\bibinfo{person}{Corrado B{\"o}hm} {and}
  \bibinfo{person}{Guiseppe Jacopini}.} \bibinfo{year}{1966}\natexlab{}.
\newblock \showarticletitle{Flow Diagrams, {T}uring Machines and Languages with
  only Two Formation Rules}.
\newblock \bibinfo{journal}{\emph{Commun. ACM}} (\bibinfo{date}{May}
  \bibinfo{year}{1966}), \bibinfo{pages}{366--371}.
\newblock
\urldef\tempurl%
\url{https://doi.org/10.1145/355592.365646}
\showDOI{\tempurl}


\bibitem[\protect\citeauthoryear{Bonchi and Pous}{Bonchi and Pous}{2013}]%
        {bonchi-pous-2013}
\bibfield{author}{\bibinfo{person}{Filippo Bonchi} {and}
  \bibinfo{person}{Damien Pous}.} \bibinfo{year}{2013}\natexlab{}.
\newblock \showarticletitle{Checking {NFA} equivalence with bisimulations up to
  congruence}. In \bibinfo{booktitle}{\emph{Proc. Principles of Programming
  Languages ({POPL})}}. \bibinfo{publisher}{ACM}, \bibinfo{address}{New York,
  NY, USA}, \bibinfo{pages}{457--468}.
\newblock
\urldef\tempurl%
\url{https://doi.org/10.1145/2429069.2429124}
\showDOI{\tempurl}


\bibitem[\protect\citeauthoryear{Cohen}{Cohen}{1994a}]%
        {Co94c}
\bibfield{author}{\bibinfo{person}{Ernie Cohen}.}
  \bibinfo{year}{1994}\natexlab{a}.
\newblock \bibinfo{title}{Lazy Caching in {K}leene Algebra}.
\newblock
\newblock
\urldef\tempurl%
\url{http://citeseerx.ist.psu.edu/viewdoc/summary?doi=10.1.1.57.5074}
\showURL{%
\tempurl}


\bibitem[\protect\citeauthoryear{Cohen}{Cohen}{1994b}]%
        {Co94d}
\bibfield{author}{\bibinfo{person}{Ernie Cohen}.}
  \bibinfo{year}{1994}\natexlab{b}.
\newblock \bibinfo{booktitle}{\emph{Using Kleene algebra to reason about
  concurrency control}}.
\newblock \bibinfo{type}{{T}echnical {R}eport}.
  \bibinfo{institution}{Telcordia, Morristown, NJ}.
\newblock


\bibitem[\protect\citeauthoryear{Cohen, Kozen, and Smith}{Cohen
  et~al\mbox{.}}{1996}]%
        {CKS96a}
\bibfield{author}{\bibinfo{person}{Ernie Cohen}, \bibinfo{person}{Dexter
  Kozen}, {and} \bibinfo{person}{Frederick Smith}.}
  \bibinfo{year}{1996}\natexlab{}.
\newblock \bibinfo{booktitle}{\emph{The complexity of {K}leene algebra with
  tests}}.
\newblock \bibinfo{type}{{T}echnical {R}eport} TR96-1598.
  \bibinfo{institution}{Computer Science Department, Cornell University}.
\newblock


\bibitem[\protect\citeauthoryear{Conway}{Conway}{1971}]%
        {conway-1971}
\bibfield{author}{\bibinfo{person}{John~Horton Conway}.}
  \bibinfo{year}{1971}\natexlab{}.
\newblock \bibinfo{booktitle}{\emph{Regular Algebra and Finite Machines}}.
\newblock \bibinfo{publisher}{Chapman and Hall}, \bibinfo{address}{London,
  United Kingdom}.
\newblock


\bibitem[\protect\citeauthoryear{Erosa and Hendren}{Erosa and Hendren}{1994}]%
        {ErosaH94}
\bibfield{author}{\bibinfo{person}{Ana~M. Erosa} {and}
  \bibinfo{person}{Laurie~J. Hendren}.} \bibinfo{year}{1994}\natexlab{}.
\newblock \showarticletitle{Taming Control Flow: {A} Structured Approach to
  Eliminating Goto Statements}. In \bibinfo{booktitle}{\emph{Proc. Computer
  Languages ({ICCL})}}. \bibinfo{publisher}{{IEEE} Computer Society},
  \bibinfo{address}{Los Alamitos, CA, USA}, \bibinfo{pages}{229--240}.
\newblock
\showISBNx{0-8186-5640-9}
\urldef\tempurl%
\url{https://doi.org/10.1109/ICCL.1994.288377}
\showDOI{\tempurl}


\bibitem[\protect\citeauthoryear{Fischer and Ladner}{Fischer and
  Ladner}{1979}]%
        {FL79}
\bibfield{author}{\bibinfo{person}{Michael~J. Fischer} {and}
  \bibinfo{person}{Richard~E. Ladner}.} \bibinfo{year}{1979}\natexlab{}.
\newblock \showarticletitle{Propositional dynamic logic of regular programs}.
\newblock \bibinfo{journal}{\emph{J. Comput. System Sci.}}
  \bibinfo{volume}{18}, \bibinfo{number}{2} (\bibinfo{year}{1979}),
  \bibinfo{pages}{194--211}.
\newblock
\urldef\tempurl%
\url{https://doi.org/10.1016/0022-0000(79)90046-1}
\showDOI{\tempurl}


\bibitem[\protect\citeauthoryear{Foster, Kozen, Mamouras, Reitblatt, and
  Silva}{Foster et~al\mbox{.}}{2016}]%
        {probnetkat-cantor}
\bibfield{author}{\bibinfo{person}{Nate Foster}, \bibinfo{person}{Dexter
  Kozen}, \bibinfo{person}{Konstantinos Mamouras}, \bibinfo{person}{Mark
  Reitblatt}, {and} \bibinfo{person}{Alexandra Silva}.}
  \bibinfo{year}{2016}\natexlab{}.
\newblock \showarticletitle{Probabilistic {NetKAT}}. In
  \bibinfo{booktitle}{\emph{Proc. European Symposium on Programming ({ESOP})}}.
  \bibinfo{publisher}{ACM}, \bibinfo{address}{New York, NY, USA},
  \bibinfo{pages}{282--309}.
\newblock
\urldef\tempurl%
\url{https://doi.org/10.1007/978-3-662-49498-1_12}
\showDOI{\tempurl}


\bibitem[\protect\citeauthoryear{Foster, Kozen, Milano, Silva, and
  Thompson}{Foster et~al\mbox{.}}{2015}]%
        {FKMST15a}
\bibfield{author}{\bibinfo{person}{Nate Foster}, \bibinfo{person}{Dexter
  Kozen}, \bibinfo{person}{Matthew Milano}, \bibinfo{person}{Alexandra Silva},
  {and} \bibinfo{person}{Laure Thompson}.} \bibinfo{year}{2015}\natexlab{}.
\newblock \showarticletitle{A Coalgebraic Decision Procedure for {NetKAT}}. In
  \bibinfo{booktitle}{\emph{Proc. Principles of Programming Languages
  ({POPL})}}. \bibinfo{publisher}{ACM}, \bibinfo{address}{New York, NY, USA},
  \bibinfo{pages}{343--355}.
\newblock
\urldef\tempurl%
\url{https://doi.org/10.1145/2676726.2677011}
\showDOI{\tempurl}


\bibitem[\protect\citeauthoryear{Garland and Luckham}{Garland and
  Luckham}{1973}]%
        {GarlandLuckham73}
\bibfield{author}{\bibinfo{person}{Stephen~J. Garland} {and}
  \bibinfo{person}{David~C. Luckham}.} \bibinfo{year}{1973}\natexlab{}.
\newblock \showarticletitle{Program schemes, recursion schemes, and formal
  languages}.
\newblock \bibinfo{journal}{\emph{J. Comput. System Sci.}} \bibinfo{volume}{7},
  \bibinfo{number}{2} (\bibinfo{year}{1973}), \bibinfo{pages}{119 -- 160}.
\newblock
\showISSN{0022-0000}
\urldef\tempurl%
\url{https://doi.org/10.1016/S0022-0000(73)80040-6}
\showDOI{\tempurl}


\bibitem[\protect\citeauthoryear{Giry}{Giry}{1982}]%
        {giry1982categorical}
\bibfield{author}{\bibinfo{person}{Michele Giry}.}
  \bibinfo{year}{1982}\natexlab{}.
\newblock \showarticletitle{A categorical approach to probability theory}.
\newblock In \bibinfo{booktitle}{\emph{Categorical aspects of topology and
  analysis}}. \bibinfo{publisher}{Springer}, \bibinfo{address}{Berlin,
  Heidelberg}, \bibinfo{pages}{68--85}.
\newblock
\urldef\tempurl%
\url{https://doi.org/10.1007/BFb0092872}
\showDOI{\tempurl}


\bibitem[\protect\citeauthoryear{Hendren, Donawa, Emami, Gao, Justiani, and
  Sridharan}{Hendren et~al\mbox{.}}{1992}]%
        {mccat}
\bibfield{author}{\bibinfo{person}{Laurie~J. Hendren}, \bibinfo{person}{C.
  Donawa}, \bibinfo{person}{Maryam Emami}, \bibinfo{person}{Guang~R. Gao},
  \bibinfo{person}{Justiani}, {and} \bibinfo{person}{B. Sridharan}.}
  \bibinfo{year}{1992}\natexlab{}.
\newblock \showarticletitle{Designing the McCAT Compiler Based on a Family of
  Structured Intermediate Representations}. In \bibinfo{booktitle}{\emph{Proc.
  Languages and Compilers for Parallel Computing ({LCPC})}}.
  \bibinfo{publisher}{Springer}, \bibinfo{address}{Berlin, Heidelberg},
  \bibinfo{pages}{406--420}.
\newblock
\urldef\tempurl%
\url{https://doi.org/10.1007/3-540-57502-2_61}
\showDOI{\tempurl}


\bibitem[\protect\citeauthoryear{Hopcroft and Karp}{Hopcroft and Karp}{1971}]%
        {HopcroftKarp71}
\bibfield{author}{\bibinfo{person}{John~E. Hopcroft} {and}
  \bibinfo{person}{Richard~M. Karp}.} \bibinfo{year}{1971}\natexlab{}.
\newblock \bibinfo{booktitle}{\emph{A linear algorithm for testing equivalence
  of finite automata}}.
\newblock \bibinfo{type}{{T}echnical {R}eport} TR 71-114.
  \bibinfo{institution}{Cornell University}.
\newblock


\bibitem[\protect\citeauthoryear{Ianov}{Ianov}{1960}]%
        {Ianov60}
\bibfield{author}{\bibinfo{person}{I. Ianov}.} \bibinfo{year}{1960}\natexlab{}.
\newblock \showarticletitle{The Logical Schemes of Algorithms}.
\newblock \bibinfo{journal}{\emph{Problems of Cybernetics}}
  (\bibinfo{year}{1960}), \bibinfo{pages}{82--140}.
\newblock


\bibitem[\protect\citeauthoryear{Kaplan}{Kaplan}{1969}]%
        {Kaplan69}
\bibfield{author}{\bibinfo{person}{Donald~M. Kaplan}.}
  \bibinfo{year}{1969}\natexlab{}.
\newblock \showarticletitle{Regular Expressions and the Equivalence of
  Programs}.
\newblock \bibinfo{journal}{\emph{J. Comput. System Sci.}}  \bibinfo{volume}{3}
  (\bibinfo{year}{1969}), \bibinfo{pages}{361--386}.
\newblock
\urldef\tempurl%
\url{https://doi.org/10.1016/S0022-0000(69)80027-9}
\showDOI{\tempurl}


\bibitem[\protect\citeauthoryear{Kleene}{Kleene}{1956}]%
        {kleene-1956}
\bibfield{author}{\bibinfo{person}{Stephen~C. Kleene}.}
  \bibinfo{year}{1956}\natexlab{}.
\newblock \showarticletitle{Representation of Events in Nerve Nets and Finite
  Automata}.
\newblock \bibinfo{journal}{\emph{Automata Studies}} (\bibinfo{year}{1956}),
  \bibinfo{pages}{3--41}.
\newblock


\bibitem[\protect\citeauthoryear{Kosaraju}{Kosaraju}{1973}]%
        {Kosaraju73}
\bibfield{author}{\bibinfo{person}{S.~Rao Kosaraju}.}
  \bibinfo{year}{1973}\natexlab{}.
\newblock \showarticletitle{Analysis of structured programs}. In
  \bibinfo{booktitle}{\emph{Proc. Theory of Computing ({STOC})}}.
  \bibinfo{publisher}{ACM}, \bibinfo{address}{New York, NY, USA},
  \bibinfo{pages}{240--252}.
\newblock
\urldef\tempurl%
\url{https://doi.org/10.1145/800125.804055}
\showDOI{\tempurl}


\bibitem[\protect\citeauthoryear{Kozen}{Kozen}{1985}]%
        {K85a}
\bibfield{author}{\bibinfo{person}{Dexter Kozen}.}
  \bibinfo{year}{1985}\natexlab{}.
\newblock \showarticletitle{A probabilistic {{\em PDL}}}.
\newblock \bibinfo{journal}{\emph{J. Comput. System Sci.}}
  \bibinfo{volume}{30}, \bibinfo{number}{2} (\bibinfo{date}{April}
  \bibinfo{year}{1985}), \bibinfo{pages}{162--178}.
\newblock
\urldef\tempurl%
\url{https://doi.org/10.1016/0022-0000(85)90012-1}
\showDOI{\tempurl}


\bibitem[\protect\citeauthoryear{Kozen}{Kozen}{1996}]%
        {K96b}
\bibfield{author}{\bibinfo{person}{Dexter Kozen}.}
  \bibinfo{year}{1996}\natexlab{}.
\newblock \showarticletitle{Kleene algebra with tests and commutativity
  conditions}. In \bibinfo{booktitle}{\emph{Proc. Tools and Algorithms for the
  Construction and Analysis of Systems ({TACAS})}}
  \emph{(\bibinfo{series}{Lecture Notes in Computer Science})},
  Vol.~\bibinfo{volume}{1055}. \bibinfo{publisher}{Springer-Verlag},
  \bibinfo{address}{Passau, Germany}, \bibinfo{pages}{14--33}.
\newblock
\urldef\tempurl%
\url{https://doi.org/10.1007/3-540-61042-1_35}
\showDOI{\tempurl}


\bibitem[\protect\citeauthoryear{Kozen}{Kozen}{1997}]%
        {K97c}
\bibfield{author}{\bibinfo{person}{Dexter Kozen}.}
  \bibinfo{year}{1997}\natexlab{}.
\newblock \showarticletitle{Kleene algebra with tests}.
\newblock \bibinfo{journal}{\emph{ACM Transactions on Programming Languages and
  Systems (TOPLAS)}} \bibinfo{volume}{19}, \bibinfo{number}{3}
  (\bibinfo{date}{May} \bibinfo{year}{1997}), \bibinfo{pages}{427--443}.
\newblock
\urldef\tempurl%
\url{https://doi.org/10.1145/256167.256195}
\showDOI{\tempurl}


\bibitem[\protect\citeauthoryear{Kozen}{Kozen}{2003}]%
        {K03a}
\bibfield{author}{\bibinfo{person}{Dexter Kozen}.}
  \bibinfo{year}{2003}\natexlab{}.
\newblock \showarticletitle{Automata on Guarded Strings and Applications}.
\newblock \bibinfo{journal}{\emph{Mat{\'e}matica Contempor{\^a}nea}}
  \bibinfo{volume}{24} (\bibinfo{year}{2003}), \bibinfo{pages}{117--139}.
\newblock


\bibitem[\protect\citeauthoryear{Kozen}{Kozen}{2008}]%
        {K08b}
\bibfield{author}{\bibinfo{person}{Dexter Kozen}.}
  \bibinfo{year}{2008}\natexlab{}.
\newblock \showarticletitle{Nonlocal Flow of Control and {K}leene Algebra with
  Tests}. In \bibinfo{booktitle}{\emph{Proc. Logic in Computer Science
  ({LICS})}}. \bibinfo{publisher}{IEEE}, \bibinfo{address}{New York, NY, USA},
  \bibinfo{pages}{105--117}.
\newblock
\urldef\tempurl%
\url{https://doi.org/10.1109/LICS.2008.32}
\showDOI{\tempurl}


\bibitem[\protect\citeauthoryear{Kozen and Patron}{Kozen and Patron}{2000}]%
        {KP00}
\bibfield{author}{\bibinfo{person}{Dexter Kozen} {and}
  \bibinfo{person}{Maria-Cristina Patron}.} \bibinfo{year}{2000}\natexlab{}.
\newblock \showarticletitle{Certification of compiler optimizations using
  {K}leene algebra with tests}. In \bibinfo{booktitle}{\emph{Proc.
  Computational Logic ({CL})}} \emph{(\bibinfo{series}{Lecture Notes in
  Artificial Intelligence})}, Vol.~\bibinfo{volume}{1861}.
  \bibinfo{publisher}{Springer-Verlag}, \bibinfo{address}{London, United
  Kingdom}, \bibinfo{pages}{568--582}.
\newblock
\urldef\tempurl%
\url{https://doi.org/10.1007/3-540-44957-4_38}
\showDOI{\tempurl}


\bibitem[\protect\citeauthoryear{Kozen and Smith}{Kozen and Smith}{1996}]%
        {KS96a}
\bibfield{author}{\bibinfo{person}{Dexter Kozen} {and}
  \bibinfo{person}{Frederick Smith}.} \bibinfo{year}{1996}\natexlab{}.
\newblock \showarticletitle{Kleene algebra with tests: Completeness and
  decidability}. In \bibinfo{booktitle}{\emph{Proc. Computer Science Logic
  ({CSL})}} \emph{(\bibinfo{series}{Lecture Notes in Computer Science})},
  Vol.~\bibinfo{volume}{1258}. \bibinfo{publisher}{Springer-Verlag},
  \bibinfo{address}{Utrecht, The Netherlands}, \bibinfo{pages}{244--259}.
\newblock
\urldef\tempurl%
\url{https://doi.org/10.1007/3-540-63172-0_43}
\showDOI{\tempurl}


\bibitem[\protect\citeauthoryear{Kozen and Tseng}{Kozen and Tseng}{2008}]%
        {KT08a}
\bibfield{author}{\bibinfo{person}{Dexter Kozen} {and}
  \bibinfo{person}{Wei-Lung~(Dustin) Tseng}.} \bibinfo{year}{2008}\natexlab{}.
\newblock \showarticletitle{The {B\"o}hm-{J}acopini Theorem is False,
  Propositionally}. In \bibinfo{booktitle}{\emph{Proc. Mathematics of Program
  Construction ({MPC})}} \emph{(\bibinfo{series}{Lecture Notes in Computer
  Science})}, Vol.~\bibinfo{volume}{5133}. \bibinfo{publisher}{Springer},
  \bibinfo{address}{Berlin, Heidelberg}, \bibinfo{pages}{177--192}.
\newblock
\urldef\tempurl%
\url{https://doi.org/10.1007/978-3-540-70594-9_11}
\showDOI{\tempurl}


\bibitem[\protect\citeauthoryear{Luckham, Park, and Paterson}{Luckham
  et~al\mbox{.}}{1970}]%
        {LuckhamParkPaterson70}
\bibfield{author}{\bibinfo{person}{David~C. Luckham}, \bibinfo{person}{David
  M.~R. Park}, {and} \bibinfo{person}{Michael~S. Paterson}.}
  \bibinfo{year}{1970}\natexlab{}.
\newblock \showarticletitle{On formalised computer programs}.
\newblock \bibinfo{journal}{\emph{J. Comput. System Sci.}} \bibinfo{volume}{4},
  \bibinfo{number}{3} (\bibinfo{year}{1970}), \bibinfo{pages}{220--249}.
\newblock
\urldef\tempurl%
\url{https://doi.org/10.1016/S0022-0000(70)80022-8}
\showDOI{\tempurl}


\bibitem[\protect\citeauthoryear{Mislove}{Mislove}{2006}]%
        {MISLOVE2006261}
\bibfield{author}{\bibinfo{person}{Michael~W. Mislove}.}
  \bibinfo{year}{2006}\natexlab{}.
\newblock \showarticletitle{On Combining Probability and Nondeterminism}.
\newblock \bibinfo{journal}{\emph{Electronic Notes in Theoretical Computer
  Science}}  \bibinfo{volume}{162} (\bibinfo{year}{2006}), \bibinfo{pages}{261
  -- 265}.
\newblock
\showISSN{1571-0661}
\urldef\tempurl%
\url{https://doi.org/10.1016/j.entcs.2005.12.113}
\showDOI{\tempurl}
\newblock
\shownote{Proc. Algebraic Process Calculi ({APC}).}


\bibitem[\protect\citeauthoryear{Morris, Gray, and Filman}{Morris
  et~al\mbox{.}}{1997}]%
        {Morris97}
\bibfield{author}{\bibinfo{person}{Paul~H. Morris}, \bibinfo{person}{Ronald~A.
  Gray}, {and} \bibinfo{person}{Robert~E. Filman}.}
  \bibinfo{year}{1997}\natexlab{}.
\newblock \showarticletitle{GOTO Removal Based on Regular Expressions}.
\newblock \bibinfo{journal}{\emph{Journal of Software Maintenance: Research and
  Practice}} \bibinfo{volume}{9}, \bibinfo{number}{1} (\bibinfo{year}{1997}),
  \bibinfo{pages}{47--66}.
\newblock
\urldef\tempurl%
\url{https://doi.org/10.1002/(SICI)1096-908X(199701)9:1<47::AID-SMR142>3.0.CO;2-V}
\showDOI{\tempurl}


\bibitem[\protect\citeauthoryear{Nelson and Oppen}{Nelson and Oppen}{1979}]%
        {NelsonO79}
\bibfield{author}{\bibinfo{person}{Greg Nelson} {and} \bibinfo{person}{Derek~C.
  Oppen}.} \bibinfo{year}{1979}\natexlab{}.
\newblock \showarticletitle{Simplification by Cooperating Decision Procedures}.
\newblock \bibinfo{journal}{\emph{ACM Transactions on Programming Languages and
  Systems (TOPLAS)}} \bibinfo{volume}{1}, \bibinfo{number}{2}
  (\bibinfo{year}{1979}), \bibinfo{pages}{245--257}.
\newblock
\urldef\tempurl%
\url{https://doi.org/10.1145/357073.357079}
\showDOI{\tempurl}


\bibitem[\protect\citeauthoryear{Oulsnam}{Oulsnam}{1982}]%
        {oulsnam:unraveling}
\bibfield{author}{\bibinfo{person}{G. Oulsnam}.}
  \bibinfo{year}{1982}\natexlab{}.
\newblock \showarticletitle{Unraveling unstructured programs}.
\newblock \bibinfo{journal}{\emph{Comput. J.}} \bibinfo{volume}{25},
  \bibinfo{number}{3} (\bibinfo{year}{1982}), \bibinfo{pages}{379--387}.
\newblock
\urldef\tempurl%
\url{https://doi.org/10.1093/comjnl/25.3.379}
\showDOI{\tempurl}


\bibitem[\protect\citeauthoryear{Paterson and Hewitt}{Paterson and
  Hewitt}{1970}]%
        {PatersonHewitt70}
\bibfield{author}{\bibinfo{person}{Michael~S. Paterson} {and}
  \bibinfo{person}{Carl~E. Hewitt}.} \bibinfo{year}{1970}\natexlab{}.
\newblock \showarticletitle{Comparative schematology}. In
  \bibinfo{booktitle}{\emph{Record of Project MAC Conference on Concurrent
  Systems and Parallel Computation}}. \bibinfo{publisher}{ACM},
  \bibinfo{address}{New York, NY, USA}, \bibinfo{pages}{119--127}.
\newblock


\bibitem[\protect\citeauthoryear{Peterson, Kasami, and Tokura}{Peterson
  et~al\mbox{.}}{1973}]%
        {Peterson73}
\bibfield{author}{\bibinfo{person}{W.~Wesley Peterson}, \bibinfo{person}{Tadao
  Kasami}, {and} \bibinfo{person}{Nobuki Tokura}.}
  \bibinfo{year}{1973}\natexlab{}.
\newblock \showarticletitle{On the Capabilities of while, repeat, and exit
  Statements}.
\newblock \bibinfo{journal}{\emph{Commun. ACM}} \bibinfo{volume}{16},
  \bibinfo{number}{8} (\bibinfo{year}{1973}), \bibinfo{pages}{503--512}.
\newblock
\urldef\tempurl%
\url{https://doi.org/10.1145/355609.362337}
\showDOI{\tempurl}


\bibitem[\protect\citeauthoryear{Pous}{Pous}{2015}]%
        {pous-2014}
\bibfield{author}{\bibinfo{person}{Damien Pous}.}
  \bibinfo{year}{2015}\natexlab{}.
\newblock \showarticletitle{Symbolic Algorithms for Language Equivalence and
  {K}leene Algebra with Tests}. In \bibinfo{booktitle}{\emph{Proc. Principles
  of Programming Languages ({POPL})}}. \bibinfo{publisher}{ACM},
  \bibinfo{address}{New York, NY, USA}, \bibinfo{pages}{357--368}.
\newblock
\urldef\tempurl%
\url{https://doi.org/10.1145/2676726.2677007}
\showDOI{\tempurl}


\bibitem[\protect\citeauthoryear{Ramshaw}{Ramshaw}{1988}]%
        {ramshaw:eliminating}
\bibfield{author}{\bibinfo{person}{Lyle Ramshaw}.}
  \bibinfo{year}{1988}\natexlab{}.
\newblock \showarticletitle{Eliminating goto's while preserving program
  structure}.
\newblock \bibinfo{journal}{\emph{J. ACM}} \bibinfo{volume}{35},
  \bibinfo{number}{4} (\bibinfo{year}{1988}), \bibinfo{pages}{893--920}.
\newblock
\urldef\tempurl%
\url{https://doi.org/10.1145/48014.48021}
\showDOI{\tempurl}


\bibitem[\protect\citeauthoryear{Rutledge}{Rutledge}{1964}]%
        {Rutledge64}
\bibfield{author}{\bibinfo{person}{Joseph~D. Rutledge}.}
  \bibinfo{year}{1964}\natexlab{}.
\newblock \showarticletitle{On Ianov's Program Schemata}.
\newblock \bibinfo{journal}{\emph{J. ACM}} \bibinfo{volume}{11},
  \bibinfo{number}{1} (\bibinfo{date}{Jan.} \bibinfo{year}{1964}),
  \bibinfo{pages}{1--9}.
\newblock
\showISSN{0004-5411}
\urldef\tempurl%
\url{https://doi.org/10.1145/321203.321204}
\showDOI{\tempurl}


\bibitem[\protect\citeauthoryear{Rutten}{Rutten}{2000}]%
        {rutten-2000}
\bibfield{author}{\bibinfo{person}{Jan J. M.~M. Rutten}.}
  \bibinfo{year}{2000}\natexlab{}.
\newblock \showarticletitle{Universal coalgebra: a theory of systems}.
\newblock \bibinfo{journal}{\emph{Theoretical Computer Science}}
  \bibinfo{volume}{249}, \bibinfo{number}{1} (\bibinfo{year}{2000}),
  \bibinfo{pages}{3--80}.
\newblock
\urldef\tempurl%
\url{https://doi.org/10.1016/S0304-3975(00)00056-6}
\showDOI{\tempurl}


\bibitem[\protect\citeauthoryear{Salomaa}{Salomaa}{1966}]%
        {Salomaa66}
\bibfield{author}{\bibinfo{person}{Arto Salomaa}.}
  \bibinfo{year}{1966}\natexlab{}.
\newblock \showarticletitle{Two complete axiom systems for the algebra of
  regular events}.
\newblock \bibinfo{journal}{\emph{J. ACM}} \bibinfo{volume}{13},
  \bibinfo{number}{1} (\bibinfo{date}{January} \bibinfo{year}{1966}),
  \bibinfo{pages}{158--169}.
\newblock


\bibitem[\protect\citeauthoryear{Shepherdson and Sturgis}{Shepherdson and
  Sturgis}{1963}]%
        {ShepherdsonSturgis63}
\bibfield{author}{\bibinfo{person}{John~C. Shepherdson} {and}
  \bibinfo{person}{Howard~E. Sturgis}.} \bibinfo{year}{1963}\natexlab{}.
\newblock \showarticletitle{Computability of Recursive Functions}.
\newblock \bibinfo{journal}{\emph{J. ACM}} \bibinfo{volume}{10},
  \bibinfo{number}{2} (\bibinfo{year}{1963}), \bibinfo{pages}{217--255}.
\newblock
\urldef\tempurl%
\url{https://doi.org/10.1145/321160.321170}
\showDOI{\tempurl}


\bibitem[\protect\citeauthoryear{Silva}{Silva}{2010}]%
        {silva-thesis}
\bibfield{author}{\bibinfo{person}{Alexandra Silva}.}
  \bibinfo{year}{2010}\natexlab{}.
\newblock \emph{\bibinfo{title}{Kleene Coalgebra}}.
\newblock \bibinfo{thesistype}{Ph.D. Dissertation}. \bibinfo{school}{Radboud
  University}.
\newblock


\bibitem[\protect\citeauthoryear{Smolka, Kumar, Kahn, Foster, Hsu, Kozen, and
  Silva}{Smolka et~al\mbox{.}}{2019}]%
        {mcnetkat}
\bibfield{author}{\bibinfo{person}{Steffen Smolka}, \bibinfo{person}{Praveen
  Kumar}, \bibinfo{person}{David~M. Kahn}, \bibinfo{person}{Nate Foster},
  \bibinfo{person}{Justin Hsu}, \bibinfo{person}{Dexter Kozen}, {and}
  \bibinfo{person}{Alexandra Silva}.} \bibinfo{year}{2019}\natexlab{}.
\newblock \showarticletitle{Scalable verification of probabilistic networks}.
  In \bibinfo{booktitle}{\emph{Proc. Programming Language Design and
  Implementation ({PLDI})}}. \bibinfo{publisher}{ACM}, \bibinfo{address}{New
  York, NY, USA}, \bibinfo{pages}{190--203}.
\newblock
\urldef\tempurl%
\url{https://doi.org/10.1145/3314221.3314639}
\showDOI{\tempurl}


\bibitem[\protect\citeauthoryear{Tarjan}{Tarjan}{1975}]%
        {Tarjan75}
\bibfield{author}{\bibinfo{person}{Robert~Endre Tarjan}.}
  \bibinfo{year}{1975}\natexlab{}.
\newblock \showarticletitle{Efficiency of a Good But Not Linear Set Union
  Algorithm}.
\newblock \bibinfo{journal}{\emph{J. ACM}} \bibinfo{volume}{22},
  \bibinfo{number}{2} (\bibinfo{year}{1975}), \bibinfo{pages}{215--225}.
\newblock
\urldef\tempurl%
\url{https://doi.org/10.1145/321879.321884}
\showDOI{\tempurl}


\bibitem[\protect\citeauthoryear{Thompson}{Thompson}{1968}]%
        {thompson-1968}
\bibfield{author}{\bibinfo{person}{Ken Thompson}.}
  \bibinfo{year}{1968}\natexlab{}.
\newblock \showarticletitle{Regular Expression Search Algorithm}.
\newblock \bibinfo{journal}{\emph{Commun. ACM}} \bibinfo{volume}{11},
  \bibinfo{number}{6} (\bibinfo{year}{1968}), \bibinfo{pages}{419--422}.
\newblock
\urldef\tempurl%
\url{https://doi.org/10.1145/363347.363387}
\showDOI{\tempurl}


\bibitem[\protect\citeauthoryear{Varacca and Winskel}{Varacca and
  Winskel}{2006}]%
        {VaraccaWinskel06}
\bibfield{author}{\bibinfo{person}{Daniele Varacca} {and}
  \bibinfo{person}{Glynn Winskel}.} \bibinfo{year}{2006}\natexlab{}.
\newblock \showarticletitle{Distributing probability over non-determinism}.
\newblock \bibinfo{journal}{\emph{Mathematical Structures in Computer Science}}
  \bibinfo{volume}{16}, \bibinfo{number}{1} (\bibinfo{year}{2006}),
  \bibinfo{pages}{87--113}.
\newblock
\urldef\tempurl%
\url{https://doi.org/10.1017/S0960129505005074}
\showDOI{\tempurl}


\bibitem[\protect\citeauthoryear{Williams and Ossher}{Williams and
  Ossher}{1978}]%
        {williams.ossher:conversion}
\bibfield{author}{\bibinfo{person}{M. Williams} {and} \bibinfo{person}{H.
  Ossher}.} \bibinfo{year}{1978}\natexlab{}.
\newblock \showarticletitle{Conversion of unstructured flow diagrams into
  structured form}.
\newblock \bibinfo{journal}{\emph{Comput. J.}} \bibinfo{volume}{21},
  \bibinfo{number}{2} (\bibinfo{year}{1978}), \bibinfo{pages}{161--167}.
\newblock
\urldef\tempurl%
\url{https://doi.org/10.1093/comjnl/21.2.161}
\showDOI{\tempurl}


\end{thebibliography}

\extended{\begin{appendices}
\crefalias{section}{appsec}

\section{Omitted Proofs}%
\label{apx:omitted}

\soundcompleteforrel*
\begin{proof}
Recall from \Cref{rem:kat-connection} that there is a language-preserving
map $\phi$ from GKAT to KAT expressions.
As with  GKAT's language model, GKAT's relational model is inherited from
KAT;\@ that is, $\rden{i}{-} = \krden{i}{-} \circ \phi$.
Thus, the claim follows by \citet{KS96a}, who showed the equivalent of
\Cref{thm:sound-complete-for-rel} for KAT:\@ \[
  \KK\den{e} = \KK\den{f} \iff \forall i.\, \krden{i}{e} = \krden{i}{f}.
\qedhere
\]
\end{proof}

\begin{lemma}%
\label{lem:pmodel-well-defined}
$\pden{i}{e}$ is a well-defined subprobability kernel. In particular,
$\pden{i}{{(e +_b 1)}^n \cdot \bneg{b}}(\sigma)(\sigma')$ increases monotonically
in $n$ and the limit for $n\to\infty$ exists.
\end{lemma}
\begin{proof}
We begin by showing the first claim by well-founded induction on $\prec$,
the smallest partial order subsuming the subterm order and satisfying \[
  {(e +_b 1)}^n \cdot \bneg{b} \prec e^{(b)}
\] for all $e, b, n$.
The claim is obvious except when $e = f^{(b)}$.
In that case, we have by induction hypothesis that
$F_n \defeq \pden{i}{{(f +_b 1)}^n \cdot \bneg{b}}(\sigma)(\sigma')$ is well
defined and bounded above by $1$ for all $n$.
To establish that $\lim_{n \to \infty} F_n$ exist and is also bounded above by
$1$, it then suffices to show the claim that $F_n$ increases monotonically
in $n$.

If $F_n=0$ then $F_{n+1} \geq F_n$ holds trivially, so assume $F_n > 0$.
This implies that $\sigma' \in \sat^\dagger(\bneg{b})$. Thus
\begin{align*}
F_{n+1}
&= \pden{i}{{(f +_b 1)}^{n+1} \cdot \bneg{b}}(\sigma)(\sigma')
  \tag{def.}\\
&= \pden{i}{{(f +_b 1)}^{n+1}}(\sigma)(\sigma')
  \tag{$\sigma' \in \sat^\dagger(\bneg{b})$}\\ 
&= {\textstyle\sum_{\sigma''}} \pden{i}{{(f +_b 1)}^{n}}(\sigma)(\sigma'') \cdot
  \pden{i}{f+_b 1}(\sigma'')(\sigma')
  \tag{def.}\\
&\geq \pden{i}{{(f +_b 1)}^{n}}(\sigma)(\sigma') \cdot \pden{i}{f+_b 1}(\sigma')(\sigma')
  \tag{nonnegativity}\\
&= \pden{i}{{(f +_b 1)}^{n}}(\sigma)(\sigma')
  \tag{$\sigma' \in \sat^\dagger(\bneg{b})$}\\ 
&= F_n
&&&&\qedhere
\end{align*}
\end{proof}

\soundcompleteforprob*
\begin{proof}
By mutual implication.
\begin{itemize}
\item[$\Rightarrow$:]
  For soundness, we will define a map $
    \kappa_i \colon \Gs \to \State \to \Dist(\State)
  $ that interprets guarded strings as sub-Markov kernels, and
  lift it to languages via pointwise summation: \[
    \kappa_i(L) \defeq \sum_{w \in L} \kappa_i(w)
  \]
  To establish the claim, we will then show the following equality:
  \begin{equation}
  \label{eq:prob-interp-sound}
    \pden{i}{-} = \kappa_i \circ \den{-}
  \end{equation}

  We define $\kappa_i \colon \Gs \to \State \to \Dist(\State)$ inductively
  as follows:
  \begin{align*}
    \kappa_i(\alpha)(\sigma) &\defeq
      [\sigma \in \sat^\dagger(\alpha)] \cdot \delta_\sigma\\
    \kappa_i(\alpha p w)(\sigma)(\sigma') &\defeq
      [\sigma \in \sat^\dagger(\alpha)] \cdot
      \sum_{\sigma''} \eval(p)(\sigma)(\sigma'') \cdot
      \kappa_i(w)(\sigma'')(\sigma)
  \end{align*}
  To prove \Cref{eq:prob-interp-sound}, it suffices to establish the following equations:
  \begin{align}\label{eq:4-1}
    \kappa_i(\den{p}) &=
      \eval(p)\\\label{eq:4-2}
    \kappa_i(\den{b})(\sigma) &=
      [\sigma \in \sat^\dagger(b)] \cdot \delta_{\sigma}\\\label{eq:4-3}
    \kappa_i(\den{e \cdot f})(\sigma)(\sigma') &=
      \sum_{\sigma''} \kappa_i(\den{e})(\sigma)(\sigma'') \cdot
      \kappa_i(\den{f})(\sigma'')(\sigma')\\\label{eq:4-4}
    \kappa_i(\den{e +_b f})(\sigma) &=
      [\sigma \in \sat^\dagger(b)] \cdot \kappa_i(\den{e})(\sigma)
      \mathrlap{{}+{}}\phantom{{}={}}  [\sigma \in \sat^\dagger(\bneg{b})] \cdot \kappa_i(\den{f})(\sigma)
  \end{align}
  From there, \Cref{eq:prob-interp-sound} follows by a straightforward
  well-founded induction on $\prec$, the partial from the proof of
  \Cref{lem:pmodel-well-defined}.

  For \Cref{eq:4-1}, we have
  \begin{align*}
    \kappa_i(\den{p})(\sigma)(\sigma')
    &= \sum_{\alpha,\beta} \kappa_i(\alpha p \beta)(\sigma)(\sigma')\\
    &= \sum_{\alpha,\beta,\sigma''} [\sigma \in \sat^\dagger(\alpha)] \cdot
      \eval(p)(\sigma)(\sigma'') \cdot
      [\sigma'' \in \sat^\dagger(\beta)] \cdot \delta_{\sigma'}(\sigma'')\\
    &= \sum_{\alpha,\beta} [\sigma \in \sat^\dagger(\alpha)] \cdot
      \eval(p)(\sigma)(\sigma') \cdot
      [\sigma' \in \sat^\dagger(\beta)]\\
    &= \eval(p)(\sigma)(\sigma'),
  \end{align*}
  where the last line follows because
  $\set{\sat^\dagger(b)}{b \in \Bexp} \subseteq 2^\State$
  is a Boolean algebra of sets  with atoms
  $\sat^\dagger(\alpha)$, $\alpha \in \At$, meaning that \[
    \State = \biguplus_{\alpha \in \At} \sat^\dagger(\alpha)
    \quad \text{ and thus } \quad
    \sum_{\alpha} [\sigma \in \sat^\dagger(\alpha)] = 1.
  \]

  For \Cref{eq:4-2}, we have
  \begin{align*}
  \kappa_i(\den{b})(\sigma)
    = \sum_{\alpha \leq b} \kappa_i(\alpha)(\sigma)
    = \sum_{\alpha \leq b} [\sigma \in \sat^\dagger(\alpha)] \cdot \delta_\sigma
    &= [\sigma \in \biguplus_{\alpha \leq b} \sat^\dagger(\alpha)] \cdot \delta_\sigma \\
    &= [\sigma \in \sat^\dagger(b)] \cdot \delta_\sigma.
  \end{align*}

  For  \Cref{eq:4-3}, we need the following auxiliary facts:
  \begin{enumerate}[label={(\textbf{{A}\arabic*})}, align=left, leftmargin=4ex]
  \item $\kappa_i(\alpha x)(\sigma)(\sigma') =
    [\sigma \in \sat^\dagger(\alpha)] \cdot \kappa_i(\alpha x)(\sigma)(\sigma')$
  \item $\kappa_i(x\alpha)(\sigma)(\sigma') =
    [\sigma' \in \sat^\dagger(\alpha)] \cdot \kappa_i(x\alpha)(\sigma)(\sigma')$
  \item $\kappa_i(x\alpha y)(\sigma)(\sigma') =
    \sum_{\sigma''} \kappa_i(x\alpha)(\sigma)(\sigma'') \cdot
    \kappa_i(\alpha y)(\sigma'')(\sigma')$
  \item $\den{e} \diamond \den{f} \cong
    \set{(\alpha, x\alpha, \alpha y)}{\alpha \in \At, x\alpha \in \den{e},
      \alpha y \in \den{f}}$
  \end{enumerate}
  Fact (A1) is immediate by definition of $\kappa_i$, and facts (A2)
  and (A3) follow by straightforward inductions on $|x|$. We defer the proof
  of (A4) to \Cref{lem:seq-isomorphism}.
  We can then compute:
  \begin{align*}
    &\phantom{{}={}} \kappa_i(\den{e\cdot f})(\sigma)(\sigma')\\
    &= \sum_{w \in \den{e} \diamond \den{f}} \kappa_i(w)(\sigma)(\sigma')\\
    &= \sum_{\alpha \in \At} \sum_{x\alpha \in \den{e}} \sum_{\alpha y \in \den{f}}
      \kappa_i(x\alpha y)(\sigma)(\sigma')
      \tag{by A4}\\
    &= \sum_{\alpha \in \At} \sum_{x\alpha \in \den{e}} \sum_{\alpha y \in \den{f}}
      \sum_{\sigma''} \kappa_i(x\alpha)(\sigma)(\sigma'') \cdot
        \kappa_i(\alpha y)(\sigma'')(\sigma')
      \tag{by A3}\\
    &= \sum_{\sigma''} \sum_{\alpha,\beta \in \At}
      \sum_{x\alpha \in \den{e}} \sum_{\alpha y \in \den{f}}
      [\alpha=\beta] \cdot
        \kappa_i(x\alpha)(\sigma)(\sigma'') \cdot
        \kappa_i(\beta y)(\sigma'')(\sigma')
  \end{align*}
  and, observing that
  \begin{align*}
    &\phantom{{}\implies{}} \kappa_i(x\alpha)(\sigma)(\sigma'') \cdot
    \kappa_i(\beta y)(\sigma'')(\sigma') \neq 0\\
    &\implies \sigma'' \in \sat^\dagger(\alpha) \land \sigma'' \in \sat^\dagger(\beta)
      \tag{by A1 and A2}\\
    &\implies \sigma'' \in \sat^\dagger(\alpha\cdot\beta)
      \tag{Boolean algebra}\\
    &\implies \alpha = \beta
      \tag{$\alpha, \beta \in \At$} 
  \end{align*}
  we obtain  \Cref{eq:4-3}:
  \begin{align*}
    \kappa_i(\den{e\cdot f})(\sigma)(\sigma')
    &= \sum_{\sigma''}
      \sum_{\substack{\alpha,\beta \in \At\\x\alpha \in \den{e}\\\beta y \in \den{f}}}
        \kappa_i(x\alpha)(\sigma)(\sigma'') \cdot
        \kappa_i(\beta y)(\sigma'')(\sigma')\\
    &= \sum_{\sigma''}
        \kappa_i(\den{e})(\sigma)(\sigma'') \cdot
        \kappa_i(\den{f})(\sigma'')(\sigma').
  \end{align*}

  For  \Cref{eq:4-4}, we need the following identity (for all $\alpha, x, b, \sigma$):
  \begin{equation}
  \label{eq:aux-prop-sound}
    [\alpha \leq b] \cdot \kappa_i(\alpha x)(\sigma)
    = [\sigma \in \sat^\dagger(b)] \cdot \kappa_i(\alpha x)(\sigma)
  \end{equation}
  Using A1, it suffices to show the equivalence
  \begin{align*}
    \alpha \leq b \land \sigma \in \sat^\dagger(\alpha)
    \quad \iff \quad
    \sigma \in \sat^\dagger(b) \land \sigma \in \sat^\dagger(\alpha)
  \end{align*}
  The implication from left to right follows directly by monotonicity of
  $\sat^\dagger$.
  For the implication from right to left, we have that either $\alpha \leq b$
  or $\alpha \leq \bneg{b}$. Using again monotonicity of $\sat^\dagger$, the
  possibility $\alpha \leq \bneg{b}$ is seen to cause a contradiction.

  With Identity~\cref{eq:aux-prop-sound} at our disposal,  \Cref{eq:4-4} is easy to establish:
  \begin{align*}
  &\phantom{{}={}} \kappa_i(\den{e +_b f})(\sigma)\\
  &= \sum_{w \in \den{e +_b f}} \kappa_i(w)(\sigma)\\
  &= \sum_{\alpha x \in \den{e}} [\alpha \leq b] \cdot \kappa_i(\alpha x)(\sigma)
    + \sum_{\beta y \in \den{f}} [\alpha \leq \bneg{b}] \cdot \kappa_i(\beta y)(\sigma)\\
  &= \sum_{\alpha x \in \den{e}} [\sigma \in \sat^\dagger(b)] \cdot
      \kappa_i(\alpha x)(\sigma)
    + \sum_{\beta y \in \den{f}} [\sigma \in \sat^\dagger(\bneg{b})] \cdot
      \kappa_i(\beta y)(\sigma)\\
  &= [\sigma \in \sat^\dagger(b)] \cdot \kappa_i({\den{e}})(\sigma) +
    [\sigma \in \sat^\dagger(\bneg{b})] \cdot \kappa_i({\den{f}})(\sigma)
  \end{align*}
  This concludes the soundness proof.

\item[$\Leftarrow$:]
  For completeness, we will exhibit an interpretation $i$ over the state space
  $\Gs$ such that
  \begin{equation}
    \label{eq:prob-interp-compl}
    \den{e} = \set{\alpha x \in \Gs}{\pden{i}{e}(\alpha)(\alpha x) \neq 0}.
  \end{equation}
  Define $i \defeq (\Gs, \eval, \sat)$, where
  \begin{mathpar}
    \eval(p)(w) \defeq \Unif(\set{wp\alpha}{\alpha \in \At})
    \and
    \sat(t) \defeq \set{x\alpha \in \Gs}{\alpha \leq t}
  \end{mathpar}
  We need two auxiliary observations:

  \begin{enumerate}[label={(\textbf{{A}\arabic*})}, align=left, leftmargin=4ex]
    \item $\alpha \in \sat^\dagger(b) \iff \alpha \in \den{b}$\label{fact:A1}
    \item Monotonicity: $\pden{i}{e}(v)(w) \neq 0 \implies \exists x.\, w=vx$.
  \end{enumerate}

  They follow by straightforward inductions on $b$ and $e$, respectively.
  To establish \Cref{eq:prob-interp-compl}, it suffices to show the following
  equivalence for all $x,y \in {(\At \cup \Sigma)}^*$: \[
    \pden{i}{e}(x\alpha)(x\alpha y) \neq 0 \quad \iff \quad
     \alpha y \in \den{e}
  \]
  We proceed by well-founded induction on the ordering $\prec$ on expressions
  from the proof of \Cref{lem:pmodel-well-defined}:
  \begin{enumerate}[\textbullet]
  \item
  For $e=b$, we use fact~\ref{fact:A1} to derive that \[
    \pden{i}{b}(x\alpha) = [\alpha \in \sat^\dagger(b)] \cdot \delta_{x\alpha}
    = [\alpha \in \den{b}] \cdot \delta_{x\alpha}.
  \]
  Thus we have \[
    \pden{i}{b}(x\alpha)(x\alpha y) \neq 0
    \ \iff \
    y=\epsilon \land \alpha \in \den{b}
    \ \iff \
    \alpha y \in \den{b}.
  \]

  \item
  For $e=p$, we have that \[
    \pden{i}{p}(x\alpha) = \Unif(\set{x\alpha p\beta}{\beta \in \At}).
  \]
  It follows that \[
    \pden{i}{p}(x\alpha)(x\alpha y) \neq 0
    \ \iff \
    \exists \beta.\, y=p\beta
    \ \iff \
    \alpha y \in \den{p}.
  \]

  \item
  For $e +_b f$, we have that \[
    \pden{i}{e +_b f}(x \alpha)(x\alpha y) = \begin{cases}
      \pden{i}{e}(x \alpha)(x\alpha y) &\text{if } \alpha \in \sat^\dagger(b)\\
      \pden{i}{f}(x \alpha)(x\alpha y) &\text{if } \alpha \in \sat^\dagger(\bneg{b})
    \end{cases}
  \]
  We will argue the case $\alpha \in \sat^\dagger(b)$ explicitly; the
  argument for the case $\alpha \in \sat^\dagger(\bneg{b})$ is analogous.
  We compute:
  \begin{align*}
    \pden{i}{e +_b f}(x\alpha)(x\alpha y) \neq 0 \
    &\iff \ \pden{i}{e}(x\alpha)(x\alpha y) \neq 0
      \tag{premise}\\
    &\iff \ \alpha y \in \den{e}
      \tag{ind. hypothesis}\\
    &\iff \ \alpha y \in \den{b} \diamond \den{e}
      \tag{A1 and premise}\\
    &\iff \ \alpha y \in \den{e +_b f}
      \tag{A1 and premise}
  \end{align*}

  \item
  For $e \cdot f$, recall that
  \begin{align*}
    \pden{i}{e\cdot f}(x\alpha)(x\alpha y)
    &=\sum_{w} \pden{i}{e}(x\alpha)(w) \cdot \pden{i}{e}(w)(x\alpha y).
  \end{align*}
  Thus,
  \begin{align*}
    &\phantom{{}\iff{}} \pden{i}{e\cdot f}(x\alpha)(x\alpha y) \neq 0\\
    &\iff \exists w.\;
      \pden{i}{e}(x\alpha)(w) \neq 0
      \land
      \pden{i}{e}(w)(x\alpha y) \neq 0
      \tag{arg. above}\\ 
    &\iff \exists z.\;
      \pden{i}{e}(x\alpha)(x\alpha z) \neq 0
      \land
      \pden{i}{f}(x\alpha z)(x\alpha y) \neq 0
      \tag{A2}\\
    &\iff \exists z.\;
      \alpha z \in \den{e}
      \land
      \pden{i}{f}(x\alpha z)(x\alpha y) \neq 0
      \tag{ind. hypothesis}\\
    &\iff \exists z, \beta.\;
      z\beta \in \den{e}
      \land
      \pden{i}{f}(xz\beta)(x\alpha y) \neq 0
      \tag{A2}\\
    &\iff \exists z, z',\beta.\;
      z\beta \in \den{e}
      \land
      \pden{i}{f}(xz\beta)(xz\beta z') \neq 0
      \land \alpha y = z\beta z'
      \tag{A2}\\
    &\iff \exists z, z',\beta.\;
      z\beta \in \den{e}
      \land
      \beta z' \in \den{f}
      \land \alpha y = z\beta z'
      \tag{ind. hypothesis}\\
    &\iff \alpha y \in \den{e \cdot f}
      \tag{def. $\den{-}$, $\diamond$}
  \end{align*}
  \item
  For $e^*$, recall that \[
    \pden{i}{e^*}(x\alpha)(x\alpha y) = \lim_{n \to \infty}
      \pden{i}{{(e +_b 1)}^n \cdot \bneg{b}}(x\alpha)(x\alpha y)
  \]
  Since this is the limit of a monotone sequence by
  \Cref{lem:pmodel-well-defined}, it follows that
  \begin{align*}
    &\phantom{{}\iff{}} \pden{i}{e^*}(x\alpha)(x\alpha y) \neq 0\\
    &\iff \exists n.\;
      \pden{i}{{(e +_b 1)}^n \cdot \bneg{b}}(x\alpha)(x\alpha y) \neq 0
      \tag{arg. above}\\ 
    &\iff \exists n.\;
      \alpha y \in \den{{(e +_b 1)}^n \cdot \bneg{b}}
      \tag{ind. hypothesis}\\
    &\iff \exists m.\;
      \alpha y \in \den{{(be)}^m \cdot \bneg{b}}
      \tag{to be argued}\\
    &\iff \alpha y \in \den{e^{(b)}}
      \tag{def. $\den{-}$}
  \end{align*}
  The penultimate step is justified by the identity
  \begin{equation}
  \label{eq:aux-prob-compl}
    \den{{(e +_b 1)}^n \cdot \bneg{b}} =
      \bigcup_{m=0}^{n} \den{{(be)}^m \cdot \bneg{b}},
  \end{equation}
  which we establish by induction on $n \geq 0$:

  The case $n=0$ is trivial. For $n > 0$, we have the following KAT
  equivalence:
  \begin{align*}
    {(be + \bneg{b})}^n \cdot \bneg{b}
    &\equiv (be + \bneg{b}) \cdot {(be + \bneg{b})}^{n-1} \cdot \bneg{b}\\
    &\equiv (be + \bneg{b}) \cdot \sum_{m=0}^{n-1} {(be)}^m \cdot \bneg{b}
      \tag{ind. hypothesis}\\
    &\equiv \sum_{m=0}^{n-1} (be) \cdot {(be)}^m \cdot \bneg{b} +
      \sum_{m=0}^{n-1} \bneg{b} \cdot {(be)}^m \cdot \bneg{b}\\
    &\equiv \sum_{m=1}^{n} {(be)}^m \cdot \bneg{b} + \bneg{b}\equiv \sum_{m=0}^{n} {(be)}^m \cdot \bneg{b},
  \end{align*}
  where the induction hypothesis yields the KAT equivalence \[
    {(be + \bneg{b})}^{n-1} \cdot \bneg{b} \equiv
      \sum_{m=0}^{n-1} {(be)}^m \cdot \bneg{b}
  \] thanks to completeness of the KAT axioms for the language model.
  The claim follows by soundness of the KAT axioms.
  \end{enumerate}
\end{itemize}
This concludes the proof of \Cref{thm:sound-complete-for-prob}.
\end{proof}

\begin{lemma}%
\label{lem:seq-isomorphism}
For deterministic languages $L,K \in \L$, we have the following isomorphism: \[
  L \diamond K \ \cong \
    \set{(\alpha, x\alpha, \alpha y)}{\alpha \in \At, x\alpha \in L,
      \alpha y \in K}
\]
\end{lemma}
\begin{proof}
We clearly have a surjective map \[
  (\alpha, x\alpha, y\alpha) \mapsto x\alpha y
\]
from right to left. To see that this map is also injective, we show that for all
$x_1 \alpha, x_2\beta \in L$ and $\alpha y_1, \beta y_2 \in K$ satisfying
$x_1 \alpha y_1 = x_2 \beta y_2$, we must have
$(\alpha, x_1\alpha, \alpha y_2) = (\beta, x_2\beta,\beta y_2)$.
This is obvious when $|x_1| = |x_2|$, so assume $|x_1| \neq |x_2|$.
We will show that this is impossible.

\WLOG we have  $|x_1| < |x_2|$. By the assumed equality, it follows that $x_2$ 
must be of the form $x_2 = x_1 \alpha z$ for some $z$, and further
$zy_1 = \beta y_2$. Now consider the language \[
  L_{x_1} \defeq \set{w \in \Gs}{x_1 w \in L}.
\]
The language is deterministic, and it contains both $\alpha$ and
$\alpha z \beta$; but the latter contradicts the former.
\end{proof}

\soundness*
\begin{proof}
  Formally, the proof proceeds by induction on the construction of $\equiv$ as a congruence.
  Practically, it suffices to verify soundness of each rule---the inductive cases of the congruence are straightforward because of how $\sem{-}$ is defined.
  \begin{itemize}
    \item[(\nameref{ax:idemp})] For $e +_b e \equiv e$, we derive
      \begin{align*}
        \sem{e +_b e} &= \sem{e} +_{\sem{b}} \sem{e}
                        \tag{Def. $\sem{-}$} \\
                      &= \sem{b} \diamond \sem{e} \cup \bneg{\sem{b}} \diamond \sem{e}
                        \tag{Def. $+_B$} \\
                      &= (\sem{b} \cup \bneg{\sem{b}}) \diamond \sem{e}
                        \tag{Def. $\diamond$} \\
                      &= \At \diamond \sem{e}
                        \tag{Def. $\bneg{B}$} \\
                      &= \sem{e}
                        \tag{Def. $\diamond$}
      \intertext{%
    \item[(\nameref{ax:skewcomm})] For $e +_b f \equiv f +_{\bneg{b}} e$, we derive
      }
        \sem{e +_b f} &= \sem{e} +_{\sem{b}} \sem{f}
                        \tag{Def. $\sem{-}$} \\
                      &= \sem{b} \diamond \sem{e} \cup \bneg{\sem{b}} \diamond \sem{f}
                        \tag{Def. $+_B$} \\
                      &= \bneg{\sem{b}} \diamond \sem{f} \cup \sem{b} \diamond \sem{e}
                        \tag{Def. $\cup$} \\
                      &= \sem{\bneg{b}} \diamond \sem{f} \cup \bneg{\sem{\bneg{b}}} \diamond \sem{e}
                        \tag{Def. $\sem{-}$, $\bneg{B}$} \\
                      &= \sem{f} +_{\sem{\bneg{b}}} \sem{e}
                        \tag{Def. $+_B$} \\
                      &= \sem{f +_{\bneg{b}} e}
                        \tag{Def. $\sem{-}$}
      \intertext{%
    \item[(\nameref{ax:skewassoc})] For $(e +_b f) +_c g \equiv e +_{bc} (f +_c g)$, we derive
      }
        \sem{(e +_b f) +_c g}
            &= \sem{c} \diamond (\sem{b} \diamond \sem{e} \cup \bneg{\sem{b}} \diamond \sem{f}) \cup \bneg{\sem{c}} \diamond \sem{g}
                \tag{Def. $\sem{-}$} \\
            &= \sem{b} \diamond \sem{c} \diamond \sem{e} \cup \bneg{\sem{b}} \diamond \sem{c} \diamond \sem{f} \cup \bneg{\sem{c}} \diamond \sem{g}
                \tag{Def. $\diamond$} \\
            &= \sem{bc} \diamond \sem{e} \cup \sem{\bneg{bc}} \diamond (\sem{c} \diamond \sem{f} \cup \bneg{\sem{c}} \diamond \sem{g})
                \tag{Def. $\sem{-}$, $\diamond$} \\
            &= \sem{e} +_{\sem{bc}} (\sem{f} +_{\sem{c}} \sem{g})
                \tag{Def. $+_B$} \\
            &= \sem{e +_{bc} (f +_c g)}
                \tag{Def. $\sem{-}$}
      \intertext{%
    \item[(\nameref{ax:guard-if})] For $e +_b f \equiv be +_b f$, we derive
      }
        \sem{e +_b f} &= \sem{e} +_{\sem{b}} \sem{f}
                        \tag{Def. $\sem{-}$} \\
                      &= \sem{b} \diamond \sem{e} \cup \bneg{\sem{b}} \diamond \sem{f}
                        \tag{Def. $+_B$} \\
                      &= \sem{b} \diamond (\sem{b} \diamond \sem{e}) \cup \bneg{\sem{b}} \diamond \sem{f}
                        \tag{Def. $\diamond$} \\
                      &= (\sem{b} \diamond \sem{e}) +_{\sem{b}} \sem{f}
                        \tag{Def. $+_B$} \\
                      &= \sem{be +_b f}
                        \tag{Def. $\sem{-}$}
      \intertext{%
    \item[(\nameref{ax:rightdistr})] For $(e +_b f) \cdot g \equiv eg +_b fg$, we derive
      }
        \sem{(e +_b f) \cdot g} &= (\sem{e} +_{\sem{b}} \sem{f}) \diamond \sem{g}
                                    \tag{Def. $\sem{-}$} \\
                                &= (\sem{b} \diamond \sem{e} \cup \bneg{\sem{b}} \diamond \sem{f}) \diamond \sem{g}
                                    \tag{Def. $+_B$} \\
                                &= (\sem{b} \diamond \sem{e}) \diamond \sem{g} \cup (\bneg{\sem{b}} \diamond \sem{f}) \diamond \sem{g}
                                    \tag{Def. $\diamond$} \\
                                &= \sem{b} \diamond (\sem{e} \diamond \sem{g}) \cup \bneg{\sem{b}} \diamond (\sem{f} \diamond \sem{g})
                                    \tag{Def. $\diamond$} \\
                                &= (\sem{e} \diamond \sem{g}) +_{\sem{b}} (\sem{f} \diamond \sem{g})
                                    \tag{Def. $\diamond$} \\
                                &= \sem{eg +_b fg}
                                    \tag{Def. $\sem{-}$}
      \intertext{%
    \item[(\nameref{ax:seqassoc})] For $(e \cdot f) \cdot g \equiv e \cdot (f \cdot g)$, we derive
      }
        \sem{(e \cdot f) \cdot g} &= (\sem{e} \diamond \sem{f}) \diamond \sem{g}
                                    \tag{Def. $\sem{-}$} \\
                                  &= \sem{e} \diamond (\sem{f} \diamond \sem{g})
                                    \tag{Def. $\diamond$} \\
                                  &= \sem{e \cdot (f \cdot g)}
                                    \tag{Def. $\sem{-}$}
      \intertext{%
    \item[(\nameref{ax:absleft})] For $0 \cdot e \equiv 0$, we derive
      }
        \sem{0 \cdot e} &= \sem{0} \diamond \sem{e}
                            \tag{Def. $\sem{-}$} \\
                        &= \emptyset \diamond \sem{e}
                            \tag{Def. $\sem{-}$} \\
                        &= \emptyset
                            \tag{Def. $\diamond$} \\
                        &= \sem{0}
                            \tag{Def. $\sem{-}$}
      \intertext{%
    \item[(\nameref{ax:absright})] The proof for $e \cdot 0 \equiv 0$ is similar to the above.
    \item[(\nameref{ax:neutrleft})] For $1 \cdot e \equiv e$, we derive
      }
        \sem{1 \cdot e} &= \sem{1} \diamond \sem{e}
                            \tag{Def. $\sem{-}$} \\
                        &= \At \diamond \sem{e}
                            \tag{Def. $\sem{-}$} \\
                        &= \sem{e}
                            \tag{Def. $\diamond$}
      \intertext{%
    \item[(\nameref{ax:neutrright})] The proof for $e \cdot 1 \equiv e$ is similar to the above.
    \item[(\nameref{ax:unroll})] For $e^{(b)} \equiv ee^{(b)} +_b 1$, we derive
      }
        \sem{e^{(b)}} &= \sem{e}^{(\sem{b})}
                        \tag{Def. $\sem{-}$} \\
                      &= \bigcup_{n \geq 0} {(\sem{b} \diamond \sem{e})}^n \diamond \bneg{\sem{b}}
                        \tag{Def. $L^{(B)}$} \\
                      &= \bneg{\sem{b}} \diamond \sem{1} \cup \sem{b} \diamond \sem{e} \diamond \bigcup_{n \geq 0} {(\sem{b} \diamond \sem{e})}^n \diamond \bneg{\sem{b}}
                        \tag{Def. $\diamond$, $L^n$, $\bigcup$} \\
                      &= \bneg{\sem{b}} \diamond \sem{1} \cup \sem{b} \diamond \sem{e} \diamond \sem{e^{(b)}}
                        \tag{Def. $L^{(B)}$} \\
                      &= \sem{e \cdot e^{(b)} +_b 1}
                        \tag{Def. $\sem{-}$, $+_B$}
      \intertext{%
    \item[(\nameref{ax:tighten})] For ${(ce)}^{(b)} \equiv {(e +_c 1)}^{(b)}$, we first argue that if $w \in {(\sem{b} \diamond \sem{c} \diamond \sem{e} \cup \bneg{\sem{c}})}^n$ for some $n$, then $w \in {(\sem{b} \diamond \sem{c} \diamond \sem{e})}^m$ for some $m \leq n$, by induction on $n$.
        In the base, where $n = 0$, we have $w \in \At$; hence, the claim holds immediately.
        For the inductive step, let $n > 0$ and write
        \begin{mathpar}
        w = w_0 \diamond w'
        \and
        w_0 \in \sem{b} \diamond \sem{c} \diamond \sem{e} \cup \bneg{\sem{c}}
        \and
        w' \in {(\sem{b} \diamond \sem{c} \diamond \sem{e} \cup \bneg{\sem{c}})}^{n-1}
        \end{mathpar}
        By induction, we know that $w' \in {(\sem{b} \diamond \sem{c} \diamond \sem{e})}^{m'}$ for $m' \leq n - 1$.
        If $w_0 \in \bneg{\sem{c}}$, then $w = w'$, and the claim goes through if we choose $m = m'$.
        Otherwise, if $w_0 \in \sem{b} \diamond \sem{c} \diamond \sem{e}$, then
        \[
            w
                = w_0 \diamond w
                \in \sem{b} \diamond \sem{c} \diamond \sem{e} \diamond {(\sem{b} \diamond \sem{c} \diamond \sem{e})}^{m'}
                = {(\sem{b} \diamond \sem{c} \diamond \sem{e})}^{m'+1}
        \]
        and thus the claim holds if we choose $m = m'+1$.
        Using this, we derive
      }
        \sem{{(ce)}^{(b)}} &= \sem{ce}^{(\sem{b})}
                            \tag{Def. $\sem{-}$} \\
                         &= \bigcup_{n \geq 0} {(\sem{b} \diamond \sem{c} \diamond \sem{e})}^n \diamond \bneg{\sem{b}}
                            \tag{Def. $L^{(B)}$} \\
                         &= \bigcup_{n \geq 0} {(\sem{b} \diamond \sem{c} \diamond \sem{e} \cup \bneg{\sem{c}})}^n \diamond \bneg{\sem{b}}
                            \tag{above derivation} \\
                         &= {(\sem{c} \diamond \sem{e} \cup \bneg{\sem{c}} \diamond \sem{1})}^{(\sem{b})}
                            \tag{Def. $L^{(B)}$, $\diamond$, $\sem{-}$} \\
                         &= {(\sem{e} +_{\sem{c}} \sem{1})}^{(\sem{b})}
                            \tag{Def. $+_B$} \\
                         &= \sem{{(e +_c 1)}^{(b)}}
                            \tag{Def. $\sem{-}$}
      \end{align*}
      This completes the proof.
    \qedhere
  \end{itemize}
\end{proof}

%
%
%

\ft*
\begin{proof}
By induction on $e$. For a primitive action $p$,
$D_\alpha(p)=(p,1)$, for all $\alpha\in \At$, and $E(p)=0$. Then
\begin{align*}
p &\stackrel{\text{\nameref{fact:trivright}}}\equiv 1 +_0 p \stackrel{\text{Lem.\ref{lemma:dsum-idemp}}}\equiv 1+_0 \dsum_{\alpha\leq 1} \alpha\cdot p \cdot 1 = 1+_{E(p)}\dsum_{\alpha\colon D_\alpha(e) = (p_\alpha, e_\alpha)} p_\alpha \cdot e_\alpha.
\end{align*}
For a primitive test $c$,
$D_\alpha(c)=[\alpha\le c]$ and $E(c)=c$. Then
\begin{align*}
c &\stackrel{\text{\nameref{fact:neutrright2}}}\equiv 1+_c 0 = 1+_{E(c)}\dsum_{\alpha\colon D_\alpha(e) = (p_\alpha, e_\alpha)} p_\alpha \cdot e_\alpha.
\end{align*}
For a conditional $e_1 +_c e_2$, we have inductively:
\begin{align}
e_i &\equiv 1 +_{E(e_i)} \dsum_{\alpha\colon D_\alpha(e_i) = (p_\alpha, e_\alpha)} p_\alpha \cdot e_\alpha, \quad \ i\in\{1,2\}.\label{eq:e1e2}
\end{align}
Then
\begin{align*}
e_1 +_c e_2
&\equiv \left(1 +_{E(e_1)} \dsum_{\alpha\colon D_\alpha(e_1) = (p_\alpha, e_\alpha)} p_\alpha \cdot e_\alpha\right) +_c \left(1 +_{E(e_2)} \dsum_{\alpha\colon D_\alpha(e_2) = (p_\alpha, e_\alpha)} p_\alpha \cdot e_\alpha\right)\tag{Eq. \cref{eq:e1e2}}
\\
&= 1 +_{E(e_1) +_c E(e_2)} \left(\dsum_{\alpha\colon D_\alpha(e_1) = (p_\alpha, e_\alpha)} p_\alpha \cdot e_\alpha\right.+_c\ \left.\dsum_{\alpha\colon D_\alpha(e_1) = (p_\alpha, e_\alpha)} p_\alpha \cdot e_\alpha\right) \tag{skew assoc.}\\
&= 1 +_{E(e_1 +_c e_2)} \dsum_{\alpha\colon D_\alpha(e_1+_c e_2) = (p_\alpha, e_\alpha)} p_\alpha \cdot e_\alpha. \tag{def $D_\alpha(e_1+_c e_2)$}
\end{align*}
For sequential composition $e_1\cdot e_2$, suppose $e_1$ and $e_2$ are decomposed as in~\eqref{eq:e1e2}.
\begin{align*}
&e_1\cdot e_2\\
&\equiv \left(1 +_{E(e_1)} \dsum_{\alpha\colon D_\alpha(e_1) = (p_\alpha, e_\alpha)} p_\alpha \cdot e_\alpha\right)\cdot e_2\tag{Eq. \cref{eq:e1e2}}
\\
&= e_2  +_{E(e_1)} \dsum_{\alpha\colon D_\alpha(e_1) = (p_\alpha, e_\alpha)} p_\alpha \cdot e_\alpha \cdot e_2\tag{right distri. \nameref{ax:rightdistr}}\\
&= \left(1 +_{E(e_2)} \dsum_{\alpha\colon D_\alpha(e_2) = (p_\alpha, e_\alpha)} p_\alpha \cdot e_\alpha\right)+_{E(e_1)} \dsum_{\alpha\colon D_\alpha(e_1) = (p_\alpha, e_\alpha)} p_\alpha \cdot e_\alpha \cdot e_2\tag{Eq. \cref{eq:e1e2}}
\\\\
&= 1 +_{E(e_1)E(e_2)} \left( \left(\dsum_{\alpha\colon D_\alpha(e_2) = (p_\alpha, e_\alpha)} p_\alpha \cdot e_\alpha\right) +_{E(e_1)}  \left(\dsum_{\alpha\colon D_\alpha(e_1) = (p_\alpha, e_\alpha)} p_\alpha \cdot e_\alpha \cdot e_2\right)\right)\tag{skew assoc. \nameref{ax:skewassoc}}
\\
&= 1 +_{E(e_1)E(e_2)} \dsum_{\alpha\colon \substack{D_\alpha(e_1) = (p_\alpha, e_\alpha)\\D_\alpha(e_2) = (p_\alpha, e_\alpha)}} \left(p_\alpha \cdot e_\alpha +_{E(e_1)}  p_\alpha \cdot e_\alpha \cdot e_2\right)\tag{skew assoc. $\dsum$}
\\
&= 1 +_{E(e_1e_2)}  \dsum_{\alpha\colon D_\alpha(e_1e_2) = (p_\alpha, e_\alpha)} p_\alpha \cdot e_\alpha\tag{def $E(e_1\cdot e_2)$ and $D_\alpha(e_1\cdot e_2)$} \\\\ 
\end{align*}
Finally, for a while loop $e^{(c)}$ we will use \Cref{lem:prod} (Productive Loop):
\begin{align*}
e^{(c)}
&\equiv {\left(D(e)\right)}^{(c)} \tag{\Cref{lem:prod}}\\
&\equiv 1 +_{\bar c} D(e)\cdot {\left(D(e)\right)}^{(c)}\tag{\nameref{ax:unroll} and \nameref{ax:skewcomm}}\\
&\equiv 1 +_{\bar c} \left(\dsum_{\alpha\colon D_\alpha(e) = (p_\alpha, x_\alpha)} p_\alpha \cdot x_\alpha\right)e^{(c)}\tag{\Cref{lem:prod} and def. of $D(e)$}\\ 
&\equiv 1 +_{\bar c}  \left(\dsum_{\alpha\colon D_\alpha(e) = (p_\alpha, x_\alpha)} p_\alpha \cdot x_\alpha \cdot e^{(c)}\right) \tag{\nameref{ax:rightdistr}}\\
&= 1 +_{E(e^{(c)})} \dsum_{\alpha\colon D_\alpha(e^{(c)}) = (p_\alpha, e_\alpha)} p_\alpha \cdot e_\alpha. \tag{Def. $D(e^{(c)})$ and $E(e^{(c)})=\bar c$} 
\end{align*}
\end{proof}

\lemmaed*
\begin{proof}
\begin{align*}
E(D(e)) &= E\left(\dsum_{\alpha\colon D_\alpha(e)=(p_\alpha ,e_\alpha)}  p_\alpha \cdot e_\alpha\right) = \sum_{\alpha\colon D_\alpha(e)=(p_\alpha ,e_\alpha)} E(p_\alpha \cdot e_\alpha) = 0\\
\bar{E(e)}\cdot D(e)&   =  \bar{E(e)}\cdot \left( \dsum_{\alpha\colon D_\alpha(e)=(p_\alpha ,e_\alpha)}  p_\alpha \cdot e_\alpha\right) \stackrel{\text{Lem~\ref{lemma:dsum-push}}}\equiv \dsum_{\alpha\colon \substack{D_\alpha(e)=(p_\alpha ,e_\alpha)\\\alpha\leq\bar{E(e)}}} p_\alpha \cdot e_\alpha\stackrel{*}= D(e)\\
\bar{E(e)}\cdot e&   \stackrel{\text{FT}}\equiv  \bar{E(e)}\cdot ( 1 +_{E(e)} D(e)) \stackrel{\text{\nameref{fact:branchsel}}}\equiv D(e)
\end{align*}
Note that for * we use the observation that  for all $\alpha$ such that $D_\alpha(e)=(p_\alpha ,e_\alpha)$ it is immediate that $\alpha \not\leq E(e)$ and hence the condition $\alpha\leq\bar{E(e)}$ is redundant.
\end{proof}

\derivable*
\begin{proof}
  We start by deriving the remaining facts for guarded union.
  \begin{itemize}
    \item[(\nameref{fact:skewassoc-dual})] For $e +_b (f +_c g) \equiv (e +_b f) +_{b + c} g$, we derive
      \begin{align}
        e +_b (f +_c g) &\equiv (g +_{\bneg{c}} f) +_{\bneg{b}} e
                        \tag{\nameref{ax:skewcomm}} \\
                        &\equiv g +_{\bneg{b} \bneg{c}} (f +_{\bneg{b}} e)
                        \tag{\nameref{ax:skewassoc}} \\
                        &\equiv g +_{\bneg{b + c}} (f +_{\bneg{b}} e)
                        \tag{Boolean algebra} \\
                        &\equiv (e +_b f) +_{b + c} g
                        \tag{\nameref{ax:skewcomm}}
      \intertext{%
    \item[(\nameref{fact:guard-if-dual})] For $e +_b f \equiv e +_b \bneg{b}f$, we derive
      }
        e +_b f &\equiv f +_{\bneg{b}} e
                \tag{\nameref{ax:skewcomm}} \\
                &\equiv \bneg{b} f +_{\bneg{b}} e
                \tag{\nameref{ax:guard-if}} \\
                &\equiv e +_b \bneg{b}f
                \tag{\nameref{ax:skewcomm}, Boolean algebra}
      \intertext{%
    \item[(\nameref{fact:leftdistr})] For $b \cdot (e +_b c f) \equiv be +_c bf$, we derive
      }
        b (e +_c f) &\equiv b \cdot (f +_{\bneg{c}} e)
                    \tag{\nameref{ax:skewcomm}} \\
                    &\equiv ((b + c) (b + \bneg{c})) (f +_{\bneg{c}} e)
                    \tag{Boolean algebra} \\
                    &\equiv (b + c) ((b + \bneg{c}) (f +_{\bneg{c}} e))
                    \tag{\nameref{ax:seqassoc}} \\
                    &\equiv (b + c) ((f +_{\bneg{c}} e) +_{b + \bneg{c}} 0 )
                    \tag{\nameref{fact:neutrright2}} \\
                    &\equiv (b + c) (f +_{\bneg{c}} (e +_b 0))
                    \tag{\nameref{fact:skewassoc-dual}} \\
                    &\equiv (b + c) ((e +_b 0) +_c f)
                    \tag{\nameref{ax:skewcomm}} \\
                    &\equiv (b + c) (be +_c f)
                    \tag{\nameref{fact:neutrright2}} \\
                    &\equiv (be +_c f) +_{b + c} 0
                    \tag{\nameref{fact:neutrright2}} \\
                    &\equiv be +_c (f +_b 0)
                    \tag{\nameref{fact:skewassoc-dual}} \\
                    &\equiv be +_c bf
                    \tag{\nameref{fact:neutrright2}}
      \intertext{%
    \item[(\nameref{fact:trivright})] For $e +_0 f \equiv f$, we derive
      }
        e +_0 f &\equiv (0 \cdot e) +_0 f
                \tag{\nameref{ax:guard-if}} \\
                &\equiv 0 +_0 f
                \tag{\nameref{ax:absleft}} \\
                &\equiv (0 \cdot f) +_0 f
                \tag{\nameref{ax:absleft}} \\
                &\equiv f +_0 f
                \tag{\nameref{ax:guard-if}} \\
                &\equiv f
                \tag{\nameref{ax:idemp}}
      \intertext{%
    \item[(\nameref{fact:branchsel})] For $b \cdot (e +_b f) \equiv be$, we derive
      }
        b(e +_b f) &\equiv be +_b bf
                   \tag{\nameref{fact:guard-if-dual}} \\
                   &\equiv be +_b \bneg{b} b f
                   \tag{\nameref{fact:guard-if-dual}} \\
                   &\equiv be +_b 0 f
                   \tag{Boolean algebra} \\
                   &\equiv be +_b 0
                   \tag{\nameref{ax:absleft}} \\
                   &\equiv be
                   \tag{\nameref{fact:neutrright2}}
      \end{align}
  \end{itemize}
  Next, we derive the remaining loop facts.
  \begin{itemize}
    \item[(\nameref{fact:guard-loop-out})] For $e^{(b)} \equiv e^{(b)}\bneg{b}$, we derive
      \begin{align}
        e^{(b)} &\equiv {(D(e))}^{(b)}
                \tag{Productive loop lemma} \\
                &\equiv D(e) {(D(e))}^{(b)} +_b 1
                \tag{\nameref{ax:unroll}} \\
                &\equiv D(e) {(D(e))}^{(b)} +_b \bneg{b}
                \tag{\nameref{fact:guard-if-dual}} \\
                &\equiv {(D(e))}^{(b)} \bneg{b}
                \tag{\nameref{ax:fixpoint}} \\
                &\equiv e^{(b)} \bneg{b}
                \tag{Productive loop lemma}
      \intertext{%
    \item[(\nameref{fact:guard-loop-in})] For $e^{(b)} \equiv {(be)}^{(b)}$, we derive
      }
        e^{(b)} &\equiv {(D(e))}^{(b)}
                \tag{Productive loop lemma} \\
                &\equiv D(e) {(D(e))}^{(b)} +_b 1
                \tag{\nameref{ax:unroll}} \\
                &\equiv b \cdot D(e) {(D(e))}^{(b)} +_b 1
                \tag{\nameref{ax:guard-if}} \\
                &\equiv {(b \cdot D(e))}^{(b)}
                \tag{\nameref{ax:fixpoint}} \\
                &\equiv {(D(be))}^{(b)}
                \tag{Def. $D$} \\
                &\equiv {(be)}^{(b)}
                \tag{Productive loop lemma}
      \intertext{%
    \item[(\nameref{fact:neutral-loop})] For $e^{(0)} \equiv 1$, we derive
      }
        e^{(0)} &\equiv {(0 \cdot e)}^{(0)}
                \tag{\nameref{fact:guard-loop-in}} \\
                &\equiv 0^{(0)}
                \tag{\nameref{ax:absleft}} \\
                &\equiv 0 \cdot 0^{(0)} +_0 1
                \tag{\nameref{ax:unroll}} \\
                &\equiv 0 +_0 1
                \tag{\nameref{ax:absleft}} \\
                &\equiv 1
                \tag{\nameref{fact:trivright}}
      \intertext{%
    \item[(\nameref{fact:absorb-loop})] For $e^{(1)} \equiv 0$, we derive
      }
        e^{(1)} &\equiv e^{(1)} \cdot \bneg{1}
                \tag{\nameref{fact:guard-loop-out}} \\
                &\equiv e^{(1)} \cdot 0
                \tag{Boolean algebra} \\
                &\equiv 0
                \tag{\nameref{ax:absright}}
      \intertext{%
    \item[(\nameref{fact:absorb-loop-bool})] For $b^{(c)} \equiv \bneg{c}$, we derive
      }
        b^{(c)} &\equiv {(D(b))}^{(c)}
                \tag{Productive loop lemma} \\
                &\equiv 0^{(c)}
                \tag{Def. $D$} \\
                &\equiv 0 \cdot 0^{(c)} +_c 1
                \tag{\nameref{ax:unroll}} \\
                &\equiv 0 +_c 1
                \tag{\nameref{ax:absleft}} \\
                &\equiv 1 +_{\bneg{c}} 0
                \tag{\nameref{ax:skewcomm}} \\
                &\equiv \bneg{c} \cdot 1
                \tag{\nameref{fact:neutrright2}} \\
                &\equiv \bneg{c}
                \tag{Boolean algebra}
      \end{align}
      This completes the proof. \qedhere
  \end{itemize}
\end{proof}

\hoarecompleteness*
\begin{proof}
By induction on $e$.
In the base, there are two cases to consider.
\begin{itemize}[left=1ex..1.4\parindent]
    \item
    If $e = d$ for some Boolean $d$, then the claim follows by completeness of
    the Boolean algebra axioms, which $\equiv$ subsumes by definition.

    \item
    If $e = a \in \Sigma$, then $\sem{bec} = \sem{be}$ implies $\sem{c} = \sem{1}$, hence $c \equiv 1$ by completeness of Boolean algebra; the claim then follows.
\end{itemize}
For the inductive step, there are three cases:
\begin{itemize}[left=1ex..1.4\parindent]
    \item
    If $e = e_0 +_d e_1$, then $\sem{bec} = \sem{be}$ implies that $\sem{dbe_{0}c} = \sem{dbe_{0}}$ and $\sem{\bneg{d}be_{1}c} = \sem{\bneg{d}be_{1}}$.
    By induction, we then know that $dbe_{0}c \equiv dbe_0$ and $\bneg{d}be_{1}c \equiv \bneg{d}be_1$.
    We can then derive as follows:
    \begin{align*}
    b(e_0 +_d e_1)c
        &\equiv be_{0}c +_d be_{1}c
            \tag{\nameref{fact:guard-if-dual}} \\
        &\equiv dbe_{0}c +_d \bneg{d}be_{1}c
            \tag{\nameref{ax:guard-if}, \nameref{fact:guard-if-dual}} \\
        &\equiv dbe_{0}c +_d \bneg{d}be_{1}
            \tag{$\bneg{d}be_{1}c \equiv \bneg{d}be_1$} \\ 
        &\equiv dbe_{0} +_d \bneg{d}be_{1}
            \tag{$dbe_{0}c \equiv dbe_0$} \\ 
        &\equiv be_{0} +_d be_1
            \tag{\nameref{ax:guard-if}, \nameref{fact:guard-if-dual}} \\
        &\equiv b \cdot (e_0 +_d e_1)
            \tag{\nameref{fact:guard-if-dual}}
    \end{align*}

    \item
    If $e = e_0 \cdot e_1$, then let
    $
        d = \sum \{ \alpha \in \At : \sem{be_0\alpha} \neq \emptyset \}
    $.
    We then know that $\sem{be_{0}d} = \sem{be_0}$, and hence $be_{0}d \equiv be_0$ by induction.
    We furthermore claim that $\sem{de_{1}c} = \sem{de_1}$.
    To see this, note that if $\alpha{}w\beta \in \sem{de_1}$, then $\alpha \leq d$, and hence there exists an $x\alpha \in \sem{be_{0}\alpha} \subseteq \sem{be_{0}d} = \sem{be_0}$.
    Thus, we know that $x\alpha{}w\beta \in \sem{be_{0}e_1} = \sem{be_{0}e_{1}c}$, meaning that $\beta \leq c$; hence, we know that $\alpha{}w\beta \in \sem{de_{1}c}$.
    By induction, $de_1c \equiv de_1$.
    We then derive:
    \begin{align*}
    be_{0}e_{1}c
        &\equiv be_{0}de_{1}c
            \tag{$be_{0} \equiv be_{0}d$} \\ 
        &\equiv be_{0}de_{1}
            \tag{$de_{1} \equiv de_{1}c$} \\ 
        &\equiv be_{0}e_{1}
            \tag{$be_{0} \equiv be_{0}d$} 
    \end{align*}

    \item
    If $e = e_0^{(d)}$, first note that if $b \equiv 0$, then the claim follows trivially.
    Otherwise, let
    \[
        h = \sum \{ \alpha \in \At : \exists n. \sem{b} \diamond \sem{de_0}^n \diamond \sem{\alpha} \neq \emptyset \}.
    \]
    We make the following observations.
    \begin{enumerate}[(i), leftmargin=1cm]
        \item\label{apx:invariant:start}
        Since $b \not\equiv 0$, we have that $\sem{b} \diamond \sem{{de_0}}^0 \diamond \sem{b} = \sem{b} \neq \emptyset$, and thus $b \leq h$.

        \item\label{apx:invariant:end}
        If $\alpha \leq h\bneg{d}$, then in particular $\gamma{}w\alpha \in \sem{b} \diamond \sem{de_0}^n \diamond \sem{\alpha}$ for some $n$ and $\gamma{}w$.
        Since $\alpha \leq \bneg{d}$, it follows that $\gamma{}w\alpha \in \sem{be_0^{(d)}} = \sem{be_0^{(d)}c}$, and thus $\alpha \leq c$.
        Consequently, $h\bneg{d} \leq c$.

        \item\label{apx:invariant:hold}
        If $\alpha{}w\beta \in \sem{dhe_0}$, then $\alpha \leq h$ and hence there exists an $n$ such that $\gamma{}x\alpha \in \sem{b} \diamond \sem{de_0}^n \diamond \sem{\beta}$.
        But then $\gamma{}x\alpha{}w\beta \in \sem{b}\ \diamond \sem{de_0}^{n+1} \diamond \sem{\beta}$, and therefore $\beta \leq h$.
        We can conclude that $\sem{dhe_0} = \sem{dhe_{0}h}$; by induction, it follows that $dhe_{0}h \equiv dhe_0$.
    \end{enumerate}

    Using these observations and the invariance lemma (\Cref{lem:invariance}), we derive
    \begin{align*}
    be_{0}^{(d)}c
        &\equiv bhe_0^{(d)}c
            \tag{By~\ref{apx:invariant:start}} \\
        &\equiv b \cdot {(he_0)}^{(d)}hc
            \tag{Invariance and~\ref{apx:invariant:hold}} \\
        &\equiv b \cdot {(he_0)}^{(d)}\bneg{d}hc
            \tag{\nameref{fact:guard-loop-out}} \\
        &\equiv b \cdot {(he_0)}^{(d)}\bneg{d}h
            \tag{By~\ref{apx:invariant:end}} \\
        &\equiv b \cdot {(he_0)}^{(d)}h
            \tag{\nameref{fact:guard-loop-out}} \\
        &\equiv bh{e_0}^{(d)}
            \tag{Invariance and~\ref{apx:invariant:hold}} \\
        &\equiv b {e_0}^{(d)}
            \tag{By~\ref{apx:invariant:start}}
    \end{align*}
\end{itemize}
This completes the proof.\qedhere
\end{proof}

\solpreserveslang*
\begin{proof}
We show that \[
  w \in \den{s(x)} \quad \iff \quad w \in \lang\X(x)
\]
for all states $x$ by induction on the length of $w \in \Gs$.
We will use that $w$ is of the form $w=\alpha u$
for some $\alpha \in \At$, $u \in {(At \cdot \Sigma)}^*$
and thus
\begin{align*}
  w \in \den{s(x)}
  &\iff w \in \den{\alpha \cdot s(x)}
    \tag{def. $\den{-}$}\\
  &\iff w \in \den{\floor{\delta^\X(x)(\alpha)}_s}
    \tag{def. sol., soundness}\\
\intertext{
  For $w=\alpha$, we have
}
  \alpha \in \den{\floor{\delta^\X(x)(\alpha)}_s}
  &\iff \delta^\X(x)(\alpha) = 1
    \tag{def. $\floor{-}$ \& $\den{-}$}\\
  &\iff \alpha \in \lang\X(x)
    \tag{def. $\lang\X$}\\
\intertext{
  For $w=\alpha p v$, we have
}
  \alpha pv \in \den{\floor{\delta^\X(x)(\alpha)}_s}
  &\iff \exists y.\; \delta^\X(x)(\alpha) = \angl{p,y} \land v \in \den{s(y)}
    \tag{def. $\floor{-}$ \& $\den{-}$}\\
  &\iff \exists y.\; \delta^\X(x)(\alpha) = \angl{p,y} \land v \in \lang\X(y)
    \tag{induction}\\
  &\iff \alpha p v \in \lang\X(x)
    \tag{def. $\lang\X$}
\end{align*}
This concludes the proof.
\end{proof}

\begin{lemma}%
\label{lem:solution-alt}
Let $\X = \angl{X, \delta^\X}$ be a $G$-coalgebra.
A function $s\colon X \to \Exp$ is a solution to $\X$ if and only if for all $\alpha \in \At$ and $x \in X$ it holds that $\alpha \cdot s(x) \equiv \alpha \cdot \floor{\delta^\X(x)(\alpha)}_s$.
\end{lemma}
\begin{proof}
We shall use some of the observations about $\dsum$ from \Cref{sec:gen_choice}.
\begin{itemize}
    \item[$(\Rightarrow)$]
    Let $s$ be a solution to $\X$; we then derive for $\alpha \in \At$ and $x \in X$ that
    \begin{align*}
    \alpha \cdot s(x)
        &\equiv \alpha \cdot \dsum_{\alpha \leq 1} \floor{\delta^{\X}(x)(\alpha)}_s
            \tag{$s$ solves $\X$} \\
        &\equiv \alpha \cdot \dsum_{\alpha \leq \alpha} \floor{\delta^{\X}(x)(\alpha)}_s
            \tag{\Cref{lemma:dsum-push}} \\
        &\equiv \alpha \cdot \floor{\delta^{\X}(x)(\alpha)}_s
            \tag{Def. $\dsum$, \nameref{fact:branchsel}}
    \intertext{%
    \item[$(\Leftarrow)$]
    Suppose that for all $\alpha \in \At$ and $x \in X$ we have $\alpha \cdot s(x) \equiv \alpha \cdot \floor{\delta^\X(x)(\alpha)}_s$.
    We can then derive
    }
    s(x)
        &\equiv \dsum_{\alpha \leq 1} s(x)
            \tag{\Cref{lemma:dsum-idemp}} \\
        &\equiv \dsum_{\alpha \leq 1} \alpha \cdot s(x)
            \tag{\Cref{lemma:dsum-guarded}} \\
        &\equiv \dsum_{\alpha \leq 1} \alpha \cdot \floor{\delta^\X(x)(\alpha)}_s
            \tag{premise} \\
        &\equiv \dsum_{\alpha \leq 1} \floor{\delta^\X(x)(\alpha)}_s
            \tag{\Cref{lemma:dsum-guarded}}
    \end{align*}
    This completes the proof.
    \qedhere
\end{itemize}
\end{proof}

\solutionsexist*
\begin{proof}
Assume $\X$ is well-nested.
We proceed by rule induction on the well-nestedness derivation.
\begin{enumerate}
  \item[(\ref{rule:s1})]
  Suppose $\delta^\X \colon X \to 2^\At$. Then
  \[
      s^{\X}(x) \defeq \sum \set{ \alpha \in \At }{ \delta^\X(x)(\alpha) = 1 }
  \]
  is a solution to $\X$.

  \item[(\ref{rule:s2})]
  Suppose $\X = (\Y+\Z)\umod{Y,h}$, where
  $h \in G(Y+Z)$ and $\Y$ and $\Z$ are well-nested with solutions
  $s^\Y$ and $s^\Z$. We need to exhibit a solution $s^\X$ to $\X$.
  For $y \in Y$ and $z \in Z$ we define
  \begin{mathpar}
    s^\X(y) \defeq s^\Y(y) \cdot \ell
    \and
    s^\X(z) \defeq s^\Z(z)
    \\
    \ell \defeq {\Bigl( \dsum_{\alpha \leq b} \floor{h(\alpha)}_{s^\Y} \Bigr)}^{(b)}
      \cdot \dsum_{\alpha \leq \bneg{b}} \floor{h(\alpha)}_{s^\Z}
    \and
    b \defeq \sum \set{\alpha \in \At}{h(\alpha) \in \Sigma \times Y}
  \end{mathpar}

  By \Cref{lem:solution-alt}, it then suffices to prove that for $x \in Y+Z$ and $\alpha \in \At$, we have
  \[
      \alpha \cdot s^\X(x)
          \equiv \alpha \cdot \floor{\delta^\X(x)(\alpha)}_{s^\X}
  \]
  There are two cases to distinguish.
  \begin{itemize}
    \item
    If $x \in Z$, then
    \begin{align*}
    \alpha \cdot s^{\X}(x)
        &= \alpha \cdot s^\Z(x)
            \tag{def. $s^\X$} \\
        &\equiv \alpha \cdot \floor{\delta^\Z(x)(\alpha)}_{s^\Z}
            \tag{$s^\Z$ solves $\Z$} \\
        &= \alpha \cdot \floor{\delta^\Z(x)(\alpha)}_{s^\X}
            \tag{def. $s^\X$} \\
        &= \alpha \cdot \floor{\delta^\X(x)(\alpha)}_{s^\X}
            \tag{def. $\X$}
    \end{align*}

    \item
    If $x \in Y$, then we find by choice of $s^\X$ and $s^\Y$ that
    \[
        \alpha \cdot s^\X(x)
            = \alpha \cdot s^\Y(x) \cdot \ell
            = \alpha \cdot \floor{\delta^{\Y}(x)(\alpha)}_{s^\Y} \cdot \ell
    \]
    We distinguish three subcases:
    \begin{itemize}
      \item
        If $\delta^\Y(x)(\alpha) \in \{ 0 \} \cup \Sigma \times Y$ then
        $\delta^\Y(x)(\alpha) = \delta^\X(x)(\alpha)$ and thus
      \begin{align*}
        \alpha \cdot \floor{\delta^{\Y}(x)(\alpha)}_{s^\Y} \cdot \ell
        &= \alpha \cdot \floor{\delta^\X(x)(\alpha)}_{s^\Y} \cdot \ell
            \tag{def. $\X$} \\
        &\equiv \alpha \cdot \floor{\delta^\X(x)(\alpha)}_{s^\X}
            \tag{def. $s^\X$}
      \intertext{\item
        If $\delta^\Y(x)(\alpha) = 1$ and $h(\alpha) \in \Sigma \times Y$,
        then $\alpha \leq b$ and we can derive
      }
        \alpha \cdot \floor{\delta^{\Y}(x)(\alpha)}_{s^\Y} \cdot \ell
        &\equiv \alpha \cdot \ell
            \tag{def. $\floor{-}$} \\
        &\equiv \alpha \cdot \floor{h(\alpha)}_{s^\Y} \cdot \ell
            \tag{$\alpha \leq b$} \\ 
        &= \alpha \cdot \floor{h(\alpha)}_{s^\X}
            \tag{def. $s^\X$} \\
        &= \alpha \cdot \floor{\delta^\X(x)(\alpha)}_{s^\X}
            \tag{def. $\X$}
      \intertext{\item
        If $\delta^\Y(x)(\alpha) = 1$ and $h(\alpha) \not\in \Sigma \times Y$,
        then $\alpha \leq \bneg{b}$ and we can derive
      }
        \alpha \cdot \floor{\delta^{\Y}(x)(\alpha)}_{s^\Y} \cdot \ell
        &\equiv \alpha \cdot \ell
            \tag{def. $\floor{-}$} \\
        &\equiv \alpha \cdot \floor{h(\alpha)}_{s^\Z}
            \tag{$\alpha \leq \bneg{b}$} \\ 
        &= \alpha \cdot \floor{h(\alpha)}_{s^\X}
            \tag{def. $s^\X$} \\
        &= \alpha \cdot \floor{\delta^\X(x)(\alpha)}_{s^\X}
            \tag{def. $\X$}
      \end{align*}
    \end{itemize}
  \end{itemize}
\end{enumerate}
This completes the proof.
\end{proof}

\begin{lemma}%
\label{lem:epsilon-vs-iota}
Let $e \in \Exp$ and $\alpha \in \At$.
Then $\iota_e(\alpha) = 1$ if and only if $\alpha \leq E(e)$.
\end{lemma}
\begin{proof}
We proceed by induction on $e$.
In the base, there are two cases.
\begin{itemize}
    \item
    If $e = b \in \Bexp$, then $\iota_e(\alpha) = 1$ if and only if $\alpha \leq b = E(b)$.

    \item
    If $e = p \in \Sigma$, then $\iota_e(\alpha) = 0$ and $E(e) = 0$.
\end{itemize}

For the inductive step, there are three cases.
\begin{itemize}
    \item
    If $e = f +_b g$, then suppose $\alpha \leq b$.
    In that case, $\iota_e(\alpha) = 1$ holds if and only if $\iota_f(\alpha) = 1$, which by induction is true precisely when $\alpha \leq E(f)$, which is equivalent to $\alpha \leq E(f +_b g)$.
    The other case can be treated analogously.

    \item
    If $e = f \cdot g$, then $\iota_e(\alpha) = 1$ implies that $\iota_f(\alpha) = 1$ and $\iota_g(\alpha) = 1$, which means that $\alpha \leq E(f)$ and $\alpha \leq E(g)$ by induction, and hence $\alpha \leq E(e)$.
    The other implication can be derived in a similar fashion.

    \item
    If $e = f^{(b)}$, then $\iota_e(\alpha) = 1$ is equivalent to $\alpha \leq \overline{b} = E(e)$.
    \qedhere
\end{itemize}
\end{proof}

\thompsonroundtrip*
\begin{proof}
We proceed by induction on $e$, showing that we can construct a solution $s_e$ to $\X_e$.
For the main claim, if we then show that $e \equiv \dsum_{\alpha \leq 1} \floor{\iota_e(\alpha)}_{s_e}$, it follows that we can extend $s_e$ to a solution $s$ of $\X_e^\iota$, by setting $s(\iota) = e$ and $s(x) = s_e(x)$ for $x \in X_e$.
In the base, there are two cases.
\begin{itemize}
    \item
    If $e = b \in \Bexp$, then we choose for $s_e$ the (empty) map from $X_e$ to $\Exp$; this (vacuously) makes $s_e$ a solution to $\X_e$.
    For the second part, we can derive using \Cref{lemma:dsum-idemp,lemma:dsum-guarded}:
    \[
        b
            \equiv \dsum_{\alpha \leq 1} b
            \equiv \dsum_{\alpha \leq 1} \alpha b
            \equiv \dsum_{\alpha \leq 1} \alpha \cdot [\alpha \leq b]
            \equiv \dsum_{\alpha \leq 1} [\alpha \leq b]
            \equiv \dsum_{\alpha \leq 1} \floor{\iota_b(\alpha)}_{s_e}
    \]

    \item
    If $e = p \in \Sigma$, then we choose $s_e(*) = 1$.
    To see that $s_e$ is a solution to $\X_e$, note by \Cref{lemma:dsum-idemp}:
    \[
        s_e(*)
            = 1
            \equiv \dsum_{\alpha \leq 1} 1
            \equiv \dsum_{\alpha \leq 1} \floor{\delta_p(*)(\alpha)}_{s_e}
    \]
    For the second part, derive as follows, using the same Lemma:
    \[
        e
            = p
            \equiv \dsum_{\alpha \leq 1} p
            \equiv \dsum_{\alpha \leq 1} p \cdot s_e(*)
            \equiv \dsum_{\alpha \leq 1} \floor{\iota_p(\alpha)}_{s_e}
    \]
\end{itemize}

\noindent
For the inductive step, there are three cases.
\begin{itemize}
    \item
    If $e = f +_b g$, then by induction we have solutions $s_f$ and $s_g$ to $\X_f$ and $\X_g$ respectively.
    We now choose $s_e$ as follows:
    \[
        s_e(x) =
            \begin{cases}
            s_f(x) & x \in X_f \\
            s_g(x) & x \in X_g
            \end{cases}
    \]
    To see that $s_e$ is a solution, we use \Cref{lem:solution-alt}.
    Suppose $x \in \X_f$; we derive for $\alpha \in \At$ that
    \begin{align*}
        \alpha \cdot \floor{\delta_e(x)(\alpha)}_{s_e}
            &\equiv \alpha \cdot \floor{\delta_f(x)(\alpha)}_{s_e} \tag{def. $\delta_e$} \\
            &\equiv \alpha \cdot \floor{\delta_f(x)(\alpha)}_{s_f} \tag{def. $s_e$} \\
            &\equiv \alpha \cdot s_f(x) \tag{induction} \\
            &\equiv \alpha \cdot s_e(x) \tag{def. $s_e$}
    \end{align*}
    The case where $x \in \X_g$ is similar.
    For the second part of the claim, we derive
    \begin{align*}
    e
        &= f +_b g \\
        &\equiv \Bigl( \dsum_{\alpha \leq 1} \floor{\iota_f(\alpha)}_{s_f} \Bigr) +_b \Bigl( \dsum_{\alpha \leq 1} \floor{\iota_g(\alpha)}_{s_g} \Bigr)
            \tag{induction} \\
        &\equiv \Bigl( b \cdot \dsum_{\alpha \leq 1} \floor{\iota_f(\alpha)}_{s_f} \Bigr) +_b \Bigl( \bneg{b} \cdot \dsum_{\alpha \leq 1} \floor{\iota_g(\alpha)}_{s_g} \Bigr)
            \tag{\nameref{ax:guard-if}, \nameref{fact:guard-if-dual}} \\
        &\equiv \Bigl( \dsum_{\alpha \leq b} \floor{\iota_f(\alpha)}_{s_f} \Bigr) +_b \Bigl( \dsum_{\alpha \leq \bneg{b}} \floor{\iota_g(\alpha)}_{s_g} \Bigr)
            \tag{\Cref{lemma:dsum-push}} \\
        &\equiv \Bigl( \dsum_{\alpha \leq b} \floor{\iota_e(\alpha)}_{s_e} \Bigr) +_b \Bigl( \dsum_{\alpha \leq \bneg{b}} \floor{\iota_e(\alpha)}_{s_e} \Bigr)
            \tag{def. $\iota_e$} \\
        &\equiv \Bigl( b \cdot \dsum_{\alpha \leq 1} \floor{\iota_e(\alpha)}_{s_e} \Bigr) +_b \Bigl( \bneg{b} \cdot \dsum_{\alpha \leq 1} \floor{\iota_e(\alpha)}_{s_e} \Bigr)
            \tag{\Cref{lemma:dsum-push}} \\
        &\equiv \Bigl( \dsum_{\alpha \leq 1} \floor{\iota_e(\alpha)}_{s_e} \Bigr) +_b \Bigl( \dsum_{\alpha \leq 1} \floor{\iota_e(\alpha)}_{s_e} \Bigr)
            \tag{\nameref{ax:guard-if}, \nameref{fact:guard-if-dual}} \\
        &\equiv \Bigl( \dsum_{\alpha \leq 1} \floor{\iota_e(\alpha)}_{s_e} \Bigr)
            \tag{\nameref{ax:idemp}}
    \end{align*}
    The case where $\alpha \leq \overline{b}$ follows similarly.

    \item
    If $e = f \cdot g$, then by induction we have solutions $s_f$ and $s_g$ to $\X_f$ and $\X_g$ respectively.
    We now choose $s_e$ as follows:
    \[
        s_e(x) =
            \begin{cases}
            s_f(x) \cdot g & x \in X_f \\
            s_g(x) & x \in X_g
            \end{cases}
    \]
    To see that $s_e$ is a solution to $\X_e$, we use \Cref{lem:solution-alt}; there are three cases to consider.
    \begin{itemize}
        \item
        If $x \in X_f$ and $\delta_f(x)(\alpha) = 1$, then we can derive
        \begin{align*}
        \alpha \cdot \floor{\delta_e(x)(\alpha)}_{s_e}
            &\equiv \alpha \cdot \floor{\iota_g(\alpha)}_{s_e}
                \tag{def. $\delta_e$} \\
            &\equiv \alpha \cdot \floor{\iota_g(\alpha)}_{s_g}
                \tag{def. $s_e$} \\
            &\equiv \alpha \cdot g
                \tag{induction} \\
            &\equiv \alpha \cdot \floor{\delta_f(x)(\alpha)}_{s_f} \cdot g
                \tag{premise} \\
            &\equiv \alpha \cdot s_f(x) \cdot g
                \tag{induction} \\
            &\equiv \alpha \cdot s_e(x)
                \tag{def. $s_e$}
        \intertext{%
            \item
            If $x \in X_f$ and $\delta_f(x)(\alpha) \neq 1$, then we can derive
        }
        \alpha \cdot \floor{\delta_e(x)(\alpha)}_{s_e}
            &\equiv \alpha \cdot \floor{\delta_f(x)(\alpha)}_{s_e}
                \tag{def. $\delta_e$} \\
            &\equiv \alpha \cdot \floor{\delta_f(x)(\alpha)}_{s_f} \cdot g
                \tag{premise} \\
            &\equiv \alpha \cdot s_f(x) \cdot g
                \tag{induction} \\
            &\equiv \alpha \cdot s_e(x)
                \tag{def. $s_e$}
        \intertext{
            \item
            If $x \in X_g$, then we can derive
        }
        \alpha \cdot \floor{\delta_e(x)(\alpha)}_{s_e}
            &\equiv \alpha \cdot \floor{\delta_g(x)(\alpha)}_{s_e}
                \tag{def. $\delta_e$} \\
            &\equiv \alpha \cdot \floor{\delta_g(x)(\alpha)}_{s_g}
                \tag{def. $s_e$} \\
            &\equiv \alpha \cdot s_g(x)
                \tag{induction} \\
            &\equiv \alpha \cdot s_e(x)
                \tag{def. $s_e$}
        \end{align*}
    \end{itemize}

    For the second claim, suppose $\iota_f(\alpha) = 1$; we then derive
    \begin{align*}
    \alpha \cdot f \cdot g
        &\equiv \alpha \cdot \floor{\iota_f(\alpha)}_{s_f} \cdot g
            \tag{induction} \\
        &\equiv \alpha \cdot g
            \tag{premise} \\
        &\equiv \alpha \cdot \floor{\iota_g(\alpha)}_{s_g}
            \tag{induction} \\
        &\equiv \alpha \cdot \floor{\iota_e(\alpha)}_{s_e}
            \tag{def. $\iota_e$} \\
    \intertext{%
        Otherwise, if $\iota_f(\alpha) \neq 1$, then we derive
    }
    \alpha \cdot f \cdot g
        &\equiv \alpha \cdot \floor{\iota_f(\alpha)}_{s_f} \cdot g
            \tag{induction} \\
        &\equiv \alpha \cdot \floor{\iota_f(\alpha)}_{s_e}
            \tag{def. $s_e$} \\
        &\equiv \alpha \cdot \floor{\iota_e(\alpha)}_{s_e}
            \tag{def. $\iota_e$}
    \end{align*}
    From the above and \Cref{lemma:dsum-idemp} we can conclude that $e = f \cdot g \equiv \dsum_{\alpha \leq 1} \floor{\iota_e(\alpha)}_{s_e}$.

    \item
    If $e = f^{(b)}$, then by induction we have a solution $s_f$ to $\X_f$.
    We now choose $s_e$ by setting $s_e(x) = s_f(x) \cdot e$.
    To see that $s_e$ is a solution to $\X_e$, we use \Cref{lem:solution-alt}; there are two cases:
    \begin{itemize}
        \item
        If $\delta_f(x)(\alpha) = 1$, then we can derive
        \begin{align*}
        \alpha \cdot \floor{\delta_e(x)(\alpha)}_{s_e}
            &\equiv \alpha \cdot \floor{\iota_e(\alpha)}_{s_e}
                \tag{def. $\delta_e$} \\
            &\equiv \alpha \cdot e
                \tag{induction} \\
            &\equiv \alpha \cdot \floor{\delta_f(x)(\alpha)}_{s_f} \cdot e
                \tag{premise} \\
            &\equiv \alpha \cdot s_f(x) \cdot e
                \tag{induction} \\
            &\equiv \alpha \cdot s_e(x)
                \tag{def. $s_e$}
        \intertext{%
            \item
            Otherwise, if $\delta_f(x)(\alpha) \neq 1$, then we can derive
        }
        \alpha \cdot \floor{\delta_e(x)(\alpha)}_{s_e}
            &\equiv \alpha \cdot \floor{\delta_f(x)(\alpha)}_{s_e}
                \tag{def. $\delta_e$} \\
            &\equiv \alpha \cdot \floor{\delta_f(x)(\alpha)}_{s_f} \cdot e
                \tag{premise} \\
            &\equiv \alpha \cdot s_f(x) \cdot e
                \tag{induction} \\
            &\equiv \alpha \cdot s_e(x)
                \tag{def. $s_e$}
        \end{align*}
    \end{itemize}

    For the second part of the claim, we consider three cases:
    \begin{itemize}
        \item
        If $\alpha \leq b$ and $\iota_f(\alpha) = 1$, then derive
        \begin{align*}
        \alpha \cdot e
            &\equiv \alpha \cdot {(1 +_{E(f)} f)}^{(b)}
                \tag{\Cref{thm:ft}} \\
            &\equiv \alpha \cdot {(\bneg{E(f)} \cdot f)}^{(b)}
                \tag{\nameref{ax:skewcomm}, \nameref{ax:tighten}} \\
            &\equiv \alpha \cdot (\bneg{E(f)} \cdot f \cdot {(\bneg{E(f)} \cdot f)}^{(b)} +_b 1)
                \tag{\nameref{ax:unroll}} \\
            &\equiv \alpha \cdot \bneg{E(f)} \cdot f \cdot e
                \tag{$\alpha \leq b$, \nameref{fact:branchsel}} \\ 
            &\equiv 0
                \tag{\Cref{lem:epsilon-vs-iota}} \\
            &\equiv \alpha \cdot \floor{\iota_e(\alpha)}_{s_e}
                \tag{def. $\iota_e$}
        \intertext{%
            \item
            If $\alpha \leq b$ and $\iota_f(\alpha) \neq 1$, then we derive
        }
        \alpha \cdot e
            &\equiv \alpha \cdot (ff^{(b)} +_b 1)
                \tag{\nameref{ax:unroll}} \\
            &\equiv \alpha \cdot ff^{(b)}
                \tag{$\alpha \leq b$, \nameref{fact:branchsel}} \\ 
            &\equiv \alpha \cdot \floor{\iota_f(\alpha)}_{s_f} \cdot e
                \tag{induction} \\
            &\equiv \alpha \cdot \floor{\iota_f(\alpha)}_{s_e}
                \tag{premise} \\
            &\equiv \alpha \cdot \floor{\iota_e(\alpha)}_{s_e}
                \tag{def. $\iota_e$}
        \intertext{%
            \item
            Otherwise, if $\alpha \leq \overline{b}$, then we derive
        }
        \alpha \cdot e
            &\equiv \alpha \cdot (ff^{(b)} +_b 1)
                \tag{\nameref{ax:unroll}} \\
            &\equiv \alpha
                \tag{$\alpha \leq \overline{b}$, \nameref{fact:branchsel}} \\ 
            &\equiv \alpha \cdot \floor{\iota_e(\alpha)}_{s_e}
                \tag{def. $\iota_e$}
        \end{align*}
        The claim then follows by \Cref{lemma:dsum-idemp}.
        \qedhere
    \end{itemize}
\end{itemize}
\end{proof}

\finalfornormal*

\begin{proof}
We need to establish  the following claims:
\begin{enumerate}[(1)]
 \item the language $\lang\X(s)$ is deterministic for all states $s \in X$;
 \item the map $\lang\X$ is a homomorphism $\X \to \L$; and
 \item the map $\lang\X$ is the unique homomorphism $\X \to \L$.
\end{enumerate}
Before we turn to proving these claims, let $L \subseteq \Gs$ be a language
and define
\[
L_{\alpha p} \defeq \set{x \in \Gs}{\alpha p x \in L}.
\]
We will need the following implication:
\begin{equation}
\label{eq:normal-thm-aux}
 \delta^\X(s)(\alpha) = (p,t)
 \quad \implies \quad
 {\lang\X(s)}_{\alpha p} = \lang\X(t).
\end{equation}
To see that it holds, we observe that given the premise, we have \[
 w\in {\lang\X(s)}_{\alpha p}
 \iff \alpha p w \in \lang\X(s)
 \iff w \in \lang\X(t).
\]
We can now show the main claims:
\begin{enumerate}[(1)]
\item
We begin by showing that $\lang\X(s)$ is deterministic for $s \in X$.
Recall that a language $L$ is deterministic if,
whenever $x,y$ are in the language
and $x$ and $y$ agree on their first $n$ atoms, then they agree
on their first $n$ actions (or lack thereof). More precisely, we need to show that
\[
\left.\begin{array}{l}
x = \alpha_1p_1\alpha_2p_2 \cdots \alpha_n p_n x' \in \lang\X(s)\\
y = \alpha_1q_1\alpha_2q_2 \cdots \alpha_n q_n y' \in \lang\X(s)
\end{array}\right\}\implies p_i=q_i \quad (\forall 1\leq i\leq n), 
\]
where the final actions  may be absent (\ie, $p_n = x' = \epsilon$ or $q_n = y' = \epsilon$).
We proceed by induction on $n$.
The case $n=0$ is trivially true.
For $n\geq 1$, take $x$ and $y$ as above. We proceed by case distinction:
\begin{itemize}
\item
  If $p_1$ is absent, \ie, $n=1$ and $p_1 = x' = \epsilon$,
  then by \Cref{def:accept} we must
  have $\delta^\X(s)(\alpha_1) = 1$ and thus cannot have
  $q_1 \in \Sigma$; hence $q_1$ is also absent, as required.
\item
  Otherwise $p_1 \in \Sigma$ is a proper action. Then by \Cref{def:accept},
  there exist $t,t'\in X$ such that:
  \[
  \begin{array}{ll}
   \delta^\X(s)(\alpha_1) = (p_1,t) \land \alpha_2p_2 \cdots \alpha_n p_n x' \in \lang\X(t)\\
    \delta^\X(s)(\alpha_1) = (q_1,t') \land \alpha_2q_2 \cdots \alpha_n q_n y' \in \lang\X(t')
    \end{array}
  \]
  This implies $(p_1,t)=(q_1,t')$ and hence
  \[
  \begin{array}{ll}
   p_1=q_1 \land \alpha_2p_2 \cdots \alpha_n p_n x' \in \lang\X(t) \land \alpha_2q_2 \cdots \alpha_n q_n y' \in \lang\X(t)
    \end{array}
  \]
  Using the induction hypothesis we can now also conclude that $p_2=q_2$, \ldots $p_n=q_n$.
\end{itemize}

\item
 Next, we show that $\lang\X$ is a homomorphism: $(G\,\lang\X) \circ \delta^\X = \delta^\LL \circ \lang\X$.

 If $\delta^\X(x)(\alpha) = 1$, then $\alpha \in \lang\X(x)$ and hence
 $\delta^\LL(\lang\X(x))(\alpha) = 1$ by definition of $\delta^\LL$.

 If $\delta^\X(x)(\alpha) = 0$, then $\alpha \not\in \lang\X(x)$ and
 for all $p \in \Sigma, w \in \Gs$, $\alpha p w \not\in \lang\X(x)$ and
 hence ${\lang\X(x)}_{\alpha p} = \emptyset$.
 Thus $\delta^\LL(\lang\X(x))(\alpha) = 0$ by definition of $\delta^\LL$.

 If $\delta^\X(x)(\alpha) = \angl{p,y}$, then $y$ is live by normality and thus
 there exists a word $w_y \in \lang\X(y)$. Thus,
 \begin{align*}
   &\phantom{{}\implies{}}~ \alpha p w_y \in \lang\X(x)
     \tag{def. $\lang\X$}\\
   &\implies~ w_y \in {\lang\X(x)}_{\alpha p}
     \tag{def. $L_{\alpha{} p}$}\\
   &\implies~ \delta^\LL(\lang\X(x))(\alpha) = \angl{p, {\lang\X(x)}_{\alpha p}}
     \tag{def. $\delta^\LL$}\\
   &\implies~ \delta^\LL(\lang\X(x))(\alpha) = \angl{p, \lang\X(y)}
     \tag{\Cref{eq:normal-thm-aux}}
 \end{align*}

 \item
 For uniqueness, let $L$ denote an arbitrary homomorphism $\X \to \L$. We
 will show that \[
   w \in L(x) \ \iff \ w \in \lang\X(x)
 \] by induction on $|w|$.

 For $w = \alpha$,
 \begin{align*}
   \alpha \in L(x)
   &\iff \delta^\LL(L(x)) = 1
     \tag{def. $\delta$}\\
   &\iff \delta^\X(x)(\alpha) = 1
     \tag{$L$ is hom.}\\ 
   &\iff \alpha \in \lang\X(x)
     \tag{def. $\lang\X$}
 \end{align*}

 For $w = \alpha p v$,
 \begin{align*}
   &\phantom{{}\iff{}} \alpha p v \in L(x)\\
   &\iff \delta^\LL(L(x))(\alpha) = \angl{p, {L(x)}_{\alpha p}} \land v \in {L(x)}_{\alpha p}
     \tag{def. $\delta^\LL$, $L_{\alpha{} p}$}\\
   &\iff \exists y.\; \delta^\X(x)(\alpha) = \angl{p, y} \land v \in L(y)
     \tag{$L$ is hom., \Cref{eq:normal-thm-aux}}\\ 
   &\iff \exists y.\; \delta^\X(x)(\alpha) = \angl{p, y} \land v \in \lang\X(y)
     \tag{induction}\\
   &\iff \alpha p v \in \lang\X(x)
     \tag{def. $\lang\X$}
 \end{align*}
 \end{enumerate}
 This concludes the proof.
 \end{proof}

\uniquesalomaa*
\begin{proof}
We recast this system as a matrix-vector equation of the form $x = Mx + D$ in the Kleene algebra with Tests of $n$-by-$n$ matrices over $\pGS$; solutions to $x$ in this equation are in one-to-one correspondence with functions $R$ as above.

We now argue that the solution is unique when the system is Salomaa. We do this by showing that the map $\sigma(x) = Mx + D$ is contractive in a certain metric on ${(\pGS)}^n$, therefore has a unique fixpoint by the Banach fixpoint theorem.

For a finite guarded string $x\in\GS$, let $\len x$ denote the number of action symbols in $x$. For example, $\len\alpha=0$ and $\len{\alpha p\beta}=1$. For $A,B\subs\GS$, define
\begin{align*}
\len A &= \begin{cases}
\min\set{\len x}{x\in A} & A\ne\emptyset\\
\infty & A=\emptyset
\end{cases} &
d(A,B) &= 2^{-\len{A\symdiff B}}
\end{align*}
where $2^{-\infty}=0$ by convention. One can show that $d(-,-)$ is a metric; in fact, it is an ultrametric, as $d(A,C) \le \max d(A,B), d(B,C)$, a consequence of the inclusion $A\symdiff C\subs A\symdiff B\cup B\symdiff C$. Intuitively, two sets $A$ and $B$ are close if they agree on short guarded strings; in other words, the shortest guarded string in their symmetric difference is long. Moreover, the space is complete, as any Cauchy sequence $A_n$ converges to the limit
\begin{align*}
\bigcup_m\bigcap_{n>m} A_n &= \set{x\in\GS}{\text{$x\in A_n$ for all but finitely many $n$}}.
\end{align*}

For $n$-tuples of sets $\seq A1n$ and $\seq B1n$, define
\begin{align*}
d(\seq A1n,\seq B1n) &= \max_{i=1}^n d(A_i,B_i).
\end{align*}
This also gives a complete metric space ${(\powerset\GS)}^n$.

For $A,B,C\subs\GS$, from \Cref{lem:metric}(i) and the fact $\len{A\diamond B} \ge \len A+\len B$, we have 
\begin{align*}
\len{(A\diamond B)\symdiff(A\diamond C)} &\ge \len{A\diamond (B\symdiff C)} \ge \len A + \len{B\symdiff C},
\end{align*}
from which it follows that
\begin{align*}
d(A\diamond B,A\diamond C) &\le 2^{-\len A}d(B,C).
\end{align*}
In particular, if $D_\alpha(e)\ne 1$ for all $\alpha$, it is easily shown by induction on $e$ that $\len x\ge 1$ for all $x\in\sem e$, thus $\len{\sem e} \ge 1$, and
\begin{align}
d(\sem{e}\diamond B,\sem{e}\diamond C) &\le 2^{-\len{\sem{e}}}d(B,C) \le \textstyle\frac 12d(B,C).\label{eq:dseq}
\end{align}
From \Cref{lem:metric}(ii) and the fact $\len{A\cup B} = \min\len A,\len B$, we have 
\begin{align*}
\len{(bA_1\cup\bneg bA_2)\symdiff(bB_1\cup\bneg bB_2)}
&= \len{(bA_1\symdiff bB_1) \cup (\bneg bA_2\symdiff\bneg bB_2)}\\
&= \min\len{bA_1\symdiff bB_1},\len{\bneg bA_2\symdiff\bneg bB_2},
\end{align*}
from which it follows that
\begin{align*}
d(bA_1\cup\bneg bA_2,bB_1\cup\bneg bB_2)
&= \max d(bA_1,bB_1),d(\bneg bA_2,\bneg bB_2).
\end{align*}
Extrapolating to any guarded sum by induction,
\begin{align}
d(\bigcup_\alpha \alpha A_\alpha,\bigcup_\alpha \alpha B_\alpha) &= \max_\alpha d(\alpha A_\alpha,\alpha B_\alpha).\label{eq:dunion}
\end{align}
Putting everything together,
\begin{align*}
d(\sigma(A),\sigma(B))
&= \max_i d(\bigcup_j\sem{e_{ij}}\diamond A_j\cup\sem{d_i},\bigcup_j\sem{e_{ij}}\diamond B_j\cup\sem{d_i})\\
&= \max_i (\max (\max_j d(\sem{e_{ij}}\diamond A_j,\sem{e_{ij}}\diamond B_j)),d(\sem{d_i},\sem{d_i})) & \text{by~\eqref{eq:dunion}}\\
&= \max_i \max_j d(\sem{e_{ij}}\diamond A_j,\sem{e_{ij}}\diamond B_j)\\
&\le \textstyle\frac 12\max_j d(A_j,B_j) & \text{by~\eqref{eq:dseq}}\\
&= \textstyle\frac 12 d(A,B).
\end{align*}
Thus the map $\sigma$ is contractive in the metric $d$ with constant of contraction $1/2$. By the Banach fixpoint theorem, $\sigma$ has a unique solution.
\end{proof}

\begin{lemma}Let $A\symdiff B$ denote the symmetric difference of $A$ and $B$. We have:%
\label{lem:metric}
\begin{enumerate}[{\upshape(i)}]
\item
$(A\diamond B)\symdiff(A\diamond C) \subs A\diamond (B\symdiff C)$.
\item
$(bA_1\cup\bneg bA_2)\symdiff(bB_1\cup\bneg bB_2) = (bA_1\symdiff bB_1) \cup (\bneg bA_2\symdiff\bneg bB_2)$.
\end{enumerate}
\end{lemma}
\begin{proof}
(i)
Suppose $x\in(A\diamond B)\setminus(A\diamond C)$. Then $x=y\diamond z$ with $y\in A$ and $z\in B$. But $z\not\in C$ since $x\not\in A\diamond C$, so $z\in B\setminus C$, therefore $x\in A\diamond(B\setminus C)$. Since $x$ was arbitrary, we have shown
\begin{align*}
(A\diamond B)\setminus(A\diamond C) &\subs A\diamond (B\setminus C).
\end{align*}
It follows that
\begin{align*}
(A\diamond B)\symdiff(A\diamond C)
&= (A\diamond B)\setminus(A\diamond C) \cup (A\diamond C)\setminus(A\diamond B)\\
&\subs A\diamond(B\setminus C) \cup A\diamond(C\setminus B)\\
&= A\diamond((B\setminus C) \cup (C\setminus B))\\
&= A\diamond (B\symdiff C).
\end{align*}

(ii)
Using the facts
\begin{align*}
A &= bA \cup \bneg bA &
b(A\symdiff B) &= bA\symdiff bB,
\end{align*}
we have
\begin{align*}
A\symdiff B
&= b(A\symdiff B) \cup \bneg b(A\symdiff B) = (bA\symdiff bB) \cup (\bneg bA\symdiff\bneg bB),
\end{align*}
therefore
\begin{align*}
\lefteqn{(bA_1\cup\bneg bA_2)\symdiff(bB_1\cup\bneg bB_2)}\qquad\\
&= (b(bA_1\cup\bneg bA_2)\symdiff b(bB_1\cup\bneg bB_2)) \cup (\bneg b(bA_1\cup\bneg bA_2)\symdiff\bneg b(bB_1\cup\bneg bB_2))\\
&= (bA_1\symdiff bB_1) \cup (\bneg bA_2\symdiff\bneg bB_2).
\end{align*}
\end{proof}

\section{Generalized guarded union}\label{sec:gen_choice}

In \Cref{sec:ft} we needed a more general type of guarded union:

\generalsum*

The definition above is ambiguous in that the choice of $\beta$ is not fixed. However, that does not change the meaning of the expression above, as far as $\equiv$ is concerned.

\begin{restatable}{lemma}{generalsumsound}
The operator $\dsum$ above is well-defined up-to $\equiv$.
\end{restatable}
\begin{proof}
We proceed by induction on the number of atoms in $\Phi$. In the base cases, when $\Phi=\emptyset$ or $\Phi=\{\alpha\}$, the claim holds immediately as the whole expression is equal to, respectively, $0$ and $e_\alpha$. For the inductive step, we need to show that for any $\beta,\gamma \in \Phi$:
\[
e_\beta +_\beta \Bigl(\; \dsum\limits_{\alpha \in \Phi\setminus\{\beta\}} e_\alpha \Bigr) \equiv e_\gamma +_\gamma \Bigl(\; \dsum\limits_{\alpha \in \Phi\setminus\{\gamma\}} e_\alpha \Bigr)
\]
We can derive
\begin{align*}
e_\beta +_\beta \Bigl(\; \dsum\limits_{\alpha \in \Phi\setminus\{\beta\}} e_\alpha \Bigr)
    &\equiv e_\beta +_\beta \Bigl( e_\gamma +_\gamma \Bigl( \; \dsum\limits_{\alpha \in \Phi\setminus\{\beta,\gamma\}} e_\alpha \Bigr) \Bigr)
            \tag{induction}    \\
    &\equiv (e_\beta +_\beta  e_\gamma) +_{\beta + \gamma} \Bigl( \; \dsum\limits_{\alpha \in \Phi\setminus\{\beta,\gamma\}} e_\alpha \Bigr)
                \tag{\nameref{fact:skewassoc-dual}} \\
    &\equiv (e_\gamma +_{\bar\beta}  e_\beta) +_{\beta + \gamma} \Bigl( \; \dsum\limits_{\alpha \in \Phi\setminus\{\beta,\gamma\}} e_\alpha \Bigr)
        \tag{\nameref{ax:skewcomm}} \\
    &\equiv e_\gamma +_{\bar\beta (\beta + \gamma)}  \Bigl(e_\beta +_{\beta + \gamma} \Bigl( \; \dsum\limits_{\alpha \in \Phi\setminus\{\beta,\gamma\}} e_\alpha \Bigr) \Bigr)
        \tag{\nameref{ax:skewassoc}} \\
            &\equiv e_\gamma +_{\gamma}  \Bigl(e_\beta +_{\beta} \Bigl( \; \dsum\limits_{\alpha \in \Phi\setminus\{\beta,\gamma\}} e_\alpha \Bigr) \Bigr)
        \tag{Boolean algebra} \\
    &\equiv e_\gamma +_\gamma \Bigl(\; \dsum\limits_{\alpha \in \Phi\setminus\{\gamma\}} e_\alpha \Bigr)
        \tag{induction}
\end{align*}
This completes the proof.
\end{proof}

The following properties are useful for calculations with $\dsum$.
\begin{lemma}%
\label{lemma:dsum-push}
Let $b, c \in \Bexp$ and suppose that for every $\alpha \leq b$, we have an $e_\alpha \in \Exp$.
The following then holds:
\[
    c \cdot \dsum_{\alpha \leq b} e_\alpha \equiv \dsum_{\alpha \leq bc} e_\alpha
\]
Recall from above that the predicate $\alpha\leq b$ is replacing the set $\Phi=\{\alpha \mid \alpha \leq b\}$.
\end{lemma}
\begin{proof}
We proceed by induction on the number of atoms below $b$.
In the base, where $b \equiv 0$, the claim holds vacuously.
For the inductive step, assume the claim holds for all $b'$ with strictly fewer atoms.
Let $\beta \in \At$ with $b = \beta + b'$ and $\beta \not\leq b'$.
There are two cases.
\begin{itemize}
    \item
    If $\beta \leq c$, then we derive
    \begin{align*}
    c \cdot \dsum_{\alpha \leq b} e_\alpha
        &\equiv c \cdot \Bigl (e_\beta +_\beta \Bigl( \dsum_{\alpha \leq b'} e_\alpha \Bigr) \Bigr)
            \tag{Def. $\dsum$} \\
        &\equiv c \cdot e_\beta +_\beta c \cdot \Bigl( \dsum_{\alpha \leq b'} e_\alpha \Bigr)
            \tag{\nameref{fact:guard-if-dual}} \\
        &\equiv c \cdot e_\beta +_\beta \Bigl( \dsum_{\alpha \leq b'c} e_\alpha \Bigr)
            \tag{induction} \\
        &\equiv e_\beta +_\beta \Bigl( \dsum_{\alpha \leq b'c} e_\alpha \Bigr)
            \tag{\nameref{ax:guard-if}, Boolean algebra} \\ 
        &\equiv \dsum_{\alpha \leq bc} e_\alpha
            \tag{Def. $\dsum$, Boolean algebra}
    \intertext{%
        where in the last step we use $b + c \equiv \beta + b'c$ and $\beta \not\leq b'c$.
        \item
        If $\beta \not\leq c$, then we derive
    }
    c \cdot \dsum_{\alpha \leq b} e_\alpha
        &\equiv c \cdot \Bigl (e_\beta +_\beta \Bigl( \dsum_{\alpha \leq b'} e_\alpha \Bigr) \Bigr)
            \tag{Def. $\dsum$} \\
        &\equiv c \cdot \Bigl( \dsum_{\alpha \leq b'} e_\alpha \Bigr)
            \tag{\nameref{fact:branchsel}, Boolean algebra} \\ 
        &\equiv \dsum_{\alpha \leq b'c} e_\alpha
            \tag{induction} \\
        &\equiv \dsum_{\alpha \leq bc} e_\alpha
            \tag{Boolean algebra}
    \end{align*}
    where for the last step we use $bc \equiv (b' + \beta)c = b'c$.
    \qedhere
\end{itemize}
\end{proof}

\begin{lemma}%
\label{lemma:dsum-idemp}
For all $e \in \Exp$ and $b \in \Bexp$, we have $\dsum\limits_{\alpha \leq b} e \equiv be$
\end{lemma}
\begin{proof}
The proof proceeds by induction on the number of atoms below $b$.
In the base, where $b \equiv 0$, the claim holds immediately.
Otherwise, assume the claim holds for all $b' \in \Bexp$ with strictly fewer atoms than $b$.
Let $\beta \in \At$ be such that $b = \beta \vee b'$ and $\beta \not\leq b'$.
We then calculate:
\begin{align*}
\dsum_{\alpha \leq b} e
    &\equiv e +_\beta \Bigl( \dsum_{\alpha \leq b'} e \Bigr)
        \tag{Def. $\dsum$} \\
    &\equiv e +_\beta b'e
        \tag{induction} \\
    &\equiv \beta e +_\beta \bneg{\beta}b'e
        \tag{\nameref{ax:guard-if}, \nameref{fact:guard-if-dual}} \\
    &\equiv \beta be +_\beta \bneg{\beta}be
        \tag{Boolean algebra} \\
    &\equiv be +_\beta be
        \tag{\nameref{ax:guard-if}, \nameref{fact:guard-if-dual}} \\
    &\equiv be
        \tag{\nameref{ax:idemp}}
\end{align*}
This completes the proof.
\end{proof}

\begin{lemma}%
\label{lemma:dsum-guarded}
Let $b \in \Bexp$ and suppose that for $\alpha \leq b$ we have an $e_\alpha \in \Exp$.
The following holds:
\[
    \dsum_{\alpha \leq b} e_\alpha
    \equiv
    \dsum_{\alpha \leq b} \alpha e_\alpha
\]
\end{lemma}
\begin{proof}
The proof proceeds by induction on the number of atoms below $b$.
In the base, where $b \equiv 0$, the claim holds immediately.
Otherwise, assume that the claim holds for all $b' \in \Bexp$ with strictly fewer atoms.
Let $\beta \in \At$ be such that $b = \beta \vee b'$ and $\beta \not\leq b'$.
We then calculate:
\begin{align*}
\dsum_{\alpha \leq b} e_\alpha
    &\equiv e +_\beta \Bigl( \dsum_{\alpha \leq b'} e_\alpha \Bigr)
        \tag{Def. $\dsum$} \\
    &\equiv \beta e_\beta +_\beta \Bigl( \dsum_{\alpha \leq b'} \alpha e_\alpha \Bigr)
        \tag{\nameref{ax:guard-if}} \\
    &\equiv \dsum_{\alpha \leq b} \alpha e_\alpha
        \tag{Def. $\dsum$}
\end{align*}
This completes the proof.
\end{proof}

\section{Coalgebraic Structure}%
\label{sec:coalg}

%
%

\subsection{Final coalgebra}%
\label{sec:final}

We give two alternative characterizations of the final $G$-coalgebra.

\subsubsection{Nonexpansive maps}%
\label{sec:nonexpansive}

For any $x\in A\star + A^\omega$, let $x|_n$ denote the prefix of $x$ of length $n$, or $x$ itself if the length of $x$ is less than $n$. One characterization of the final $G$-coalgebra is $(\FF,D)$, where
\begin{itemize}
\item
$\FF$ is the set of maps
$f:\At^\omega \to (\Sigma\star\times 2) + \Sigma^\omega$ that are \emph{nonexpansive} under the usual metric on $\At^\omega$ and $(\Sigma\star\times 2) + \Sigma^\omega$; that is, if $x,y\in\At^\omega$ and $x|_n=y|_n$, then $f(x)|_n=f(y)|_n$.

Nonexpansiveness is the manifestation of determinacy in the final coalgebra. It follows from nonexpansiveness that $\hd f(\alpha x)=\hd f(\alpha y)$ for all $x,y\in\At^\omega$. Here $\hd$ and $\tl$ are the usual head and tail functions on nonnull finite or infinite sequences.

\item
$D_\alpha:\FF\to 2+\Sigma\times\FF$, where
\begin{align*}
D_\alpha(f) &= \begin{cases}
(\hd f(\alpha x),\lam x{\tl f(\alpha x)}), & \text{if $f(\alpha x)\not\in 2$}\\
f(\alpha x), & \text{if $f(\alpha x)\in 2$.}
\end{cases}
\end{align*}
\end{itemize}

The unique homomorphism $N:(X,\delta)\to(\FF,D)$ is defined coinductively by
\begin{align*}
N(s) &= \lam{x:\At^\omega}{\begin{cases}
p\cdot N(t)(\tl x), & \text{if $\delta_{\hd x}(s) = (p,t)$}\\
\delta_{\hd x}(s), & \text{if $\delta_{\hd x}(s)\in 2$}.
\end{cases}}
\end{align*}
A state $s$ of this coalgebra is live if $N(s)(x)\in\Sigma\star\times\{1\}$ for some $x\in\At^\omega$.

\subsubsection{Labeled trees}%
\label{sec:labeledtrees}

Another characterization of the final $G$-coalgebra is in terms of labeled trees.
The nodes of the trees are represented by elements of $\At\star$ and the labels are elements of $\Sigma$ and $2$. A \emph{labeled tree} is a partial map $t:\At^+\pfun\Sigma+2$ such that if $t(x)\in 2$, then $t(y)\in\Sigma$ for all nonnull proper prefixes $y$ of $x$, and $t(z)$ is undefined for all proper extensions $z$ of $x$. The root $\eps$ is always unlabeled. In this characterization, the structure map is
\begin{align*}
& D_\alpha:(\At^+\pfun\Sigma)\to 2 + \Sigma\times(\At^+\pfun\Sigma)\\
& D_\alpha(t) = \begin{cases}
(t(\alpha),t@\alpha) & \text{if $\alpha\in\dom t$ and $t(\alpha)\in\Sigma$,}\\
t(\alpha), & \text{if $\alpha\in\dom t$ and $t(\alpha)\in 2$,}\\
\text{undefined}, & \text{if $\alpha\not\in\dom t$}
\end{cases}
\end{align*}
where $t@\alpha$ is the subtree rooted at $\alpha$: $(t@\alpha)(x) = t(\alpha x)$, $x\in\At^+$.

The unique homomorphism $T$ from $(X,\delta)$ to the final coalgebra is defined coinductively by
\begin{align*}
T(s)(\alpha) &= \begin{cases}
p, & \text{if $\delta_\alpha(s)=(p,t)$}\\
\delta_\alpha(s), & \text{if $\delta_\alpha(s)\in 2$}
\end{cases}
\qquad
T(s)(\alpha x) = \begin{cases}
T(t)(x), & \text{if $\delta_\alpha(s)=(p,t)$}\\
\text{undefined}, & \text{if $\delta_\alpha(s)\in 2$.}
\end{cases}& x\in\At^+.
\end{align*}

\subsection{Bisimilarity}%
\label{sec:bisim}

Let $N$ and $T$ denote the unique homomorphisms from any $G$-coalgebra to the final $G$-coalgebra as characterized in \cref{sec:nonexpansive,sec:labeledtrees}, respectively. Let $\sim$ denote bisimilarity (\Cref{def:bisimilarity}).

\begin{lemma}%
\label{lem:bisim}
The following are equivalent:
\begin{enumerate}[\upshape(i)]
\item $s\sim t$;
\item $N(s)=N(t)$;
\item $T(s)=T(t)$.
\end{enumerate}
In addition, if $X$ and $Y$ are normal, then these conditions are also equivalent to
\begin{enumerate}[\upshape(i)]
\setcounter{enumi}3
\item $\lang{}(s) = \lang{}(t)$;
\item $\Lang{}(s)=\Lang{}(t)$.
\end{enumerate}
\end{lemma}
\begin{proof}
The equivalence of (ii) and (iii) follows from the one-to-one correspondence between non-expansive maps $f:\At^\omega \to (\Sigma\star\times 2) + \Sigma^\omega$ and labeled trees $t:\At^+\pfun\Sigma+2$ and the observation that $N$ and $T$ assign corresponding values to any state. Given a labeled tree $t$, the corresponding non-expansive map $f$ assigns to $x\in\At^\omega$ the concatenation of the labels $t(x|_n)$ along the path $x$. Conversely, given $f$, the corresponding $t$ assigns to $x\in\At^n$ the $n$-th symbol of $f(y)$ for any $y\in\Ato$ such that $x\prec y$.

To show the equivalence of (i) and (ii), we show that the kernel of $N$ is the maximal bisimulation. To show that it is a bisimulation, suppose $N(s)=N(t)$. If $\delta_\alpha(s)\in 2$, then for any $x$, $N(t)(\alpha x) = N(s)(\alpha x) = \delta_\alpha(s)\in 2$, so $
\delta_\alpha(t) = N(t)(\alpha x) = N(s)(\alpha x) = \delta_\alpha(s)$.

If $\delta_\alpha(s) = (p,u)$, then for any $x$, $N(t)(\alpha x) = N(s)(\alpha x) = p\cdot N(u)(x)$. Thus it must be that $\delta_\alpha(t) = (p,v)$ for some $v$ and
\begin{align*}
p\cdot N(v)(x) = N(t)(\alpha x) = N(s)(\alpha x) = p\cdot N(u)(x).
\end{align*}
Since $x$ was arbitrary, $N(u) = N(v)$.

Now we show that any bisimulation refines the kernel of $N$; that is, if $s\equiv t$, then $N(s) = N(t)$. Let $\alpha x\in\Ato$ be arbitrary. If $\delta_\alpha(s)=\delta_\alpha(t)\in 2$, then
\begin{align*}
N(s)(\alpha x) &= \delta_{\alpha}(s) = \delta_{\alpha}(t) = N(t)(\alpha x).
\end{align*}
On the other hand, if $\delta_\alpha(s) = (p,u)$ and $\delta_\alpha(t) = (p,v)$, then $u\equiv v$ and
\begin{align*}
N(s)(\alpha x) &= p\cdot N(u)(x) &
N(t)(\alpha x) &= p\cdot N(v)(x).
\end{align*}
By the coinductive hypothesis, $N(u)=N(v)$, therefore
\begin{align*}
N(s)(\alpha x) &= p\cdot N(u)(x) = p\cdot N(v)(x) = N(t)(\alpha x).
\end{align*}
Thus in all cases, $N(s)(\alpha x) = N(t)(\alpha x)$. As $\alpha x$ was arbitrary, $N(s) = N(t)$.

For (iv) and (v), see~\cite[\S2.5]{KT08a}. It was shown there that if the two normal $G$-automata accept the same set of finite strings, then they are bisimilar. This is because normality implies that $\lang{}(s)$ is dense in $\Lang{}(s)$. The result of~\cite[\S2.5]{KT08a} was proved under the assumption of no failures, but this assumption turns out not to be needed. Normality implies that any dead state $s$ must immediately fail under all inputs, that is, $\delta_\alpha(s)=0$ for all $\alpha$, thus any two such states are bisimilar.
\end{proof}

\subsection{Determinacy and closure}%
\label{sec:detclosure}

We have discussed the importance of the \emph{determinacy property} (\Cref{def:action-det}) of languages represented by GKAT expressions and $G$-automata.

In addition to the determinacy property, the languages $\Lang{}(s)$ satisfy a certain topological closure property. Let $\At^\omega$ have its Cantor space topology generated by basic open sets $\set{y\in\Ato}{x\preceq y}$ for $x\in\At\star$, where $\preceq$ is the prefix relation. This is the same as the metric topology generated by the metric $d(x,y) = 2^{-n}$, where $n$ is the length of the longest prefix on which $x$ and $y$ agree, or $0$ if $x=y$. The space is compact and metrically complete (all Cauchy sequences converge to a limit).

For $x\in\GS\cup\GSo$, let $\atoms(x)\in\At^+\cup\Ato$ be the sequence of atoms in $x$ and let $\act(x)\in\Sigma\star\cup\Sigma^\omega$ be the sequence of actions in $x$. E.g., $\atoms(\alpha p\beta q\gamma) = \alpha\beta\gamma$ and $\act(\alpha p\beta q\gamma) = pq$. For $A\subs\GS\cup\GSo$, let
\begin{align*}
\up A &= \bigcup_{x\in A} \set{y\in\Ato}{\atoms(x)\preceq y}.
\end{align*}

\noindent The property of $\Lang{}(s)$ of interest is

\paragraph{Closure}
A set $A\subs\GS\cup\GSo$ satisfies the \emph{closure property} if $\up A$ is topologically closed in $\Ato$.

\medskip

The set $\up{\Lang{}(s)}$ is the set of infinite sequences of atoms leading to acceptance, starting from state $s$. This is a closed set, as the limit of any Cauchy sequence of atoms leading to acceptance---whether after a finite or infinite time---also leads to acceptance.

The set $\up{\lang{}(s)}$ is the set of infinite sequences of atoms leading to acceptance after a \emph{finite} time, starting from state $s$. This is an open set, as it is the union of basic open sets $\up{\{x\}}$ for $x\in\lang{}(s)$.

The set $\up{\Lang{}(s)}$ is the disjoint union of strings in $\Ato$ leading to acceptance after a finite (respectively, infinite) time:
\begin{align*}
\up{\Lang{}(s)} &= \up{\lang{}(s)}\ \uplus\ \set{\atoms(x)}{x\in\Lo(s)}
\end{align*}
The two sets on the right-hand side are disjoint due to the determinacy property. The complement of this set is $\Ato\setminus\up{\Lang{}(s)}$, the set of strings leading to rejection after a finite time starting from $s$.
Since $\up{\Lang{}(s)}$ is closed, its complement $\Ato\setminus\up{\Lang{}(s)}$ is open, thus a union of basic open sets.
Let $B\subs\At\star$ be a collection of minimal-length finite strings of atoms such that
\begin{align*}
\Ato\setminus\up{\Lang{}(s)} &= \bigcup_{x\in B}\up{\{x\}}.
\end{align*}
Because the strings in $B$ are of minimal length, they are prefix-incomparable, thus the basic open sets $\up{\{x\}}$ for $x\in B$ are disjoint and maximal with respect to set inclusion.

\subsection{Language models}%
\label{sec:languagemodel}
The subsets of $\GS\cup\wGs$ satisfying the determinacy property and the closure property form the carrier of a $G$-coalgebra $\LL'$. The structure map is the semantic Brzozowski derivative:
\begin{align*}
\delta^{\LL'}_\alpha(A) = \begin{cases}
(p,\set{x}{\alpha px\in A}) & \text{if $\set{x}{\alpha px\in A}\ne\emptyset$}\\
1 & \text{if $\alpha\in A$}\\
0 & \text{otherwise.}
\end{cases}
\end{align*}
Exactly one of these conditions holds by determinacy.
Although this looks similar to the language model $\LL$ of \Cref{sec:bisimNormal}, they are not the same: the states of $\LL'$ contain finite and infinite strings, whereas the states of $\LL$ contain of finite strings only. The models $\LL$ and $\LL'$ are not isomorphic. We derive the precise relationship below.

\begin{lemma}
If $A\subs\GS\cup\wGs$ satisfies determinacy and closure, then so does $\set{x}{\alpha px\in A}$.
\end{lemma}
\begin{proof}
Suppose $A$ satisfies determinacy. Let $y,z\in\set{x}{\alpha px\in A}$ agree on their first $n$ atoms. Then $\alpha py,\alpha pz\in A$ and agree on their first $n+1$ atoms. Since $A$ satisfies determinacy, $\alpha py$ and $\alpha pz$ agree on their first $n+1$ actions. Then $y$ and $z$ agree on their first $n$ actions. As $y$ and $z$ were arbitrary, $\set{x}{\alpha px\in A}$ satisfies determinacy.

Now suppose $A$ also satisfies closure. Let $x_0,x_1,\ldots$ be a Cauchy sequence in $\up{\set{x}{\alpha px\in A}}$. Then $\alpha x_0,\alpha x_1,\ldots$ is a Cauchy sequence in $\up A$. Since $\up A$ is closed, the sequence has a limit $\alpha x\in\up A$. There must exist $z\in A$ such that $\alpha x=\atoms(z)$. By determinacy, $z$ must be of the form $\alpha p y$, and $\alpha x=\atoms(z)=\alpha\atoms(y)$, thus $\atoms(y)$ is the limit of $x_0,x_1,\ldots~$.
\end{proof}

\begin{lemma}%
\label{lem:LAA}
For $A$ a state of $\LL'$, $\Lang{}(A)=A$.
\end{lemma}
\begin{proof}
We wish to show that $\accept(A,x)$ iff $x\in A$. For $\alpha\in\At$,
\begin{align*}
\accept(A,\alpha)
\Iff \delta^{\LL'}_\alpha(A) = 1
\Iff \alpha\in A.
\end{align*}
For $\alpha px$,
\begin{align*}
\accept(A,\alpha px)
&\Iff \exists B\ \delta^{\LL'}_\alpha(A) = (p,B) \wedge \accept(B,x)\\
&\Iff \lefteqn{\set{y}{\alpha py\in A}\ne\emptyset \wedge \accept(\set{y}{\alpha py\in A},x)}\\
&\Iff x\in\set{y}{\alpha py\in A} & \text{by the coinductive hypothesis}\\
&\Iff \alpha px\in A. &&\qedhere
\end{align*}
\end{proof}

The language model $\LL'$ embeds in the final $G$-coalgebra $\FF$ of \cref{sec:nonexpansive}. For $x\in\Ato$ and $z\in\Sigma\star\cup\Sigma^\omega$, let us write $x\|z$ for the unique $y\in\GS\cup\GSo$ such that $\atoms(y)\preceq x$ and $\act(y)=z$. For $A$ satisfying the determinacy and closure conditions,
\begin{align*}
A &\mapsto \lam{x:\Ato}{\begin{cases}
z, & z\in\Sigma^\omega \wedge x\|z\in A,\\
z1, & z\in\Sigma\star \wedge x\|z\in A,\\
z0, & z\in\Sigma\star,\ \text{$z$ is $\preceq$-minimal such that for no extension $z'$ of $z$ is $x\|z'\in A$.}
\end{cases}}
\end{align*}
However, $\LL'$ and $\FF$ are not isomorphic, because $\LL'$ does not distinguish early and late rejection: an automaton could take several transitions before rejecting or reject immediately, and the same set of finite and infinite strings would be accepted. Consequently, $L:(X,\delta^X)\to\LL'$ is not a coalgebra homomorphism in general. However, normality rules out this behavior. As we show in \Cref{lem:normal}, $L$ is a homomorphism if $(X,\delta^X)$ is normal. Thus $\LL'$ contains the unique homomorphic image of all normal $G$-coalgebras.

In \Cref{thm:normalfinal} we will identify a subcoalgebra $\LL''$ of $\LL'$ that is final in the category of normal $G$-coalgebras.

\begin{lemma}%
\label{lem:normal}
If $(X,\delta^X)$ is normal, then the following hold:
\begin{enumerate}[{\upshape(i)}]
\item
$\up{\Lang{}(s)}$ is the closure of $\up{\lang{}(s)}$ in $\At^\omega$;
\item
$\Lang{}:(X,\delta^X)\to\LL'$ is a coalgebra homomorphism.
\end{enumerate}
\end{lemma}
\begin{proof}
(i) Let $y\in\up{\Lang{}(s)}$. Then there exists $x\in\GS\cup\GSo$ such that either (a) $x\in\lang{}(s)$ and $\atoms(x)\prec y$, or (b) $x\in\Lo(s)$ and $\atoms(x)=y$. In the (a) case, $y\in\up{\lang{}(s)}$ and we are done. In the (b) case, by normality, all prefixes $z$ of $x$ have an extension $z'$ such that $z'\in\lang{}(s)$. The strings $\atoms(z')$ are in $\up{\lang{}(s)}$ and form a Cauchy sequence with limit $\atoms(x)=y$, thus $y$ is in the closure of $\up{\lang{}(s)}$.

(ii) We wish to show that for any $s\in X$ and $\alpha\in\At$,
\begin{align}
GL(\delta_\alpha^X(s)) &= \delta_\alpha^{\LL'}(\Lang{}(s)),\label{eq:normal}
\end{align}
where $GL(p,t) = (p,\Lang{}(t))$, $GL(1) = 1$, and $GL(0)=0$. We have
\begin{align*}
\delta_\alpha^X(s)=1 &\Imp \alpha\in \Lang{}(s) \Imp \delta_\alpha^{\LL'}(\Lang{}(s))=1,
\end{align*}
so~\eqref{eq:normal} holds if $\delta_\alpha^X(s)=1$. Similarly,
\begin{align*}
\delta_\alpha^X(s)=0 &\Imp \forall x\ \alpha x\not\in \Lang{}(s)
\Imp \alpha\not\in \Lang{}(s) \wedge \set x{\alpha px\in \Lang{}(s)} = \emptyset
\Imp \delta_\alpha^{\LL'}(\Lang{}(s)) = 0,
\end{align*}
so~\eqref{eq:normal} holds if $\delta_\alpha^X(s)=0$.

Finally, if $\delta_\alpha^X(s)=(p,t)$, then by normality $t$ is live, so $\lang{}(t)\ne\emptyset$. But if $x\in\lang{}(t)$, then $\alpha px\in\lang{}(s)\subs \Lang{}(s)$, so $\set x{\alpha px\in \Lang{}(s)}\ne\emptyset$. By definition of $\LL'$, $\delta_\alpha^{\LL'}(\Lang{}(s)) = (p,\set x{\alpha px\in \Lang{}(s)})$, and $\Lang{}(t)=\set x{\alpha px\in \Lang{}(s)}$ since $\accept(s,\alpha px)$ iff $\accept(t,x)$. Thus
\begin{align*}
GL(\delta_\alpha^X(s)) &= GL(p,t) = (p,\Lang{}(t)) = (p,\set x{\alpha px\in \Lang{}(s)}) = \delta_\alpha^{\LL'}(\Lang{}(s)).\qedhere
\end{align*}
\end{proof}

\begin{theorem}%
\label{thm:normalfinal}
Let $\LL''$ denote the subcoalgebra of $\LL'$ consisting of those sets $A\in\LL'$ such that $\up A$ is the closure of $\up{(A\cap\GS)}$; that is, such that $\up{(A\cap\GS)}$ is dense in $\up A$. Then $\LL''$ is normal and final in the category of normal $G$-coalgebras.
\end{theorem}
\begin{proof}
To show that $\LL''$ is normal, we need to show that $\sem A$ is nonempty for all nonempty $A\in\LL''$. Suppose $x\in A$. Either $x\in\GS$ itself or $\atoms(x)\in\up A$, in which case $\atoms(x)$ is the limit of strings in $\up{(A\cap\GS)}$. In either case $A\cap\GS$ is nonempty, thus
\begin{align*}
\sem A
&= \Lang{}(A)\cap\GS && \text{by definition of $\sem A$}\\
&= A\cap\GS && \text{by \Cref{lem:LAA}}\\
&\ne \emptyset.
\end{align*}
We have shown that $\LL''$ is normal. By \Cref{lem:normal}(ii), for any normal $G$-coalgebra $(X,\delta)$, $L:X\to\LL'$ is a coalgebra homomorphism, and by \Cref{lem:normal}(i), its image is in $\LL''$. By \Cref{lem:LAA}, $L$ is the identity on $\LL''$, thus $\LL''$ is final. 
\end{proof}

We conclude that $\LL''$ is isomorphic to the language model $\LL$ of \Cref{sec:bisimNormal}: the states of $\LL$ are obtained from those of $\LL''$ by intersecting with $\GS$, and the states of $\LL''$ are obtained from those of $\LL$ by taking the topological closure. This establishes that $\LL$ is isomorphic to a coequationally-defined subcoalgebra of the final $G$-coalgebra.

We remark that there is a weaker notion of normality that corresponds exactly to the language model $\LL'$.
Let us call a state $s$ of a $G$-coalgebra an \emph{explicit failure state} if all computations from $s$ lead to explicit failure after a finite time; that is, if $\Lang{}(s)=\emptyset$. By K\"onig's lemma, if $s$ is an explicit failure state, then there is a universal bound $k$ such that all computations from $s$ fail before $k$ steps. Every explicit failure state is a dead state, but the converse does not hold in general.

Let us say a $G$-coalgebra satisfies the \emph{early failure property} if there are no transitions to explicit failure states. The coalgebra $\LL'$ satisfies the early failure property and is final in the category of $G$-coalgebras satisfying early failure.

One can identify explicit failure states by depth-first search and convert to an equivalent automaton satisfying early failure by replacing all transitions to explicit failure states with immediate failure.

\section{Probabilistic Models -- continuous version}
\label{sec:prob-model-continuous}
In this subsection, we give a more general version of the probabilistic models of \Cref{sec:prob-model}, in terms
of Markov kernels, a common class of interpretations for probabilistic programming languages
(PPLs). We show that the language model is sound and complete for this class of models as well. We assume familiarity with basic measure theory.

We briefly review some basic primitives commonly used in the denotational
semantics of PPLs. For a measurable space $(X,\B)$, we
let $\Dist(X,\B)$ denote the set of subprobability measures over $X$, \ie, the set of
countably additive maps $\mu : \B \to [0,1]$ of total mass at most $1$: $\mu(X)\leq 1$. In the interest of readability, in what follows we will write $\Dist(X)$ leaving the $\B$ implicit.
A common distribution is the \emph{Dirac distribution} or \emph{point mass} on
$x \in X$, denoted $\delta_x\in\Dist(X)$; it is the map $A\mapsto [x\in A]$ assigning
probability $1$ or $0$ to a measurable set $A$ according as $A$ contains $x$.\footnote{The \emph{Iverson bracket} $[\phi]$ is defined to be $1$ if the statement
$\phi$ is true, and $0$ otherwise.}
Denotational models of PPLs typically interpret programs as
\emph{Markov kernels}, maps of type $X \to \Dist(X)$.
Such kernels can be composed in sequence using Kleisli composition,
since $\Dist(-)$ is a monad.

\begin{definition}[Probabilistic Interpretation]
Let $i = (\State, \B, \eval, \sat)$ be a triple consisting of
\begin{itemize}
  \item a measurable space $(\State,\B)$ with \emph{states} $\State$ and a $\sigma$-algebra of \emph{measurable sets} $\B\subs 2^\State$,
  \item for each action $p \in \Sigma$,
    a Markov kernel $\eval(p) \colon \State \to \Dist(\State)$, and
  \item for each primitive test $t \in T$,
    a measurable set of states $\sat(t) \in \mathsf B$.
\end{itemize}
The \emph{probabilistic interpretation} of $e \in \Exp$ with respect to $i$
is the Markov kernel $\pden{i}{e} \colon \State \to \Dist(\State)$ defined as follows:
\begin{align*}
  \pden{i}{p} &\defeq
    \eval(p)\\
  \pden{i}{b}(\sigma) &\defeq
    [\sigma \in \sat(b)] \cdot \delta_\sigma\\
  \pden{i}{e\cdot f}(\sigma)(A) &\defeq
    \int_{\sigma'} \pden{i}{e}(\sigma)(d\sigma') \cdot
    \pden{i}{f}(\sigma')(A)\quad\text{(Lebesgue integral)}\\
  \pden{i}{e +_b f}(\sigma) &\defeq
    [\sigma \in \sat(b)] \cdot \pden{i}{e}(\sigma)
    + [\sigma \in \sat(\bneg{b})] \cdot \pden{i}{f}(\sigma)\\
  \pden{i}{e^{(b)}}(\sigma) &\defeq
    \lim_{n\to\infty} \pden{i}{{(e +_b 1)}^n \cdot \bneg{b}}(\sigma)
&&\qedhere
\end{align*}
It is known from~\cite{K85a} that the limit in the definition of
$\pden{i}{e^{(b)}}$ exists, and that $\pden{i}{e}$ is a Markov kernel for
all $e$.
\end{definition}

\begin{restatable}{theorem}{soundcompleteforprobcont}%
\label{thm:sound-complete-for-prob-cont}
The language model is sound and complete for the probabilistic model in the
following sense:
\[ \den{e} = \den{f} \quad \iff \quad\forall i.\, \pden{i}{e} = \pden{i}{f}\]
\end{restatable}
\begin{proof}[Proof Sketch]
The $\Imp$ direction is essentially Lemma 1 of~\cite{K85a}.
\begin{itemize}
\item[$\Rightarrow$:]
  For soundness, we  define a map
  $\kappa_i \colon \Gs \to \State \to \Dist(\State)$
  that interprets guarded strings as Markov kernels as follows:
  \begin{align*}
    \kappa_i(\alpha)(\sigma) &\defeq
      [\sigma \in \sat(\alpha)] \cdot \delta_\sigma\\
    \kappa_i(\alpha p w)(\sigma)(A) &\defeq
      [\sigma \in \sat(\alpha)] \cdot
      \int_{\sigma'} \eval(p)(\sigma)(\sigma') \cdot
      \kappa_i(w)(\sigma')(A).
  \end{align*}
  We then lift $\kappa_i$ to languages via pointwise summation,  \[
    \kappa_i(L) \defeq \sum_{w \in L} \kappa_i(w)
  \]
  and establish that any probabilistic interpretation factors through
  the language model via $\kappa_i$:
  \begin{equation*}
    \pden{i}{-} = \kappa_i \circ \den{-}.
  \end{equation*}

\item[$\Leftarrow$:]
  For completeness, we construct an interpretation $i \defeq (\Gs, \B, \eval, \sat)$
  over the state space $\Gs$ as follows. Let $\B$ be the Borel sets of the Cantor space topology on $\GS$.
  \begin{mathpar}
    \eval(p)(w) \defeq \Unif(\set{wp\alpha}{\alpha \in \At})
    \and
    \sat(t) \defeq \set{x\alpha \in \Gs}{\alpha \leq t}
  \end{mathpar}
  and show that $\den{e}$ is fully determined by $\pden{i}{e}$:
  \begin{equation*}
    \den{e} = \set{\alpha x \in \Gs}{\pden{i}{e}(\alpha)(\{\alpha x\}) \neq 0}.
  \qedhere
  \end{equation*}
\end{itemize}
\end{proof}


\end{appendices}
}{}

\end{document}